\documentclass[a4paper,onecolumn,10pt,unpublished]{quantumarticle}

\pdfoutput=1

\usepackage[utf8]{inputenc}
\usepackage[english]{babel}
\usepackage[T1]{fontenc}

\usepackage[dvipsnames]{xcolor}
\usepackage[colorlinks=true,urlcolor=Blue,citecolor=Blue,allcolors=Blue]{hyperref}

\usepackage{amsmath,amssymb,amsfonts,amsthm}
\usepackage{mathtools,physics,qcircuit,bbm}

\usepackage{bbold}
\DeclareMathAlphabet{\mathbbold}{U}{bbold}{m}{n}
\DeclareMathAlphabet{\mathbb}{U}{msb}{m}{n}
\newcommand{\bbzero}{\mathbbold{0}}

\usepackage{float}
\usepackage{graphicx}
\graphicspath{{./Figures/}}

\usepackage[numbers,sort&compress]{natbib}

\addto\captionsenglish{
    \renewcommand{\contentsname}{\large\fontfamily{sans}Contents}
}
\makeatletter
\patchcmd{\l@section}{\addvspace{1.0em \@plus\p@}}{\addvspace{0.75em \@plus\p@}}{}{}
\makeatother

\newtheorem{Definition}{Definition}
\newtheorem{Theorem}{Theorem}
\newtheorem*{Theorem*}{Theorem}
\newtheorem{Lemma}{Lemma}
\newtheorem*{Lemma*}{Lemma}
\newtheorem{Corollary}{Corollary}

\newcommand{\Legendre}[2]{\left(\frac{#1}{#2}\right)}

\begin{document}

\title{High-threshold magic state distillation with quantum quadratic residue codes}
\author{Michael Zurel}
\orcid{0000-0003-0333-1174}
\email{mzurel@sfu.ca}
\author{Santanil Jana}
\orcid{0009-0006-5243-0653}
\author{Nadish de Silva}
\orcid{0000-0001-7904-306X}
\affiliation{Department of Mathematics, Simon Fraser University, Burnaby, Canada}

\begin{abstract}
    \vspace{-10pt}
    We present applications of quantum quadratic residue codes in magic state distillation. This includes showing that existing codes which are known to distill magic states, like the $5$-qubit perfect code, the $7$-qubit Steane code, and the $11$-qutrit and $23$-qubit Golay codes, are equivalent to certain quantum quadratic residue codes. We also present new examples of quantum quadratic residue codes that distill qubit $T$ states and qutrit Strange states with high thresholds, and we show that there are infinitely many quantum quadratic residue codes that distill $T$ states with a non-trivial threshold. All of these codes, including the codes with the highest currently known thresholds for $T$ state and Strange state distillation, are unified under the umbrella of quantum quadratic residue codes.
    \vspace{-15pt}
\end{abstract}

\tableofcontents

\section{Introduction}\label{Section:Introduction}

Magic state distillation is one of the leading candidates for scalable fault-tolerant quantum computation~\cite{BravyiKitaev2005,CampbellVuillot2017}. In this approach, the stabilizer formalism describes the operations that can be performed fault-tolerantly, namely, preparation of stabilizer states, Clifford unitary gates, and Pauli measurements. But these operations alone are not universal for quantum computation~\cite{Gottesman1997,Gottesman1998a,Gottesman1998b,AaronsonGottesman2004}. Furthermore, the Eastin-Knill theorem shows that no universal set of unitary gates can be implemented transversally in any fixed stabilizer code~\cite{EastinKnill2009}.

Quantum computational universality is recovered by injecting ``magic'' (non-stabilizer) states into otherwise stabilizer quantum circuits. The problem of fault-tolerantly implementing a universal set of computational primitives is thus shifted from that of fault-tolerantly implementing a universal set of unitary gates, to that of preparing sufficiently high-fidelity magic states. For this task, Bravyi and Kitaev showed that a supply of noisy magic states can be consumed by a so-called \emph{magic state distillation protocol} in order to return fewer purified copies of the magic state~\cite{BravyiKitaev2005}. Repeated application of these protocols allow for the production of nearly pure magic states, enabling universal quantum computation. Since then, many new magic states and distillation protocols have been discovered~\cite{CampbellBrowne2009,Reichardt2009,AnwarBrowne2012,BravyiHaah2012,CampbellBrowne2012,DawkinsHoward2015,HowardDawkins2016,HaahWecker2017,Prakash2020,PrakashSaha2025,WillsYamasaki2025,PrakashSinghal2026}.

A foundational question for this quantum computational model, with implications for quantum computation and quantum information more broadly, is: \emph{which quantum states can be distilled using only stabilizer operations to useful magic states?} We call these states \emph{universal} for quantum computation. For qubits, a long-standing conjecture is that any state that is not a convex mixture of stabilizer states is universal~\cite{Reichardt2005,Reichardt2009,VeitchEmerson2014,HowardDawkins2016,HowardCampbell2017}. The corresponding conjecture for odd-prime-dimensional qudits is that any state whose Wigner function takes negative values is universal~\cite{MariEisert2012,VeitchEmerson2012,DawkinsHoward2015,Prakash2020,PrakashGupta2020}.

Mixtures of stabilizer states for qubits and Wigner nonnegative states for odd-prime-dimensional qudits are known to not be distillable to useful magic states~\cite{Gottesman1997,MariEisert2012,VeitchEmerson2012,VeitchEmerson2014}. Thus, nonstabilizerness and Wigner negativity are necessary preconditions for universality. Both conjectures align with traditional notions of non-classicality. Stabilizer states with stabilizer operations are efficiently simulatable classically~\cite{Gottesman1997,Gottesman1998a,AaronsonGottesman2004}, and therefore are classical from a computational point of view, whereas Wigner negativity is a traditional notion of non-classicality in physics~\cite{Hudson1974,KenfackZyczkowski2004,Galvao2005,CormickPittenger2006,Gross2006,Gross2007}. Furthermore, negativity in the Wigner function implies the failure of a classical (i.e., noncontextual) hidden variable description of the computation~\cite{Spekkens2008,HowardEmerson2014,DelfosseRaussendorf2017,SchmidPusey2022,RaussendorfFeldmann2023}. The conjectures assert that the respective notions of non-classicality (non-stabilizerness for qubits, Wigner negativity for odd-prime-dimensional qudits) are not only necessary but also sufficient for quantum computational universality. Both conjectures remain open.

An obstruction to proving these conjectures is that no protocol of finite size can distill useful magic states from resource states arbitrarily close to certain regions of the boundary of the known non-universal regions. Campbell and Browne showed that, for qubits, no fixed finite-size stabilizer code can distill useful magic states from states arbitrarily close to a face of the stabilizer polytope~\cite{CampbellBrowne2010}. For any given protocol size there is always a finite gap between the best achievable noise threshold of a protocol of that size and the face of the stabilizer polytope. Therefore, a positive resolution of the qubit conjecture would require a sequence of distillation protocols of increasing size whose distillation thresholds converge to the stabilizer polytope boundary. Prakash and Gupta established an analogous limitation for odd-dimensional qudits~\cite{PrakashGupta2020}, implying that any proof of the Wigner negativity conjecture would similarly need an infinite family of protocols with thresholds approaching the relevant boundary.

Simultaneously, there is a common expectation that increasing code size tends to make distillation thresholds worse rather than better. This expectation is supported empirically by the observation that the codes with the best known noise thresholds for particular target magic states are often among the smallest codes that distill those states. For example, for qubit $T$-type magic states, the best known threshold is achieved by the $5$-qubit code~\cite{BravyiKitaev2005}. For qutrit ``Strange'' states, the best known threshold is achieved by the $11$-qutrit Golay code~\cite{Prakash2020,PrakashSinghal2026}. Together these observations suggest that, while large codes are necessary in principle to get a sequence of codes with distillation thresholds approaching the boundary of the non-universal region, identifying code families that maintain high thresholds with increasing code size is non-trivial.

This motivates the present focus on quantum quadratic residue (QR) codes and their role in defining such infinite families of magic state distillation protocols. The first contribution of this paper is to show that several known distillation protocols, such as those based on the $5$-qubit perfect code (for $T$ state distillation), the $7$-qubit Steane code (for $H$ state distillation), and the $11$-qutrit Golay code (for Strange state distillation), are equivalent to certain quantum QR codes. The second contribution is to identify new members of the quantum QR code family that distill qubit $T$ states and qutrit Strange states.

Although none of the new protocols we find surpass the noise thresholds of the $5$-qubit code or the $11$-qutrit Golay code, we do obtain fairly large codes that distill these states with high thresholds. This fact, together with the unification of the best known codes for $T$ state and Strange state distillation under a single umbrella, make quantum QR codes a natural candidate family for constructing infinite sequences of magic state distillation protocols that could resolve the above conjectures. However, these codes are notoriously hard to analyze and so, for now, computing the noise thresholds of these codes is a computational problem that can only be carried out for some relatively small examples, and we cannot conclude much about the asymptotic thresholds of these codes or any subsequence thereof. That said, one asymptotic result we do prove is that there are infinitely many quantum QR codes which distill qubit $T$ states with a non-trivial threshold.

\paragraph{Outline.} The rest of this paper is structured as follows. First, in Section~\ref{Section:Background}, we provide some background on classical and quantum coding theory and on magic state distillation, including a correspondence between quantum stabilizer codes and certain classical codes that is useful in defining and describing the codes and magic state distillation protocols of present interest. In Section~\ref{Section:QRCodes}, we introduce classical and quantum quadratic residue codes and prove some of their properties. In Section~\ref{Section:MSD}, we show that quantum quadratic residue codes distill magic states. This section includes rederivations of relations between the performance of a magic state distillation protocol and certain weight enumerators of the corresponding code~\cite{Rall2017,Prakash2020,PrakashGupta2020,KalraPrakash2025,PrakashSinghal2026}, as well as methods for computing the weight enumerators of quantum quadratic residue codes. Finally, we conclude with a discussion of these results in Section~\ref{Section:Discussion}.

\section{Background}\label{Section:Background}

Here we briefly review some background material on classical coding theory, quantum coding theory, and magic state distillation that will be necessary for stating and proving our main results.

\subsection{Classical coding theory}

An $[n,k,\delta]_q$ linear code is a $k$-dimensional subspace of the vector space $\mathbb{F}_q^n$, where $\mathbb{F}_q$ is the finite field of order $q$. Here, $n$ is the length of the code, $k$ is the dimension, $q$ is the size of the alphabet, and $\delta$ is the minimum distance of the code. When the minimum distance is not relevant or not known, the notation $[n,k]_q$ is also used.

The vector space $\mathbb{F}_q^n$ comes equipped with the standard Euclidean inner product
\begin{equation}
    \langle u,v\rangle=\sum_{k=1}^nu_kv_k\in\mathbb{F}_q.
\end{equation}
The dual code of an $[n,k]_q$ code $C$ is the $[n,n-k]_q$ code
\begin{equation}
    C^\perp:=\{v\in\mathbb{F}_q^n\;|\;\langle u,v\rangle=0\;\forall u\in C\}.
\end{equation}
A code $C$ is called self-orthogonal (or, more specifically, Euclidean self-orthogonal) if $C\subseteq C^\perp$, and it is (Euclidean) self-dual if $C=C^\perp$.

When the size of the alphabet is a square, $q=r^2$, we can define another inner product on $\mathbb{F}_q^n$. To start, the Frobenius automorphism is $a\mapsto\bar{a}=a^r$, this is the generator of the Galois group $\mathrm{Gal}(\mathbb{F}_q/\mathbb{F}_r)\cong\mathbb{Z}_2$. Then the Hermitian inner product on $\mathbb{F}_q^n$ is defined as
\begin{equation}
    \langle u,v\rangle_H=\sum\limits_{k=1}^n\bar{u}_kv_k.
\end{equation}
The Hermitian dual code $C^{\perp_H}$ of a code $C$, along with Hermitian self-orthogonality and Hermitian self-duality, can be defined similar to the Euclidean case.

The weight enumerator polynomial of a code $C$ is a useful invariant defined as
\begin{equation}\label{Equation:WeightEnumerator}
    W_C(x,y)=\sum\limits_{v\in C}x^{n-wt(v)}y^{wt(v)}=\sum\limits_{w=0}^nA_wx^{n-w}y^w,
\end{equation}
where $wt(v):=|\{k\;|\;v_k\ne0\}|$ is the Hamming weight of the vector $v$, and the coefficients $A_w:=|\{v\in C\;|\;wt(v)=w\}|$ determine the weight distribution of the code. The weight enumerator of a code $C$ is related to that of the its dual code by the MacWilliams identity~\cite{MacWilliams1962,MacWilliams1963},
\begin{equation}\label{Equation:MacWilliamsIdentity}
    W_{C^\perp}(x,y)=\frac{1}{|C|}W_C(x+(q-1)y,x-y).
\end{equation}
This identity holds for both Euclidean and Hermitian duals~\cite{MacWilliamsSloane1977}.

\paragraph{Cyclic codes.}

A cyclic code of length $n$ over $\mathbb{F}_q$ is a linear code $C\subset\mathbb{F}_q^n$ such that for any vector $c=(c_0,c_1,\dots,c_{n-1})\in C$, the code also contains the vector $(c_{n-1},c_0,\dots,c_{n-2})$ obtained by cyclically permuting the coordinates. We can identify cyclic codes with certain polynomials. First, vectors in $\mathbb{F}_q^n$ are in bijective correspondence with polynomials in $\mathbb{F}_q[x]$ of degree at most $n-1$ through the map $c=(c_0,c_1,\dots,c_{n-1})\mapsto c(x)=c_0+c_1x+\cdots+c_{n-1}x^{n-1}$. With multiplication modulo $x^n-1$, a cyclic shift of the vector's coordinates corresponds to multiplication by $x$, $xc(x)=c_{n-1}+c_0x+\cdots+c_{n-2}x^{n-1}$. Under this identification, cyclic codes are ideals of the ring $\mathbb{F}_q[x]/(x^n-1)$. The ring $\mathbb{F}_q[x]/(x^n-1)$ is a principal ideal ring, and every cyclic code (ideal of $\mathbb{F}_q[x]/(x^n-1)$), $C$, is generated by a unique monic polynomial $g(x)$ of minimum degree dividing $x^n-1$. The polynomial $g(x)$ is called the generator polynomial of the code $C$, and $\dim(C)=n-\deg(g)$~\cite{PlessBrualdi1998}.

Assume that $\gcd(n,q)=1$, let $\mathbb{F}_{q^m}$ be the splitting field of $x^n-1$ over $\mathbb{F}_q$, and let $\zeta\in\mathbb{F}_{q^m}$ be a primitive $n^{th}$ root of unity. The defining set, or set of zeros, of the code $C$ is the set
\begin{equation}
    \mathcal{Z}(C):=\left\{k\in\mathbb{Z}_n\;|\;g\left(\zeta^k\right)=0\right\}.
\end{equation}
Since, with the assumption $\gcd(n,q)=1$, $x^n-1=\prod_{k\in\mathbb{Z}_n}(x-\zeta^k)$ in $\mathbb{F}_{q^m}[x]$, the zeros uniquely determine the generator polynomial, and hence also the code. The Euclidean dual code $C^\perp$ of a cyclic code $C$ is a cyclic code with zeros~\cite{MacWilliamsSloane1977}
\begin{equation}
    \mathcal{Z}(C^\perp)=\mathbb{Z}_n\setminus(-\mathcal{Z}(C)),
\end{equation}
and when $q=r^2$, the Hermitian dual $C^{\perp_H}$ of a cyclic code $C$ is also a cyclic code with zeros~\cite[Proposition~3.4]{DicuangcoSole2007}
\begin{equation}\label{Equation:HermitianDualZeros}
    \mathcal{Z}(C^{\perp_H})=\mathbb{Z}_n\setminus(-r\mathcal{Z}(C)).
\end{equation}

An idempotent of a cyclic code $C$ is a polynomial $e(x)\in C$ such that $e(x)^2=e(x)$ in $\mathbb{F}_q[x]/(x^n-1)$. When $\gcd(n,q)=1$, each cyclic code has a unique idempotent generator, called the idempotent of the code.

\paragraph{Additive codes.}

Additive codes are more general than the linear codes described above. An additive code over a finite field $\mathbb{F}_q$ is not necessarily $\mathbb{F}_q$-linear, but it is closed under addition. In other words, an additive code $C\subset\mathbb{F}_q^n$ is a subgroup of $(\mathbb{F}_q^n,+)$. Additive codes will appear briefly in Section~\ref{Section:CodeCorrespondence}, but most of this paper will be about linear codes.

\medskip

See, for example, Refs.~\cite{MacWilliamsSloane1977,PlessBrualdi1998} for a more thorough introduction to coding theory.

\subsection{Quantum coding theory}

Our setting is quantum computation on systems of multiple qudits, where each qudit has prime Hilbert space dimension $d$. For our main results, we restrict to the cases $d=2$ and $d=3$.

The single-qudit generalized Pauli operators are defined as
\begin{equation}
    Z=\sum\limits_{j\in\mathbb{Z}_d}\omega^j\ketbra{j}{j}\quad\text{and}\quad X=\sum\limits_{j\in\mathbb{Z}_d}\ketbra{j+1}{j}
\end{equation}
where $\omega=\exp(2\pi i/d)$. The $n$-qudit Pauli operators are constructed from tensor products of the single-qudit operators. We label the $n$-qudit Pauli operators up to overall phases by the elements of the vector space $\mathbb{F}_d^{2n}$ as follows. For each $v=(v_z,v_x)\in\mathbb{F}_d^n\times\mathbb{F}_d^n$,
\begin{equation}\label{Equation:PauliOperators}
    T_v=\begin{cases}
        i^{-\langle v_z,v_x\rangle}\bigotimes\limits_{k=1}^n Z^{[v_z]_k}X^{[v_x]_k} \qquad&\text{if $d=2$},\\
        \omega^{-\langle v_z,v_x\rangle\cdot2^{-1}}\bigotimes\limits_{k=1}^n Z^{[v_z]_k}X^{[v_x]_k} \qquad& \text{if $d$ is odd},
    \end{cases}
\end{equation}
where $2^{-1}$ is the multiplicative inverse of $2$ in $\mathbb{F}_d$. The Pauli group $\mathcal{P}$ is the group generated by the $n$-qudit Pauli operators.

The symplectic form on $\mathbb{F}_d^{2n}$ defined by $[u,v]:=\langle u_z,v_x\rangle-\langle u_x,v_z\rangle$ computes the commutator of the Pauli operators through the relation
\begin{equation}
    [T_u,T_v]:=T_uT_vT_u^{-1}T_v^{-1}=\omega^{[u,v]}\mathbbm{1}.
\end{equation}
We define a function $\beta$ that tracks the phases that appear when commuting Pauli operators compose through the relation
\begin{equation}\label{Equation:BetaDefinition}
    T_uT_v=\omega^{-\beta(u,v)}T_{u+v},\qquad\forall u,v\in\mathbb{F}_d^{2n}\text{ such that }[u,v]=0.
\end{equation}
With the phase convention chosen in Eq.~\eqref{Equation:PauliOperators}, when $d$ is odd, we have the identities $T_u^{-1}=T_u^\dagger=T_{-u}$ and $T_uT_v=\omega^{[u,v]\cdot 2^{-1}}T_{u+v}$ for all $u,v\in\mathbb{F}_d^{2n}$, and so $\beta\equiv0$. These are the finite-dimensional analogues of the Weyl relations. For $d=2$, there is no corresponding choice of phase convention that reproduces these relations~\cite{RaussendorfFeldmann2023}.

The Clifford group, $\mathcal{C}\ell$, is the normalizer of the Pauli group in the unitary group, i.e., the subgroup of unitary operators that map Pauli operators to Pauli operators under conjugation. They act on Pauli operators as
\begin{equation}\label{Equation:CliffordActionOnPaulis}
    gT_ug^\dagger=\omega^{\Phi_g(u)}T_{S_gu},\qquad\forall g\in\mathcal{C}\ell
\end{equation}
where $S_g\in\mathrm{Sp}(\mathbb{F}_d^{2n})$ is a symplectic map on $\mathbb{F}_d^{2n}$. The Pauli group is a normal subgroup of the Clifford group such that $\mathcal{C}\ell/\mathcal{P}\cong\mathrm{Sp}(\mathbb{F}_d^{2n})$.

A projector onto the eigenspace of a set $I\subset\mathbb{F}_d^{2n}$ of compatible Pauli observables corresponding to eigenvalues $\{\omega^{r(v)}\;|\;v\in I\}$ is given by
\begin{equation}\label{Equation:PauliProjector}
    \Pi_I^r=\frac{1}{|I|}\sum\limits_{v\in I}\omega^{-r(v)}T_v.
\end{equation}
In order for $\Pi_I^r$ to be a valid projector, the set $I$ and the function $r:I\to\mathbb{F}_d$ must satisfy some conditions. First, to have idempotence, $(\Pi_I^r)^2=\Pi_I^r$, I must be closed under addition, i.e., $I$ is a subspace of $\mathbb{F}_d^{2n}$. Second, the operators labeled by $I$ must commute so $I$ must be an isotropic subspace, i.e., a subspace on which the symplectic form vanishes. Finally, the eigenvalues must be compatible so $r$ must be a consistent (noncontextual) value assignment for $I$. That is, $r:I\to\mathbb{F}_d$ is a function satisfying
\begin{equation}
    \omega^{-r(u)-r(v)}T_uT_v=\omega^{-r(u+v)}T_{u+v}\quad\forall u,v\in I,
\end{equation}
or equivalently, with Eq.~\eqref{Equation:BetaDefinition},
\begin{equation}\label{Equation:NoncontextualityCondition}
    r(u)+r(v)-r(u+v)=-\beta(u,v)\qquad\forall u,v\in I.
\end{equation}

An $[\![n,k,\delta]\!]_d$ quantum stabilizer code is a rank $n-k$ Abelian subgroup of the $n$-qudit Pauli group that contains no element of the centre of the Pauli group other than $\mathbbm{1}$. Such a code can be identified with a pair $(I,r)$ where $I$ is an $n-k$-dimensional isotropic subspace of $\mathbb{F}_d^{2n}$, and $r:I\to\mathbb{F}_d$ is a consistent value assignment for $I$. The code space is the subspace of the $n$-qudit Hilbert space fixed by $\Pi_I^r$. For some purposes, we can ignore the phases $r$ and identify the code with the subspace $I$. Note that since $I$ is a subspace of a vector space over a finite field, we can regard it as a classical linear code which is self-orthogonal with respect to the symplectic inner product.

The symplectic weight, or Pauli weight, of a vector $v=(v_z,v_x)\in\mathbb{F}_d^n\times\mathbb{F}_d^n$ is
\begin{equation}
    swt(v)=|\{k\in\mathbb{Z}_n\;|\;[v_z]_k\ne0\vee[v_x]_k\ne0\}|.
\end{equation}
Equivalently, it is the number of qudits on which the Pauli operator $T_v$ acts non-trivially. The distance $\delta$ of a stabilizer code is the minimum symplectic weight of the nonzero vectors in $I^\perp\setminus I$. Symplectic weight distributions and symplectic weight enumerators can be defined analogous to Eq.~\eqref{Equation:WeightEnumerator}.

\medskip

For more background on quantum error correction and fault-tolerant quantum computation see, for example, Refs.~\cite{NielsenChuang2010,Gottesman2024}.

\subsection{Correspondence between classical and quantum codes}\label{Section:CodeCorrespondence}

There is a correspondence between stabilizer codes on $d$-dimensional qudits where $d$ is a prime, and certain additive codes over $\mathbb{F}_{d^2}$~\cite{CalderbankSloane1998,AshikhimKnill2001,KetkarSarvepalli2005}. To describe this correspondence, we first need a map $\iota:\mathbb{F}_d^{2n}\to\mathbb{F}_{d^2}^n$ defined as follows. Fix a quadratic extension $\mathbb{F}_{d^2}/\mathbb{F}_d$ and choose an $\alpha\in\mathbb{F}_{d^2}\setminus\mathbb{F}_d$ so that $\{1,\alpha\}$ is an $\mathbb{F}_d$-basis for $\mathbb{F}_{d^2}$. Then for each $v=(v_z,v_x)\in\mathbb{F}_d^n\times\mathbb{F}_d^n$,
\begin{equation}
    \iota(v)=v_z+\alpha v_x\in\mathbb{F}_{d^2}^n.
\end{equation}

Let $x\mapsto\bar{x}:=x^d$ denote the Frobenius automorphism, this is the generator of the Galois group $\mathrm{Gal}(\mathbb{F}_{d^2}/\mathbb{F}_d)\cong\mathbb{Z}_2$. Then $\mathbb{F}_{d^2}^n$ comes equipped with a Hermitian inner product,
\begin{equation}\label{Equation:HermitianProduct}
    \langle x,y\rangle_H=\sum_{k=1}^n \bar{x}_k y_k\in\mathbb{F}_{d^2}.
\end{equation}

Another inner product on $\mathbb{F}_{d^2}^n$ will also be useful. First, we define the field trace map $\Tr_{\mathbb{F}_{d^2}/\mathbb{F}_d}:\mathbb{F}_{d^2}\to\mathbb{F}_d$ as\footnote{The simple form of the trace map here comes from the fact that $\mathrm{Gal}(\mathbb{F}_{d^2}/\mathbb{F}_d)$ has order $2$. In general, the field trace sums over every element of the Galois group.}
\begin{equation}
    \Tr_{\mathbb{F}_{d^2}/\mathbb{F}_d}(x)=x+\bar{x}.
\end{equation}
Then, letting $\lambda=(\alpha-\bar{\alpha})^{-1}\in\mathbb{F}_{d^2}$ (this is well-defined since $\alpha\neq\bar{\alpha}$ for $\alpha\in\mathbb{F}_{d^2}\setminus\mathbb{F}_d$), we define an $\mathbb{F}_d$-bilinear alternating form on $\mathbb{F}_{d^2}^n$ by
\begin{equation}\label{Equation:ModifiedTraceHermitian}
    \langle x,y\rangle_{MTH}:=\Tr_{\mathbb{F}_{d^2}/\mathbb{F}_d}(\lambda\langle x,y\rangle_H)\in\mathbb{F}_d.
\end{equation}
We call this the modified Trace-Hermitian inner product. If $\lambda=1$, then this reduces to the usual Trace-Hermitian inner product $\langle x,y\rangle_{TH}=\Tr_{\mathbb{F}_{d^2}/\mathbb{F}_d}(\langle x,y\rangle_H)$.

Through the map $\iota$, the standard symplectic form on $\mathbb{F}_d^{2n}$ is exactly reproduced by the modified Trace-Hermitian form on $\mathbb{F}_{d^2}^n$. More precisely, for all $u,v\in\mathbb{F}_d^{2n}$,
\begin{equation}\label{Equation:SymplecticEqualsMTH}
    [u,v]:=\langle u_z,v_x\rangle-\langle u_x,v_z\rangle=\Tr_{\mathbb{F}_{d^2}/\mathbb{F}_d}(\lambda\langle\iota(u),\iota(v)\rangle_H)=:\langle \iota(u),\iota(v)\rangle_{MTH}.
\end{equation}
To see this, write $\iota(u)=u_z+\alpha u_x$ and $\iota(v)=v_z+\alpha v_x$. Then, expanding $\langle \iota(u),\iota(v)\rangle_H$ gives
\begin{equation}\label{Equation:HermitianProductExpansion}
    \langle \iota(u),\iota(v)\rangle_H = \langle u_z,v_z\rangle+\alpha\langle u_z,v_x\rangle+\bar{\alpha}\langle u_x,v_z\rangle+\bar{\alpha}\alpha\langle u_x,v_x\rangle.
\end{equation}
Since $\Tr_{\mathbb{F}_{d^2}/\mathbb{F}_d}(\lambda)=\Tr_{\mathbb{F}_{d^2}/\mathbb{F}_d}(\lambda\bar{\alpha}\alpha)=0$, $\Tr_{\mathbb{F}_{d^2}/\mathbb{F}_d}(\lambda\alpha)=1$, and $\Tr_{\mathbb{F}_{d^2}/\mathbb{F}_d}(\lambda\bar{\alpha})=-1$, multiplying Eq.~\eqref{Equation:HermitianProductExpansion} by $\lambda$ and taking a trace yields Eq.~\eqref{Equation:SymplecticEqualsMTH}.

Therefore, an isotropic subspace $I\subseteq\mathbb{F}_d^{2n}$ (which specifies a stabilizer code up to phases) corresponds via $\iota$ to an $\mathbb{F}_d$-additive code $C=\iota(I)\subseteq \mathbb{F}_{d^2}^n$ which is self-orthogonal with respect to the inner product Eq.~\eqref{Equation:ModifiedTraceHermitian}. Conversely, if $C\subseteq\mathbb{F}_{d^2}^n$ is an $\mathbb{F}_d$-additive code satisfying $C\subseteq C^{\perp_{MTH}}$, then $I=\iota^{-1}(C)$ is an isotropic subspace of $\mathbb{F}_d^{2n}$, and hence defines a stabilizer code on $n$ qudits with $k=n-\log_d|I|$ logical qudits.

The symplectic weight of $v=(v_z,v_x)\in\mathbb{F}_d^{2n}$ is equal to the Hamming weight of $\iota(v)\in\mathbb{F}_{d^2}^n$, since $\iota(v)_k=0\in\mathbb{F}_{d^2}$ if and only if $([v_z]_k,[v_x]_k)=(0,0)\in\mathbb{F}_d^2$. Therefore, the distance of the corresponding stabilizer code is
\begin{equation}\label{Equation:QuantumCodeDistance}
    \delta=\min\left\{wt(c)\;|\;c\in C^{\perp_{MTH}}\setminus C\right\}.
\end{equation}

Hermitian self-orthogonal codes (with respect to Eq.~\eqref{Equation:HermitianProduct}) are necessarily also self-orthogonal with respect to the modified Trace-Hermitian inner product Eq.~\eqref{Equation:ModifiedTraceHermitian}. Such codes are a special case of stabilizer codes. Similarly, stabilizer codes on $d$-dimensional qudits are only required to be $\mathbb{F}_d$-linear (i.e., the image under $\iota$ is an additive code over $\mathbb{F}_{d^2}$). Codes which under the map $\iota$ are $\mathbb{F}_{d^2}$-linear codes are also a special case of stabilizer codes. The codes of interest in this paper will happen to be $\mathbb{F}_{d^2}$-linear and Hermitian self-orthogonal.

The following lemma will be useful later on.

\begin{Lemma}\label{Lemma:MTHEqualsHermitian}
    If a code $C\subseteq\mathbb{F}_{d^2}^n$ is $\mathbb{F}_{d^2}$-linear, then $C^{\perp_{MTH}}=C^{\perp_H}$.
\end{Lemma}

\begin{proof}[Proof of Lemma~\ref{Lemma:MTHEqualsHermitian}]
    If $y\in C^{\perp_H}$, then $\langle x,y\rangle_H=0$ for all $x\in C$, and hence,
    \begin{equation}
        \langle x,y\rangle_{MTH}=\Tr_{\mathbb{F}_{d^2}/\mathbb{F}_d}(\lambda\langle x,y\rangle_H)=0\qquad \forall x\in C.
    \end{equation}
    Therefore, $C^{\perp_H}\subseteq C^{\perp_{MTH}}$.

    Conversely, let $y\in C^{\perp_{MTH}}$, fix any $x\in C$, and let $b:=\langle x,y\rangle_H$. Since $C$ is $\mathbb{F}_{d^2}$-linear, we have $ax\in C$ for every $a\in\mathbb{F}_{d^2}$, and so
    \begin{equation}
        0 = \langle ax,y\rangle_{MTH} = \Tr_{\mathbb{F}_{d^2}/\mathbb{F}_d}(\lambda\langle ax,y\rangle_H) = \Tr_{\mathbb{F}_{d^2}/\mathbb{F}_d}(\lambda\bar{a}b)\qquad\forall a\in\mathbb{F}_{d^2}.
    \end{equation}
    Since $a$ ranges over all of $\mathbb{F}_{d^2}$, so does $\lambda\bar{a}$. Hence, $\Tr_{\mathbb{F}_{d^2}/\mathbb{F}_d}(zb)=0$ for all $z\in\mathbb{F}_{d^2}$. This is only possible if $b=0$, i.e., $\langle x,y\rangle_H=0$. Since $x\in C$ was arbitrary, $y\in C^{\perp_H}$. Thus $C^{\perp_{MTH}}=C^{\perp_H}$.
\end{proof}

\medskip

In this paper, we will focus on two cases in particular, namely, qubits and qutrits.

\paragraph{Qubits.} In this case, we take $\mathbb{F}_4\cong\mathbb{F}_2[\alpha]/(\alpha^2+\alpha+1)$ with Frobenius conjugation $x\mapsto\bar{x}=x^2$ and trace map $\Tr_{\mathbb{F}_4/\mathbb{F}_2}(x)=x+x^2$. For any $\alpha\in\mathbb{F}_4\setminus\mathbb{F}_2$ we have $\alpha-\bar{\alpha}=1$, hence $\lambda=1$ and the modified Trace-Hermitian inner product reduces to the usual Trace-Hermitian inner product. We can choose an explicit identification of $\mathbb{F}_2^2=\{0=(0,0),z=(1,0),x=(0,1),y=(1,1)\}$ with $\mathbb{F}_4$, for example $(0,z,x,y)\longleftrightarrow(0,1,\alpha,\alpha^2)$. Then additive codes $C\subseteq\mathbb{F}_4^n$ which are Trace-Hermitian self-orthogonal correspond to $n$-qubit stabilizer codes.

\paragraph{Qutrits.} Here, take $\mathbb{F}_9\cong\mathbb{F}_3[\alpha]/(\alpha^2+1)$ with conjugation $x\mapsto\bar{x}=x^3$ and trace map $\Tr_{\mathbb{F}_9/\mathbb{F}_3}(x)=x+x^3$.\footnote{Different choices of the irreducible quadratic polynomial in $\mathbb{F}_9\cong\mathbb{F}_3[\alpha]/f(\alpha)$ (and different choices of basis element $\alpha$) lead to different but equivalent maps $\iota$ and forms Eq.~\eqref{Equation:ModifiedTraceHermitian}.} In this case, additive codes $C\subseteq\mathbb{F}_9^n$ which are self-orthogonal with respect to Eq.~\eqref{Equation:ModifiedTraceHermitian} correspond to $n$-qutrit stabilizer codes.

\subsection{Magic state distillation}

A magic state distillation protocol is a procedure using stabilizer operations (preparations of stabilizer states, Clifford unitary gates, Pauli measurements, tracing out subsystems, post-selection and/or classical conditioning on measurement outcomes), that consumes a number of copies of a noisy magic state to produce fewer copies of a state with higher fidelity to the pure target magic state. In this paper, we focus on $n$-to-$1$ qudit distillation protocols.

Formally, we fix a target pure magic state $\ket{M}$ and an experimentally-preparable mixed resource state $\rho$. An $n$-to-$1$ qudit magic state distillation protocol is a sequence of stabilizer operations described by a completely positive, trace non-increasing map $\Phi$. If the protocol distills $\ket{M}$, then there is a distillable region (a neighbourhood of states) around $\ket{M}$ such that if the resource state $\rho$ is in this distillable region, then with success probability $P_s=\Tr(\Phi(\rho^{\otimes n}))$, the post-distillation state $\rho'=\Phi(\rho^{\otimes n})/P_s$ has greater fidelity with respect to $\ket{M}$ than the resource state, $\Tr(\ketbra{M}{M}\rho')>\Tr(\ketbra{M}{M}\rho)$. Iterating the protocol pushes the noise to zero allowing us to prepare essentially pure magic states.

Campbell and Browne~\cite{CampbellBrowne2009} showed that, in analyzing useful magic state distillation protocols, it suffices to focus on protocols with a certain structure. In particular, they showed that for any $n$-to-$1$ qudit protocol $\Phi$ and resource state $\rho$, there is another protocol with the same or better performance (measured by the fidelity of the output state), of the form
\begin{equation}\label{Equation:StabilizerReduction}
    \rho\longmapsto\rho'=\frac{\mathcal{R}(\rho^{\otimes n})}{\Tr(\mathcal{R}(\rho^{\otimes n}))},
\end{equation}
where the recovery map $\mathcal{R}(A)$ is defined as
\begin{equation}
    \mathcal{R}(A)=\Tr_{2,\dots,n}(D\Pi_I^rA\Pi_I^rD^\dagger).
\end{equation}
Here, $I\subset\mathbb{F}_d^{2n}$ is an $n-1$-dimensional isotropic subspace, $r:I\to\mathbb{F}_d$ is a consistent value assignment, and $D$ is a decoding Clifford operation that maps the code space of $(I,r)$ onto the first physical qudit. In words, we (i)~prepare $n$ copies of the resource state, (ii)~measure commuting Pauli observables labeled by $I\subset\mathbb{F}_d^{2n}$ and post-select on measurement outcomes $r:I\to\mathbb{F}_d$ (equivalently, project with $\Pi_I^r$), (iii)~use a Clifford group operation to decode the code space, $I^\perp/I\cong\mathbb{F}_d^2$, onto the first physical qudit, and (iv)~trace out the last $n-1$ qudits.

Although Ref.~\cite{CampbellBrowne2009} proves this structural result only for qubits, the same reduction extends to prime-dimensional qudits with essentially the same argument, since the argument depends only on the properties of the stabilizer formalism that have direct analogues for odd-prime dimensions~\cite{Gottesman1998b}. We provide a sketch of the proof for prime-dimensional qudits in Appendix~\ref{Section:MSDStructure}.

We emphasize that this result concerns only the fidelity of the post-selected output state, and hence also the achievable noise threshold of the protocols. That is, the reduction of a general magic state distillation protocol to one of the form Eq.~\eqref{Equation:StabilizerReduction} preserves or improves the conditional output state fidelity, but it is not guaranteed to preserve the success probability of the protocol, and in fact, typically decreases it. Since the primary figure of merit in this paper is the distillation threshold, the restriction to stabilizer reductions is without loss of generality for our purposes.

\section{Quadratic residue codes}\label{Section:QRCodes}

Quadratic residue (QR) codes are three related families of codes defined by the set of quadratic residues modulo an odd prime $p$. To distinguish them, they are called the augmented, extended, and expurgated quadratic residue codes.

Let $\mathbb{F}_q$ be the finite field of order $q$ with $q$ a prime power, let $p$ be an odd prime such that $\gcd(p,q)=1$ and $q$ is a quadratic residue modulo $p$, and let $\zeta$ be a primitive $p^{th}$ root of unity in the splitting field $\mathbb{F}_{q^m}$ of $x^p-1$ over $\mathbb{F}_q$. Let $\mathcal{Q}$ be the set of nonzero quadratic residues modulo $p$, and let $\mathcal{N}=\mathbb{F}_p^\times\setminus\mathcal{Q}$ be the set of non-residues modulo $p$. We have $|\mathcal{Q}|=|\mathcal{N}|=(p-1)/2$.

The length $p$ augmented quadratic residue code over $\mathbb{F}_q$ is a $[p,(p+1)/2]_q$ cyclic code with generator polynomial $\prod_{s\in\mathcal{Q}}(x-\zeta^s)$. The length $p$ expurgated QR code over $\mathbb{F}_q$ is a $[p,(p-1)/2]_q$ cyclic code with generator polynomial $(x-1)\prod_{s\in\mathcal{Q}}(x-\zeta^s)$.

Extending the augmented QR code by adding a parity check digit gives the extended QR code, a $[p+1,(p+1)/2]_q$ code. The augmented QR code can be obtained from the extended QR code by puncturing on any coordinate, and the expurgated QR code can be obtained from the extended QR code by shortening on any coordinate.

\medskip

See, for example, Ref.~\cite{MacWilliamsSloane1977,PlessBrualdi1998} for more background on quadratic residue codes.

\subsection{Quantum quadratic residue codes}

When a code over $\mathbb{F}_{d^2}$ is Hermitian self-orthogonal, it can be identified with a quantum stabilizer code on $d$-dimensional qudits via the correspondence described in Section~\ref{Section:CodeCorrespondence}. This motivates the following definition.

\begin{Definition}
    A quantum quadratic residue code on $d$-dimensional qudits is the image under the map $\iota^{-1}$ defined in Section~\ref{Section:CodeCorrespondence} of a Hermitian self-orthogonal expurgated quadratic residue code over $\mathbb{F}_{d^2}$.
\end{Definition}

Quantum quadratic residue codes have previously been introduced~\cite{Rains1999,KetkarSarvepalli2005}, but only as examples of the procedure of constructing quantum codes from classical codes described in Section~\ref{Section:CodeCorrespondence}. In this paper, we explore these codes in depth, along with their properties and applications.

\medskip

The following theorem identifies when an expurgated quadratic residue code is Hermitian self-orthogonal, and thus, when it can be used to define a quantum quadratic residue code.

\begin{Theorem}\label{Theorem:QRHermitianOrthogonality}
    A length $p$ expurgated quadratic residue code over $\mathbb{F}_q$ where $q=r^2$ is Hermitian self-orthogonal if and only if $-r$ is a quadratic non-residue modulo $p$.
\end{Theorem}

\begin{proof}[Proof of Theorem~\ref{Theorem:QRHermitianOrthogonality}]
    We use two facts about cyclic codes: (i)~for two cyclic codes $C_1$ and $C_2$ with zeros $\mathcal{Z}(C_1)$ and $\mathcal{Z}(C_2)$, we have $C_1\subseteq C_2$ if and only if $\mathcal{Z}(C_2)\subseteq\mathcal{Z}(C_1)$~\cite[Lemma~3.5]{DicuangcoSole2007}, and (ii)~for a cyclic code $C\subset\mathbb{F}_q^n$, $q=r^2$, with zeros $\mathcal{Z}(C)$, the hermitian dual $C^{\perp_H}$ of $C$ is a cyclic code with zeros $\mathcal{Z}(C^{\perp_H})=\mathbb{Z}_n\setminus(-r\mathcal{Z}(C))$~\cite[Proposition~3.4]{DicuangcoSole2007}.

    Let $C$ be the length $p$ expurgated QR code over $\mathbb{F}_q$, with $q=r^2$. By definition, $C$ is the cyclic code with zeros
    \begin{equation}
        \mathcal{Z}(C):=\{0\}\cup\mathcal Q\subset\mathbb{Z}_p.
    \end{equation}
    Then the Hermitian dual $C^{\perp_H}$ of $C$ is a cyclic code with zeros
    \begin{equation}
        \mathcal{Z}(C^{\perp_H})=\mathbb{Z}_p\setminus (\{0\}\cup-r\mathcal{Q}).
    \end{equation}
    Therefore, $C\subseteq C^{\perp_H}$ if and only if $\mathbb{Z}_p\setminus(\{0\}\cup-r\mathcal{Q})\subseteq\{0\}\cup\mathcal{Q}$, or equivalently, $C\subseteq C^{\perp_H}$ if and only if $\mathcal{N}\subseteq-r\mathcal{Q}$. Since $\mathcal{Q}\le\mathbb{F}_p^\times$ is the index-$2$ subgroup of squares modulo $p$, multiplication by $a\in\mathbb{F}_p^\times$ preserves $\mathcal{Q}$ and $\mathcal{N}$ if and only if $a\in\mathcal{Q}$, and it swaps $\mathcal{Q}$ and $\mathcal{N}$ if and only if $a\in\mathcal{N}$. Thus, $\mathcal{N}\subseteq-r\mathcal{Q}$ if and only if $-r\in\mathcal{N}$, i.e., if and only if $-r$ is a quadratic non-residue modulo $p$.
\end{proof}

For the cases $q=4$ and $q=9$, we have the following corollaries.

\begin{Corollary}\label{Corollary:F4HermitianSelfOrthogonality}
    An expurgated quadratic residue code over $\mathbb{F}_4$ is Hermitian self-orthogonal if and only if the code length is $5$ or $7$ modulo $8$.
\end{Corollary}

\begin{proof}[Proof of Corollary~\ref{Corollary:F4HermitianSelfOrthogonality}]
    Let $C$ be a length $p$ expurgated QR code over $\mathbb{F}_4$. From Theorem~\ref{Theorem:QRHermitianOrthogonality}, the condition for $C$ to be Hermitian self-orthogonal is that $-2$ is a quadratic non-residue modulo $p$. Using multiplicativity of the Legendre symbol and the supplements to the law of quadratic reciprocity~\cite{NivenMontgomery1991},
    \begin{equation}
        \Legendre{-2}{p}=\Legendre{-1}{p}\Legendre{2}{p}=(-1)^{\frac{p-1}{2}}\cdot(-1)^{\frac{p^2-1}{8}}.
    \end{equation}
    Now we check the four possible cases for an odd prime $p$ mod $8$.
    \begin{equation}
        \begin{cases}
            \text{If $p\equiv1\mod{8}$, then $\frac{p-1}{2}$ is even and $\frac{p^2-1}{8}$ is even so $\Legendre{-2}{p}=1$.} \\
            \text{If $p\equiv3\mod{8}$, then $\frac{p-1}{2}$ is odd and $\frac{p^2-1}{8}$ is odd so $\Legendre{-2}{p}=1$.} \\
            \text{If $p\equiv5\mod{8}$, then $\frac{p-1}{2}$ is even and $\frac{p^2-1}{8}$ is odd so $\Legendre{-2}{p}=-1$.} \\
            \text{If $p\equiv7\mod{8}$, then $\frac{p-1}{2}$ is odd and $\frac{p^2-1}{8}$ is even so $\Legendre{-2}{p}=-1$.}
        \end{cases}
    \end{equation}
    Therefore, $-2$ is a quadratic non-residue modulo $p$ if and only if $p\equiv5\mod8$ or $p\equiv7\mod8$.
\end{proof}

\begin{Corollary}\label{Corollary:F9HermitianSelfOrthogonality}
    An expurgated quadratic residue code over $\mathbb{F}_9$ is Hermitian self-orthogonal if and only if the code length is $5$ or $11$ modulo $12$.
\end{Corollary}

\begin{proof}[Proof of Corollary~\ref{Corollary:F9HermitianSelfOrthogonality}]
    Let $C$ be a length $p$ expurgated QR code over $\mathbb{F}_9$. From Theorem~\ref{Theorem:QRHermitianOrthogonality}, the condition for $C$ to be Hermitian self-orthogonal is that $-3$ is a quadratic non-residue modulo $p$. Using multiplicativity of the Legendre symbol and quadratic reciprocity~\cite{NivenMontgomery1991},
    \begin{equation}
        \Legendre{-3}{p}=\Legendre{-1}{p}\Legendre{3}{p}=(-1)^{\frac{p-1}{2}}\cdot\Legendre{p}{3}(-1)^{\frac{p-1}{2}}=\Legendre{p}{3}.
    \end{equation}
    We have $\Legendre{p}{3}=-1$ if and only if $p\equiv2\mod{3}$. An odd prime $p$ must be $1$, $3$, $5$, $7$, or $11$ modulo $12$, and of these cases only $5$ and $11$ are compatible with $p$ being $2$ modulo $3$. Therefore, $-3$ is a quadratic non-residue mod $p$ if and only if $p\equiv5\mod12$ or $p\equiv11\mod12$.
\end{proof}

In the definition of quadratic residue codes above, we have some conditions on the length $p$ of the code that are required in order for the code to exist. Namely, we assume that $p$ is an odd prime, that $\gcd(p,q)=1$, and that $q$ is a quadratic residue modulo $p$. In the case of quantum quadratic residue codes on $d$-dimensional qudits, the corresponding expurgated quadratic residue code is over $\mathbb{F}_{d^2}$. That is, $q=d^2$ is a square. Therefore, $q$ is a quadratic residue modulo any prime $p$ such that $\gcd(p,q)=1$, and so the third condition is satisfied automatically. For the case $q=4$, any odd prime $p$ satisfies $\gcd(p,4)=1$, and so the conclusion of Corollary~\ref{Corollary:F4HermitianSelfOrthogonality} is that there exists a quantum QR code on qubits for any odd-prime length $p$ such that $p$ is $5$ or $7$ modulo $8$. For the case $q=9$, any odd prime $p>3$ satisfies $\gcd(p,9)=1$, and so with Corollary~\ref{Corollary:F9HermitianSelfOrthogonality}, there exists a quantum QR code on qutrits for any odd-prime length $p$ such that $p$ is $5$ or $11$ modulo $12$.

\subsection{Parameters of quantum QR codes}

As shown in the previous section, a quantum quadratic residue code on $d$-dimensional qudits exists for any odd-prime length $p$ such that $\gcd(p,d^2)=1$, and $-d$ is a quadratic non-residue modulo $p$. In particular, quantum QR codes on qubits exist for any odd-prime length that is $5$ or $7$ modulo $8$, and quantum QR codes on qutrits exist for any odd-prime length that is $5$ or $11$ modulo $12$.

The cardinality of a length $p$ expurgated QR code over $\mathbb{F}_q$, $q=r^2$, is $|C|=q^{(p-1)/2}=r^{p-1}$. Therefore, a length $p$ quantum QR code on $d$-dimensional qudits has a stabilizer group of size $d^{p-1}$, and has code space of dimension $d^{p-(p-1)}=d$. That is, quantum quadratic residue codes always encode exactly one qudit.

We can also determine the distance of quantum QR codes in terms of the distance of the corresponding expurgated QR codes. This is achieved by the following theorem.

\begin{Theorem}\label{Theorem:QQRDistance}
    Let $C$ be a Hermitian self-orthogonal expurgated quadratic residue code with parameters $[p,(p-1)/2,\delta_C]_{d^2}$. The distance of the corresponding $[\![p,1,\delta_Q]\!]_d$ quantum quadratic residue code $\iota^{-1}(C)$ is $\delta_Q=\delta_C-1$.
\end{Theorem}

The proof of this theorem relies on Prange's Theorem.

\begin{Theorem*}[Prange's Theorem]
    Let $C$ be a homogeneous $[n,k,\delta]$ code over $\mathbb{F}_q$ with $\delta>1$. Let $C^*$ be the code obtained from $C$ by puncturing on some coordinate, and let $C_*$ be the code obtained from $C$ by shortening on some coordinate. Then for $0\le i\le n-1$, we have
    \begin{align}\label{Equation:PrangeWeights}
        \begin{cases}
            A_i(C^*)=\frac{n-i}{n}A_i(C)+\frac{i+1}{n}A_{i+1}(C),\\
            A_i(C_*)=\frac{n-i}{n}A_i(C).
        \end{cases}
    \end{align}
\end{Theorem*}

The version of Prange's theorem quoted here is Theorem~10.12 in Ref.~\cite{PlessBrualdi1998}. An immediate consequence of this theorem is that the weight enumerators of the codes $C^*$ and $C_*$ can be obtained by taking partial derivatives of the weight enumerator polynomial of $C$. In particular,
\begin{equation}\label{Equation:PrangeDerivative1}
    W_{C^*}(x,y)=\frac{1}{n}\left(\frac{\partial}{\partial x} + \frac{\partial}{\partial y}\right)W_C(x,y)
\end{equation}
and
\begin{equation}\label{Equation:PrangeDerivative2}
    W_{C_*}(x,y)=\frac{1}{n}\frac{\partial}{\partial x}W_C(x,y).
\end{equation}

The precondition for Prange's theorem to apply is that the code in question must be homogeneous. A sufficient condition for homogeneity is that the automorphism group of the code acts transitively on the coordinate positions~\cite{PlessBrualdi1998}. The permutation action of the automorphism group on the coordinate positions of the extended quadratic residue code is transitive (in fact, it is $2$-transitive). This follows from the Gleason-Prange theorem~\cite{MattsonAssmus1964,Blahut1991}, which states that, labeling the coordinates of the extended QR code by the elements of the projective line $\mathrm{P}^1(\mathbb{F}_p)=\{0,1,\dots,p-1,\infty\}$, the permutation action of the automorphism group contains a copy of the projective special linear group $\mathrm{PSL}(2,p)$. The action of $\mathrm{PSL}(2,p)$ on the projective line is $2$-transitive, and so Prange's theorem applies to extended QR codes. The augmented and expurgated QR codes are obtained from the extended QR code by puncturing and shortening respectively, and so we can obtain the weight distributions of these codes from the weight distribution of the extended code via Prange's theorem.

\medskip

We can now proceed to the proof of Theorem~\ref{Theorem:QQRDistance}.

\begin{proof}[Proof of Theorem~\ref{Theorem:QQRDistance}]
    Let $C$ be a $[p,(p-1)/2,\delta_C]_{d^2}$ Hermitian self-orthogonal expurgated quadratic residue code, and let $I=\iota^{-1}(C)$ be the corresponding $[\![p,1,\delta_Q]\!]_d$ quantum quadratic residue code. Since $C$ is $\mathbb{F}_{d^2}$-linear, by Lemma~\ref{Lemma:MTHEqualsHermitian}, $C^{\perp_{MTH}}=C^{\perp_H}$, so
    \begin{equation}
        \delta_Q:=\min\left\{swt(v)\;|\;v\in I^\perp\setminus I\right\}=\min\left\{wt(c)\;|\;c\in C^{\perp_H}\setminus C\right\}.
    \end{equation}

    The zeros of the expurgated QR code are $\mathcal{Z}(C)=\{0\}\cup\mathcal{Q}$. Therefore, with Eq.~\eqref{Equation:HermitianDualZeros}, we have $\mathcal{Z}\left(C^{\perp_H}\right)=\mathbb{Z}_p\setminus(\{0\}\cup(-d\mathcal{Q}))$. Since $C$ is Hermitian self-orthogonal, $-d\in\mathcal{N}$ by Theorem~\ref{Theorem:QRHermitianOrthogonality}, so multiplication by $-d$ swaps $\mathcal{Q}$ and $\mathcal{N}$, and therefore $-d\mathcal{Q}=\mathcal{N}$. Thus $\mathcal{Z}(C^{\perp_H})=\mathbb{Z}_p\setminus(\{0\}\cup\mathcal{N})=\mathcal{Q}$, i.e. $C^{\perp_H}$ is the length-$p$ augmented QR code.
    
    Now let $\tilde{C}$ be the length $p+1$ extended QR code over $\mathbb{F}_{d^2}$. Puncturing $\tilde{C}$ on any coordinate gives the augmented code, and shortening on any coordinate gives the expurgated code (see, e.g., \cite{PlessBrualdi1998}). Prange's theorem, Eq.~\eqref{Equation:PrangeWeights}, implies that shortening preserves the minimum distance while puncturing decreases it by exactly $1$. Hence $\delta(C^{\perp_H})=\delta(\tilde C)-1$ and $\delta(C)=\delta(\tilde C)$, so $\delta(C^{\perp_H})=\delta_C-1$.
    
    Because, by definition, $C$ has no nonzero vector of weight less than $\delta_C$, any vector in $C^{\perp_H}$ of weight $\delta_C-1$ lies in $C^{\perp_H}\setminus C$, and therefore $\delta_Q=\min\{wt(c)\;|\;c\in C^{\perp_H}\setminus C\}=\delta_C-1$.
\end{proof}

\begin{Corollary}\label{Corollary:SquareRootBound}
    The minimum distance of a $[\![p,1,\delta]\!]_d$ quantum quadratic residue code satisfies
    \begin{equation}
        \delta\ge\lceil\sqrt{p}\rceil-1.
    \end{equation}
\end{Corollary}

\begin{proof}[Proof of Corollary~\ref{Corollary:SquareRootBound}]
    This follows from Theorem~\ref{Theorem:QQRDistance} together with the square root bound for classical quadratic residue codes~\cite[Ch.~16]{MacWilliamsSloane1977}.
\end{proof}

The parameters of some of the Hermitian self-orthogonal expurgated quadratic residue codes over $\mathbb{F}_4$ and $\mathbb{F}_9$ are shown in Table~\ref{Table:CodeParameters}, together with the parameters of the corresponding qubit and qutrit quantum quadratic residue codes.

\begin{table}
    \centering
    \parbox{0.49\linewidth}{
    \centering
    \begin{tabular}{|c|c|}
        \hline
        Classical code & Quantum code \\
        \hline\hline
        $[5,2,4]_4$ & $[\![5,1,3]\!]_2$ \\
        \hline
        $[7,3,4]_4$ & $[\![7,1,3]\!]_2$ \\
        \hline
        $[13,6,6]_4$ & $[\![13,1,5]\!]_2$ \\
        \hline
        $[23,11,8]_4$ & $[\![23,1,7]\!]_2$ \\
        \hline
        $[29,14,12]_4$ & $[\![29,1,11]\!]_2$ \\
        \hline
        $[31,15,8]_4$ & $[\![31,1,7]\!]_2$ \\
        \hline
        $[37,18,12]_4$ & $[\![37,1,11]\!]_2$ \\
        \hline
        $[47,23,12]_4$ & $[\![47,1,11]\!]_2$ \\
        \hline
        $[53,26,16]_4$ & $[\![53,1,15]\!]_2$ \\
        \hline
        $[61,30,18]_4$ & $[\![61,1,17]\!]_2$ \\
        \hline
        $[71,35,12]_4$ & $[\![71,1,11]\!]_2$ \\
        \hline
    \end{tabular}
    }
    \parbox{0.49\linewidth}{
    \centering
    \begin{tabular}{|c|c|}
        \hline
        Classical code & Quantum code\\
        \hline\hline
        $[5,2,4]_9$ & $[\![5,1,3]\!]_3$ \\
        \hline
        $[11,5,6]_9$ & $[\![11,1,5]\!]_3$ \\
        \hline
        $[17,8,8]_9$ & $[\![17,1,7]\!]_3$ \\
        \hline
        $[23,11,9]_9$ & $[\![23,1,8]\!]_3$ \\
        \hline
        $[29,14,12]_9$ & $[\![29,1,11]\!]_3$ \\
        \hline
        $[41,20,14]_9$ & $[\![41,1,13]\!]_3$ \\
        \hline
        $[47,23,15]_9$ & $[\![47,1,14]\!]_3$ \\
        \hline
    \end{tabular}
    }
    \caption{Parameters of classical expurgated quadratic residue codes over $\mathbb{F}_4$ and $\mathbb{F}_9$ which are Hermitian self-orthogonal, and the corresponding parameters of the qubit and qutrit quantum quadratic residue codes.}
    \label{Table:CodeParameters}
\end{table}

\subsection{CSS quantum QR codes}

In some cases, quantum quadratic residue codes happen to be CSS codes. This is the result of the following theorem.

\begin{Theorem}\label{Theorem:CSSQRCodes}
    Let $C$ be a Hermitian self-orthogonal expurgated quadratic residue code of length $p$ over $\mathbb{F}_{d^2}$. The corresponding quantum quadratic residue code, $\iota^{-1}(C)$, is a CSS code if and only if $d$ is a quadratic residue modulo $p$. Furthermore, if $d$ is a quadratic residue modulo $p$, then an expurgated quadratic residue code $D$ of length $p$ over $\mathbb{F}_d$ exists, and the $Z$ and $X$ stabilizers of $\iota^{-1}(C)$ are each given by $D$. That is, $\iota^{-1}(C)$ has a generator matrix of the form
    \begin{equation}
        G=\begin{bmatrix}
            H & | & 0\\
            0 & | & H
        \end{bmatrix},
    \end{equation}
    where $H$ is a generator matrix of $D$.
\end{Theorem}

\begin{proof}[Proof of Theorem~\ref{Theorem:CSSQRCodes}]
    First we consider the ``only if'' direction. Suppose $I=\iota^{-1}(C)$ is a CSS code, that is, $I=(D_Z\times\{0\})\oplus(\{0\}\times D_X)$ for classical codes $D_Z,D_X\subseteq\mathbb{F}_d^p$, so that $C=D_Z+\alpha D_X$. Since $C$ is a linear code, it is preserved by multiplication by $\alpha$. Therefore, for any $v\in D_Z$, the image of $\alpha v$ under $\iota^{-1}$ is $(0,v)$, so $D_Z\subseteq D_X$. Similarly, for $v\in D_X$ the image of $\alpha^2v$ under $\iota^{-1}$ has $Z$ component a nonzero multiple of $v$, so $D_X\subseteq D_Z$. Hence $D_Z=D_X=:D$ and $C=D\otimes_{\mathbb{F}_d}\mathbb{F}_{d^2}$, with $\dim_{\mathbb{F}_d}(D)=(p-1)/2$. The code $D$ inherits cyclicity from $C$, and its generator polynomial, being the minimal-degree monic polynomial of $C$, is $g(x)=(x-1)\prod_{s\in\mathcal{Q}}(x-\zeta^s)$. Applying the Frobenius automorphism $y\mapsto y^d$ to the coefficients fixes $g\in\mathbb{F}_d[x]$ while permuting its roots by $\zeta^s\mapsto\zeta^{ds}$. Therefore $d\mathcal{Q}=\mathcal{Q}$, and since $\mathcal{Q}$ is the subgroup of squares in $\mathbb{F}_p^\times$, $d\in\mathcal{Q}$.

    Now consider the ``if'' direction. Assume $d$ is a quadratic residue modulo $p$. Then an expurgated QR code $D$ of length $p$ over $\mathbb{F}_d$ exists~\cite{MacWilliamsSloane1977}. Let $R=\mathbb{F}_d[x]/(x^p-1)$, and let $e(x)\in R$ be the idempotent of $D$. By Theorem~12.3 of Ref.~\cite{PlessBrualdi1998}, $e(x)$ is also the idempotent of the expurgated QR code of length $p$ over any extension field of $\mathbb{F}_d$, and in particular, $e(x)$ is the idempotent of $C$. That is, $D=\langle e(x)\rangle_{\mathbb{F}_d}$ and $C=\langle e(x)\rangle_{\mathbb{F}_{d^2}}$.

    Now use the same choice of $\alpha\in\mathbb{F}_{d^2}\setminus\mathbb{F}_d$ as in Section~\ref{Section:CodeCorrespondence}, so that $\mathbb{F}_{d^2}=\mathbb{F}_d\oplus\alpha\mathbb{F}_d$ as an $\mathbb{F}_d$-vector space, and hence $\mathbb{F}_{d^2}[x]/(x^p-1)=R\oplus \alpha R$. Therefore,
    \begin{equation}
        C=(R\oplus \alpha R)e=Re+\alpha Re=D+\alpha D.
    \end{equation}
    Then, by definition of the bijection $\iota$, we have $\iota^{-1}(C)=D\times D$. Therefore, the corresponding quantum code is a CSS stabilizer code with both $Z$-checks and $X$-checks given by $D$.

    Finally, if $H$ is a generator matrix for $D$, then the rows of $[H\,|\,0]$ and $[0\,|\,H]$ generate $\iota^{-1}(C)\cong D\times D\subset\mathbb{F}_d^n\times\mathbb{F}_d^n$. Therefore $\iota^{-1}(C)$ has generator matrix
    \begin{equation}
        G = \begin{bmatrix}
            H & | & 0\\
            0 & | & H
        \end{bmatrix},
    \end{equation}
    as claimed.
\end{proof}

\begin{Corollary}\label{Corollary:QubitCSSQRCodes}
    A quantum quadratic residue code on qubits is a CSS code if and only if the length of the code is a prime $p\equiv7\mod8$.
\end{Corollary}

\begin{proof}[Proof of Corollary~\ref{Corollary:QubitCSSQRCodes}]
    An expurgated quadratic residue code of odd-prime length $p$ over $\mathbb{F}_2$ exists if and only if $2$ is a quadratic residue modulo $p$. By the second supplement to the law of quadratic reciprocity~\cite{NivenMontgomery1991}, this occurs if and only if $p\equiv\pm1\mod8$. Incorporating the Hermitian self-orthogonality condition Corollary~\ref{Corollary:F4HermitianSelfOrthogonality}, a length $p$ quantum quadratic residue code on qubits is a CSS code if and only if $p\equiv7\mod8$.
\end{proof}

\begin{Corollary}\label{Corollary:QutritCSSQRCodes}
    A quantum quadratic residue code on qutrits is a CSS code if and only if the length of the code is a prime $p\equiv11\mod12$.
\end{Corollary}

\begin{proof}[Proof of Corollary~\ref{Corollary:QutritCSSQRCodes}]
    An expurgated quadratic residue code of odd-prime length $p$ over $\mathbb{F}_3$ exists if and only if $3$ is a quadratic residue modulo $p$. Using the third supplement to the law of quadratic reciprocity~\cite{NivenMontgomery1991}, this occurs if and only if $p\equiv\pm1\mod12$. With the Hermitian self-orthogonality condition Corollary~\ref{Corollary:F9HermitianSelfOrthogonality}, a length $p$ quantum quadratic residue code on qutrits is a CSS code if and only if $p\equiv11\mod12$.
\end{proof}

For qubits, the smallest example of a non-CSS quantum quadratic residue code is the $[\![5,1,3]\!]_2$ code. This code is equivalent to the well-known $5$-qubit perfect code~\cite{BennettWootters1996,LaflammeZurek1996} (also called the Laflamme code) with stabilizer group
\begin{equation}
    \langle XZZXI,IXZZX,XIXZZ,ZXIXZ\rangle,
\end{equation}
where by equivalent we mean up to permutations of qudits and local Cliffords. This can be verified directly, for example, using the Magma computer algebra system~\cite{Magma1997}. The smallest example of a CSS code among qubit quantum QR codes is the $[\![7,1,3]\!]_2$ code. This is the $7$-qubit Steane code~\cite{Steane1996}, a CSS code with stabilizer group
\begin{equation}
    \langle IIIXXXX,IXXIIXX,XIXIXIX,IIIZZZZ,IZZIIZZ,ZIZIZIZ\rangle.
\end{equation}
Each of the $Z$ and $X$ stabilizers are formed from a copy of the dual of the $[7,4,3]_2$ binary Hamming code, which itself is equivalent to the length $7$ binary quadratic residue code.

For qutrits, the smallest example of a CSS code among quantum QR codes is the $11$-qutrit Golay code~\cite{Prakash2020}. This is a CSS code with $Z$ and $X$ stabilizers each given by a copy of the dual of the length $11$ ternary Golay code. The length $11$ ternary Golay code is equivalent to the length $11$ ternary quadratic residue code~\cite{MacWilliamsSloane1977}, and so the $11$-qutrit Golay code is equivalent to the $11$-qutrit quantum quadratic residue code. The smallest non-CSS code is the $[\![5,1,3]\!]_3$ code with stabilizer generator matrix\footnote{Here we use the conventions established in Section~\ref{Section:CodeCorrespondence} for the identification of $\mathbb{F}_9^n$ with $\mathbb{F}_3^{2n}$.}
\begin{equation}\label{Equation:QutritCodeGeneratorMatrix}
    \left[\begin{array}{cc|cc}
        \mathbbm{1}_\ell & A_Z & \bbzero_\ell & A_X \\
        \bbzero_\ell & 2A_X & \mathbbm{1}_\ell & A_Z
    \end{array}\right]
\end{equation}
where $\ell=2$,
\begin{equation}
    A_Z=\begin{bmatrix}
        2&2&1\\
        1&2&2
    \end{bmatrix}
    \qquad\text{and}\qquad
    A_X=\begin{bmatrix}
        0&2&1\\
        1&2&0
    \end{bmatrix}.
\end{equation}
The next non-CSS code is the $[\![17,1,7]\!]_3$ code with stabilizer generator matrix Eq.~\eqref{Equation:QutritCodeGeneratorMatrix} where $\ell=8$,
\begin{equation}
    A_Z=\begin{bmatrix}
        2&2&1&2&0&0&1&2&1\\
        1&2&1&2&1&1&2&2&2\\
        2&2&0&0&1&2&2&1&1\\
        1&2&1&1&2&2&1&0&1\\
        1&0&1&2&2&1&1&2&1\\
        1&1&2&2&1&0&0&2&2\\
        2&2&2&1&1&2&1&2&1\\
        1&2&1&0&0&2&1&2&2
    \end{bmatrix},
    \qquad\text{and}\qquad
    A_X=\begin{bmatrix}
        0&1&0&0&1&2&0&0&2\\
        2&1&2&2&2&2&0&2&2\\
        2&2&2&1&0&1&0&2&2\\
        2&0&0&1&0&1&2&2&1\\
        1&2&2&1&0&1&0&0&2\\
        2&2&0&1&0&1&2&2&2\\
        2&2&0&2&2&2&2&1&2\\
        2&0&0&2&1&0&0&1&0
    \end{bmatrix}.
\end{equation}

\subsection{Transversal Clifford gates of quantum QR codes}
In this section, we describe the transversal Clifford gates of quantum quadratic residue codes. Transversal gates act independently on the physical qudits, and therefore do not spread errors within a code block, making them a standard way of fault-tolerantly implementing logical operations~\cite{NielsenChuang2010,Gottesman2024}. For the case of qubits, we can use the recent classification of Ref.~\cite{DasuBurton2025}, Theorem~6.1 therein provides a complete characterization of transversal Clifford gates of qubit stabilizer codes. This classification is a qubit-specific refinement of the approach based on endomorphism algebras of Ref.~\cite{Rains1999} which applies to stabilizer codes on qudits of any prime dimension.

Following Ref.~\cite{DasuBurton2025}, we consider $\ell$ code blocks of an $n$-qudit stabilizer code arranged as an $\ell\times n$ array of qudits with one code block per row, and we consider gates of the form $g^{\otimes n}$ where $g$ is a fixed $\ell$-qudit Clifford gate applied to each column. In Ref.~\cite{DasuBurton2025}, these are called \emph{diagonal transversal gates} (note that the term ``diagonal'' refers to their action on code blocks, the gates themselves need not be diagonal). Conjugating by $g$ induces a symplectic transformation $S_g\in\mathrm{Sp}(\mathbb{F}_d^{2\ell})$ of the phase space $\mathbb{F}_d^{2\ell}$ (i.e. the labels of the Pauli operators Eq.~\eqref{Equation:PauliOperators}) of a single column, and conjugating by $g^{\otimes n}$ preserves the stabilizer group of $\ell$ code blocks if and only if the column-wise action of $S_g$ preserves $I^{\oplus\ell}$. Therefore, we identify the group of transversal Clifford gates on $\ell$ code blocks with
\begin{equation}
    G_I^\ell=\left\{S\in\mathrm{Sp}(\mathbb{F}_d^{2\ell})\;|\;S\text{ applied to each column preserves }I^{\oplus\ell}\right\}.
\end{equation}
Phases can be ignored for this identification since a transversal Clifford gate preserving the stabilizer group up to phases can be corrected by composing with a Pauli operator~\cite{Gross2006,DasuBurton2025}.

Let $I$ be a quantum quadratic residue code and $C$ the corresponding linear code over $\mathbb{F}_{d^2}$, and $C^{(\ell)}$ denote $\ell$ code blocks of $C$. Identify the phase space of a column with $\mathbb{F}_{d^2}^\ell$ via the map $\iota$ defined in Section~\ref{Section:CodeCorrespondence} applied to each of the $\ell$ qudits. Any matrix $A\in M_\ell(\mathbb{F}_{d^2})$ acting column-wise preserves $C^{(\ell)}$, since the rows of the image are $\mathbb{F}_{d^2}$-linear combinations of codewords of $C$, and by Eq.~\eqref{Equation:SymplecticEqualsMTH} its action is symplectic if and only if $\bar{A}^TA=\mathbbm{1}$, where $\bar{A}$ denotes the entry-wise Frobenius conjugate. Writing $\mathrm{U}(\mathbb{F}_{d^2}^\ell):=\{A\in M_\ell(\mathbb{F}_{d^2})\;|\;\bar{A}^TA=\mathbbm{1}\}$ for the unitary group of the Hermitian form Eq.~\eqref{Equation:HermitianProduct}, we obtain $\mathrm{U}(\mathbb{F}_{d^2}^\ell)\subseteq G^\ell_I$ for all $\mathbb{F}_{d^2}$-linear codes~\cite{Rains1999,DasuBurton2025}, and in particular, for quantum quadratic residue codes. When the code is not CSS, these are all of the transversal Clifford gates, $G^\ell_I=\mathrm{U}(\mathbb{F}_{d^2}^\ell)$. This follows from Theorem~15 and Example~2 of Ref.~\cite{Rains1999} together with Theorem~5.5 of Ref.~\cite{DasuBurton2025}. On the other hand, when the codes are $\mathbb{F}_{d^2}$-linear, and are also CSS codes with symmetric $Z$ and $X$ stabilizers, we have additional transversal gates, namely, $G_I^\ell=\mathrm{Sp}(\mathbb{F}_d^{2\ell})$~\cite{DasuBurton2025,PrakashSinghal2026}.

We now specialize to qubit and qutrit quantum quadratic residue codes.

\paragraph{Qubits.} For qubit quadratic residue codes we have two cases. First, if the length of the code is $p\equiv7\mod8$, then the code is $\mathbb{F}_4$-linear, and Corollary~\ref{Corollary:QubitCSSQRCodes} shows that it is a CSS code with symmetric $Z$ and $X$ stabilizers. For such codes, the group of transversal Clifford gates is $\mathrm{Sp}(\mathbb{F}_2^{2\ell})$~\cite{DasuBurton2025}. That is, the codes have a complete transversal Clifford group. On the other hand, Corollary~\ref{Corollary:QubitCSSQRCodes} shows that codes of length $p\equiv5\mod8$ are still $\mathbb{F}_4$-linear, but they are not CSS codes. The group of transversal Clifford gates for these codes is $\mathrm{U}(\mathbb{F}_4^\ell)$~\cite{DasuBurton2025}. Specializing to a single code block, $\ell=1$, $\mathrm{U}(\mathbb{F}_4^1)=\{\mu\in\mathbb{F}_4^\times\;|\;\bar{\mu}\mu=\mu^3=1\}=\mathbb{F}_4^\times\cong\mathbb{Z}_3$. This group is generated by the transversal $M_3$ gate (also called the facet gate), which cyclically permutes the Pauli $X$, $Y$, and $Z$ operators. Furthermore, these codes have no non-trivial transversal two-qubit gates, the only transversal two-qubit gates are swap gates followed by transversal single-qubit gates. This is spelled out explicitly for the case of the $[\![5,1,3]\!]_2$ code in Corollaries~7.2 and~7.3 of Ref.~\cite{DasuBurton2025}. The result is the same for all qubit quadratic residue codes of length $p\equiv5\mod8$.

\paragraph{Qutrits.} For qutrits, we also have two cases depending on whether or not the code is a CSS code. First, for codes of length $p\equiv11\mod12$, the codes are $\mathbb{F}_9$-linear, and according to Corollary~\ref{Corollary:QutritCSSQRCodes} they are CSS codes with symmetric $Z$ and $X$ stabilizers. Such codes have $\mathrm{Sp}(\mathbb{F}_3^{2\ell})$ as their group of transversal Clifford gates, i.e., they have a full transversal Clifford group~\cite[Lemma~1]{PrakashSinghal2026}. On the other hand, Corollary~\ref{Corollary:QutritCSSQRCodes} says that qutrit QR codes of length $p\equiv5\mod12$ are $\mathbb{F}_9$-linear, but they are not CSS codes. Such codes have $\mathrm{U}(\mathbb{F}_9^\ell)$ as their group of transversal gates~\cite{Rains1999}. On a single code block, $\mathrm{U}(\mathbb{F}_9^1)=\{\mu\in\mathbb{F}_9^\times\;|\;\bar{\mu}\mu=\mu^4=1\}\cong\mathbb{Z}_4$, namely, the cyclic group $\langle\alpha\rangle=\{\pm1,\pm\alpha\}$. Multiplication by $\alpha$ acts on $(v_z|v_x)$ as $\left(\begin{smallmatrix}0&-1\\1&0\end{smallmatrix}\right)$, the symplectic action of the inverse of the qutrit Fourier gate $F=\frac{1}{\sqrt{3}}\sum_{j,k\in\mathbb{Z}_3}\omega^{jk}\ketbra{j}{k}$. That is, non-CSS qutrit quantum QR codes have a transversal Fourier gate, with $F^2$ the parity gate $\ket{j}\mapsto\ket{-j}$ and $F^4=\mathbbm{1}$. In contrast with the qubit case, across two blocks the group $\mathrm{U}(\mathbb{F}_9^2)$, of order $96$, strictly contains the subgroup of monomial transformations of order $32$, and so there are non-trivial transversal two-qutrit gates.

\section{Magic state distillation with quantum QR codes}\label{Section:MSD}

Quantum quadratic residue codes are well-suited for defining $n$-to-$1$ qudit magic state distillation protocols. In this section, we first show that some quantum quadratic residue codes are equivalent to codes that are known to distill magic states, including the qubit $H$ and $T$ states and the qutrit Strange state. Then we analyze the performance of new quantum quadratic residue codes for distilling qubit $T$ states and qutrit Strange states.

\subsection{Special cases}\label{Section:MSDSpecialCases}

Some quantum quadratic residue codes are equivalent to codes that have previously been shown to distill magic states, where by equivalent we mean up to permutations of qudits and local Cliffords. For example, as mentioned above, the qubit quantum quadratic residue code of length $5$ is equivalent to the $5$-qubit perfect code~\cite{BennettWootters1996,LaflammeZurek1996} (also called the Laflamme code) with parameters $[\![5,1,3]\!]_2$. This code was the first code shown to distill $T$ states~\cite{BravyiKitaev2005}, and it still has the best threshold among all codes known to distill $T$ states.

As another example, the $7$-qubit quantum QR code is equivalent to the Steane code~\cite{Steane1996}, with parameters $[\![7,1,3]\!]_2$. This code distills qubit $H$ states~\cite{Reichardt2005}. Since the length of the code is $7$ mod $8$, Corollary~\ref{Corollary:QubitCSSQRCodes} predicts that this code is a CSS code, as expected.

Finally, the $11$-qutrit Golay code is known to distill qutrit Strange states~\cite{Prakash2020}. As mentioned above, this code is equivalent to the $11$-qutrit quantum QR code. Since the code length is $11$ mod $12$, this is a CSS code consistent with Corollary~\ref{Corollary:QutritCSSQRCodes}. Similarly, the $23$-qubit Golay code, a CSS code constructed from two copies of the dual of the length $23$ binary Golay code, is equivalent to the $23$-qubit quantum quadratic residue code. This code distills $H$ states~\cite{Reichardt2005}.

\subsection{Qubit \texorpdfstring{$T$}{T} state distillation}\label{Section:TStateDistillation}

The state $\ket{T}=\cos\beta\ket{0}+e^{i\pi/4}\sin\beta\ket{1}$ with $\cos2\beta=1/\sqrt{3}$ is one of the first proposed magic states~\cite{BravyiKitaev2005}. It is the single-qubit state with density matrix
\begin{equation}
    \ketbra{T}{T}=\frac{1}{2}\left(I+\frac{1}{\sqrt{3}}(X+Y+Z)\right).
\end{equation}
Rall showed that the performance of a magic state distillation protocol distilling $T$ states is determined by the ``signed'' weight enumerators of the corresponding code~\cite{Rall2017}, and more recently, Kalra and Prakash showed that if the code has a transversal $M_3$ gate, then the performance is determined by the usual (unsigned) weight enumerators~\cite{KalraPrakash2025}. For completeness, we partially reproduce these results connecting distillation performance to weight enumerators here. Then we show how to compute the weight enumerators of quantum quadratic residue codes, and we analyze their performance for $T$ state distillation.

\subsubsection{\texorpdfstring{$T$}{T} state distillation and weight enumerators}\label{Section:TDistillationAndWeights}

Let $M_3$ denote the single-qubit Clifford gate which cyclically permutes the Pauli operators, $M_3XM_3^\dagger=Y$, $M_3YM_3^\dagger=Z$, and $M_3ZM_3^\dagger=X$. Averaging over the group generated by $M_3$ gives a stabilizer channel
\begin{equation}
    \mathcal{T}_T(\rho):=\frac{1}{3}\sum\limits_{k\in\mathbb{Z}_3}M_3^k\rho{M_3^k}^\dagger,
\end{equation}
which maps arbitrary single-qubit states onto the line segment
\begin{equation}\label{Equation:TStateLine}
    \frac{1}{2}\mathbbm{1}+\frac{t}{2\sqrt{3}}(X+Y+Z),\quad t\in[-1,1].
\end{equation}
For $T$ state distillation, we assume we have access to a supply of noisy $T$ states of the form
\begin{equation}\label{Equation:NoisyTState}
    \rho_T(\epsilon)=(1-\epsilon)\ketbra{T}{T}+\frac{\epsilon}{2}\mathbbm{1}=\frac{1}{2}\mathbbm{1}+\frac{1-\epsilon}{2\sqrt{3}}(X+Y+Z).
\end{equation}
No generality is lost in making this assumption since given any resource state we can always apply the twirl $\mathcal{T}_T$ to get a state of this form. Similarly, after a distillation protocol is performed, we can apply the twirl $\mathcal{T}_T$ again to get a new state $\rho_T(\epsilon')$ with a new noise parameter, ideally satisfying $\epsilon'<\epsilon$. A protocol is said to distill $T$ states with threshold $\epsilon_{\mathrm{th}}$ if there exists an $\epsilon_{\mathrm{th}}>0$ such that $\epsilon'(\epsilon)<\epsilon$ for all $0<\epsilon<\epsilon_{\mathrm{th}}$, where $\epsilon'(\epsilon)$ is the post-distillation noise as a function of the pre-distillation noise.

The starting point of an $n$-to-$1$ qudit distillation protocol is an $n$-fold tensor product of noisy $T$ states. These states can be expanded in the Pauli basis as
\begin{equation}\label{Equation:nTStatePauliBasis}
    \rho_T(\epsilon)^{\otimes n}=\frac{1}{2^n}\sum_{v\in\mathbb{F}_2^{2n}}t(\epsilon)^{swt(v)}T_v,
\end{equation}
where $t(\epsilon)=(1-\epsilon)/\sqrt{3}$, and $swt(v)$ is the symplectic weight of the Pauli label $v\in\mathbb{F}_2^{2n}$.

Fix an $n$-to-$1$ qudit distillation protocol defined by a $[\![n,1]\!]_2$ stabilizer code $I\subset\mathbb{F}_2^{2n}$, post-selection on measurement outcomes $r:I\to\mathbb{F}_2$, and decoding Clifford $D$. The recovery map of the protocol is
\begin{equation}
    \mathcal{R}\left(\rho^{\otimes n}\right)=\Tr_{2,\dots,n}\left(D\Pi_I^r\rho^{\otimes n}\Pi_I^rD^\dagger\right).
\end{equation}
Let $P_s(\epsilon)=\Tr(\mathcal{R}(\rho_T(\epsilon)^{\otimes n}))$ be the success probability of the protocol starting with resource state $\rho_T(\epsilon)$. This gives us the first connection to the signed weight enumerator of the code $(I,r)$, defined as
\begin{equation}
    W_I^r(x,y):=\sum\limits_{v\in I}(-1)^{r(v)}x^{n-swt(v)}y^{swt(v)},
\end{equation}
analogous to Eq.~\eqref{Equation:WeightEnumerator}. Namely, the success probability can be written as
\begin{align}
    P_s(\epsilon) =& \Tr(\Pi_I^r\rho_T(\epsilon)^{\otimes n})\\
    =& \frac{1}{|I|}\sum\limits_{v\in I}(-1)^{r(v)}\Tr(T_v\rho_T(\epsilon)^{\otimes n})\\
    =& \frac{1}{2^{n-1}}\sum\limits_{v\in I}(-1)^{r(v)}t(\epsilon)^{swt(v)}\\
    =& \frac{1}{2^{n-1}}W_I^r\left(1,t(\epsilon)\right).
\end{align}
Here we get the second line by expanding the projector as in Eq.~\eqref{Equation:PauliProjector} and using linearity of the trace, and the third line uses the Pauli basis expansion of noisy $T$ states, Eq.~\eqref{Equation:nTStatePauliBasis}. 

Next we compute the post-distillation state,
\begin{equation}
    \rho':=\frac{\mathcal{R}\left(\rho_T(\epsilon)^{\otimes n}\right)}{P_s(\epsilon)},
\end{equation}
which, after twirling by $\mathcal{T}_T$ gives a state $\mathcal{T}_T(\rho')=\rho_T(\epsilon')$ with noise parameter $\epsilon'$. First, define the signed coset weight enumerators of the code $(I,r)$ as
\begin{equation}
    W_{u+I}^r(x,y)=\sum\limits_{v\in I}(-1)^{r(v)+\beta(u,v)}x^{n-swt(u+v)}y^{swt(u+v)},\quad\forall u\in I^\perp\setminus I,
\end{equation}
and the signed dual weight enumerator as
\begin{equation}
    W_{I^\perp}^r(x,y)=\sum\limits_{u\in I^\perp/I}W_{u+I}^r(x,y)
\end{equation}
where the outer sum is over representatives of each coset. For any non-identity single-qubit Pauli operator $T_u,\,u\in\{x,y,z\}$, the corresponding logical Pauli operator for the code is $T_{\bar{u}}=D^\dagger(T_u\otimes\mathbbm{1}_{2^{n-1}})D$. Since $\bar{u}\in I^\perp$, $T_{\bar{u}}$ commutes with the projector $\Pi_I^r$, and so the Pauli expectation value of the post-distillation state is
\begin{align}
    \Tr(T_u\rho') =& \frac{1}{P_s(\epsilon)}\Tr(T_u\Tr_{2,\dots,n}(D\Pi_I^r\rho_T(\epsilon)^{\otimes n}\Pi_I^rD^\dagger))\\
    =& \frac{1}{P_s(\epsilon)}\Tr(T_{\bar{u}}\Pi_I^r\rho_T(\epsilon)^{\otimes n})\\
    =& \frac{1}{2^{n-1}\cdot P_s(\epsilon)}\sum\limits_{v\in I}(-1)^{r(v)+\beta(\bar{u},v)}t(\epsilon)^{swt(\bar{u}+v)}\\
    =& \frac{W_{\bar{u}+I}^r\left(1,t(\epsilon)\right)}{W_I^r\left(1,t(\epsilon)\right)}.
\end{align}
After twirling the post-distillation state by $\mathcal{T}_T$, we get a depolarized $T$ state $\rho_T(\epsilon')$ with
\begin{align}
    \frac{1-\epsilon'}{\sqrt{3}} =& \frac{1}{3}\frac{W_{\bar{x}+I}^r\left(1,t(\epsilon)\right)+W_{\bar{y}+I}^r\left(1,t(\epsilon)\right)+W_{\bar{z}+I}^r\left(1,t(\epsilon)\right)}{W_I^r\left(1,t(\epsilon)\right)}\\
    =& \frac{1}{3}\frac{W_{I^\perp}^r\left(1,t(\epsilon)\right)-W_I^r\left(1,t(\epsilon)\right)}{W_I^r\left(1,t(\epsilon)\right)},
\end{align}
where the last equality follows from $I^\perp=I\sqcup(\bar{x}+I)\sqcup(\bar{y}+I)\sqcup(\bar{z}+I)$. Solving this for $\epsilon'$, we get the post-distillation noise $\epsilon'$ as a function of the pre-distillation noise $\epsilon$,
\begin{equation}\label{Equation:TPostStateSignedWeights}
    \epsilon'(\epsilon)=1-\frac{1}{\sqrt{3}}\frac{W_{I^\perp}^r\left(1,t(\epsilon)\right)-W_I^r\left(1,t(\epsilon)\right)}{W_I^r\left(1,t(\epsilon)\right)}.
\end{equation}

\medskip

At this stage, we will reduce generality by assuming the code $(I,r)$ is an $M_3$ code, i.e., it has a transversal $M_3$ gate. This means a logical $M_3$ gate is simply a physical $M_3$ gate applied to each qubit. A transversal $M_3$ gate heavily constrains the signs $r:I\to\mathbb{F}_2$. Concretely, a transversal $M_3$ gate means that conjugation by $M_3^{\otimes n}$ must preserve the stabilizer group $S_I^r=\{(-1)^{r(a)}T_a\;|\;a\in I\}$, so that for every stabilizer group element $P\in S_I^r$, we also have $M_3^{\otimes n}P(M_3^{\otimes n})^\dagger\in S_I^r$. This gives a constraint on the signs known as Rall's rule, the proof of which we reproduce here.

\begin{Lemma*}[Rall's rule~\cite{Rall2017,KalraPrakash2025}]
    Let $(I,r)$ be an $[\![n,1]\!]_2$ stabilizer code with a transversal $M_3$ gate. Then for every $u\in I$, the stabilizer group element $P_u:=(-1)^{r(u)}T_u$ has even symplectic weight, and its sign $(-1)^{r(u)}$ is determined by
    \begin{equation}\label{Equation:RallsRule}
        (-1)^{r(u)}=
        \begin{cases}
            +1, & swt(u)\equiv0\pmod 4,\\
            -1, & swt(u)\equiv2\pmod 4.
        \end{cases}
    \end{equation}
    Equivalently, since $swt(u)$ is even, $(-1)^{r(u)}=(-1)^{swt(u)/2}$.
\end{Lemma*}

\begin{proof}[Proof of Rall's rule]
    Fix $u\in I$ and let $P_u=(-1)^{r(u)}T_u\in S_I^r$. Define its $M_3$ conjugates
    \begin{equation}
        P_u':=M_3^{\otimes n}P_u(M_3^{\otimes n})^\dagger,\qquad\text{and}\qquad
        P_u'':=(M_3^{\otimes n})^\dagger P_uM_3^{\otimes n}.
    \end{equation}
    Because $I$ is an $M_3$ code, $P_u',P_u''\in S_I^r$. Since $S_I^r$ is a stabilizer group, it is abelian and closed under multiplication, so $P_uP_u'P_u''\in S_I^r$.
    
    On the other hand, $M_3$ cyclically permutes single-qubit Paulis without introducing additional signs. Therefore, on any qubit where $T_u$ acts nontrivially (as $X$, $Y$, or $Z$), the product of its three $M_3$-orbit elements is $X\cdot Y\cdot Z = i\mathbbm{1}$, and on any qubit where $T_u$ acts trivially the product is $\mathbbm{1}$. Thus,
    \begin{equation}
        P_uP_u'P_u'' = (-1)^{r(u)}i^{swt(u)}\mathbbm{1}.
    \end{equation}
    But a stabilizer group cannot contain $-\mathbbm{1}$, only $\mathbbm{1}$, and so $(-1)^{r(u)}i^{swt(u)}=1$. This forces $i^{swt(u)}\in\{\pm1\}$, so $swt(u)$ is even, and furthermore $(-1)^{r(u)}=+1$ if $swt(u)\equiv0\pmod4$ and $(-1)^{r(u)}=-1$ if $swt(u)\equiv2\pmod4$.
\end{proof}

Rall's rule implies that the signed weight enumerator $W_I^r$ is determined completely by the (unsigned) weight distribution of $I$, and in particular, we can see that if $(I,r)$ is an $M_3$ code, then $A_w=0$ for all odd $w$, and
\begin{equation}\label{Equation:SignedEnumeratorFromSimple}
    W_I^r(x,y)=\sum\limits_{a\in I}(-1)^{r(a)}x^{n-swt(a)}y^{swt(a)}
    =\sum_{w=0}^n A_w(-1)^{w/2}x^{n-w}y^w.
\end{equation}
Equivalently, $W_I^r(x,y)=W_I(x,iy)$~\cite{KalraPrakash2025}, where the right-hand side is real-valued because only even weights occur. In particular, we can rewrite the success probability as
\begin{equation}\label{Equation:TsSuccessProbSimple}
    P_s(\epsilon)=\frac{1}{2^{n-1}}W_I^r\left(1,t(\epsilon)\right)
    =\frac{1}{2^{n-1}}\sum_{w=0}^n A_w(-1)^{w/2}t(\epsilon)^w
    =\frac{1}{2^{n-1}}W_I\left(1,it(\epsilon)\right).
\end{equation}

With the assumption of a transversal $M_3$ gate, a similar simplification is possible for the post-distillation state $\rho'$. We can choose the logical Pauli operators $\bar{x},\bar{y},\bar{z}\in I^\perp\setminus I$ such that the corresponding logical Pauli operators are cyclically permuted by the transversal $M_3$ gate, i.e., $M_3^{\otimes n}T_{\bar{x}}(M_3^{\otimes n})^\dagger=T_{\bar{y}}$, $M_3^{\otimes n}T_{\bar{y}}(M_3^{\otimes n})^\dagger=T_{\bar{z}}$, and $M_3^{\otimes n}T_{\bar{z}}(M_3^{\otimes n})^\dagger=T_{\bar{x}}$. Then the Pauli expectation values of the post-distillation state are related by
\begin{align}
    \Tr(T_y\rho') =& \Tr(M_3T_xM_3^\dagger\rho')\\
    =& \frac{1}{P_s(\epsilon)}\Tr(M_3^{\otimes n}T_{\bar{x}}{M_3^\dagger}^{\otimes n}\Pi_I^r\rho_T(\epsilon)^{\otimes n})\\
    =& \frac{1}{P_s(\epsilon)}\Tr(T_{\bar{x}}\Pi_I^r(M_3^\dagger\rho_T(\epsilon)M_3)^{\otimes n})\\
    =& \frac{1}{P_s(\epsilon)}\Tr(T_{\bar{x}}\Pi_I^r\rho_T(\epsilon)^{\otimes n})=\Tr(T_x\rho').
\end{align}
Here in the third line we use the cyclic property of the trace and the fact that since $M_3^{\otimes n}$ is a logical gate, it must preserve the code space, and so it commutes with the projector $\Pi_I^r$. Similarly, $\Tr(T_z\rho')=\Tr(T_y\rho')$. Therefore, no twirling is necessary as the post-distillation state $\rho'$ is already on the line segment Eq.~\eqref{Equation:TStateLine}. Every element of $I^\perp\setminus I$ lies in exactly one of the three cosets $\bar{x}+I$, $\bar{y}+I$, or $\bar{z}+I$, and we get a rule similar to Rall's rule for the signs of the elements of these cosets.

\begin{Lemma}[\cite{KalraPrakash2025}]\label{Lemma:OddCosetRule}
    Let $(I,r)$ be an $[\![n,1]\!]_2$ $M_3$ code, and choose $\bar{x},\bar{y},\bar{z}\in I^\perp\setminus I$ as above. For any $u\in I^\perp\setminus I$, define the signed Pauli operator $P_u$ by writing $u=\bar{v}+w$ with $\bar{v}\in\{\bar{x},\bar{y},\bar{z}\}$ and $w\in I$, and setting
    \begin{equation}
        P_u:=(-1)^{r(w)+\beta(\bar{v},w)}T_u.
    \end{equation}
    Then every $u\in I^\perp\setminus I$ has odd symplectic weight, and moreover, there exists a sign $\lambda\in\{\pm1\}$, independent of $u$, such that
    \begin{equation}\label{Equation:OddCosetSignRule}
        (-1)^{r(w)+\beta(\bar{v},w)}=\lambda(-1)^{(swt(u)-1)/2}.
    \end{equation}
\end{Lemma}

\begin{proof}[Proof of Lemma~\ref{Lemma:OddCosetRule}]
    Fix a $u\in I^\perp\setminus I$. We write $u=\bar{v}+w$ with $\bar{v}\in\{\bar{x},\bar{y},\bar{z}\}$ and $w\in I$ and define the corresponding signed Pauli operators $P_u=(-1)^{r(w)+\beta(\bar{v},w)}T_u$. Define its $M_3$ conjugates as
    \begin{equation}
        P_u'=M_3^{\otimes n}P_u(M_3^{\otimes n})^\dagger \quad\text{and}\quad P_u''=(M_3^{\otimes n})^\dagger P_uM_3^{\otimes n}.
    \end{equation}
    These three operators act on the code space as the three nontrivial logical Paulis in the order $\bar{X}\bar{Y}\bar{Z}$. Therefore, on the code space their product is
    \begin{equation}
        P_uP_u'P_u''\Pi_I^r=\lambda i\Pi_I^r
    \end{equation}
    where $\lambda\in\{\pm1\}$ is a fixed constant depending on the sign convention Eq.~\eqref{Equation:PauliOperators} for the logical Pauli operators.

    On the other hand, $M_3$ cyclically permutes the single-qubit Paulis without introducing extra signs, so on each qubit where $T_u$ acts nontrivially, the product $P_uP_u'P_u''$ contributes a factor of $i\mathbbm{1}$, and on each qubit where $T_u$ acts trivially, the contribution is $\mathbbm{1}$. Therefore,
    \begin{equation}
        P_uP_u'P_u''=(-1)^{r(w)+\beta(\bar{v},w)}i^{swt(u)}\mathbbm{1}.
    \end{equation}
    Multiplying this expression by $\Pi_I^r$ and comparing against the expression above, we see that this forces $swt(u)$ to be odd, and
    \begin{equation}
        (-1)^{r(w)+\beta(\bar{v},w)}=\lambda(-1)^{(swt(u)-1)/2},
    \end{equation}
    which proves the claim.
\end{proof}

Rall's rule shows that every word in $I$ has even symplectic weight, while Lemma~\ref{Lemma:OddCosetRule} shows that every word in $I^\perp\setminus I$ has odd symplectic weight. Therefore, if $\{B_w\;|\;w=0,\dots,n\}$ is the symplectic weight distribution of $I^\perp$, i.e.,
\begin{equation}
    W_{I^\perp}(x,y)=\sum_{w=0}^n B_w x^{n-w}y^w,
\end{equation}
then $B_w=A_w$ for even $w$, and the odd-weight terms come entirely from $I^\perp\setminus I$. Using Lemma~\ref{Lemma:OddCosetRule} for the numerator of Eq.~\eqref{Equation:TPostStateSignedWeights}, we get
\begin{equation}\label{Equation:OddPartRule}
    W_{I^\perp}^r\left(1,t(\epsilon)\right)-W_I^r\left(1,t(\epsilon)\right) = \lambda\sum_{\substack{w=0\\\text{$w$ odd}}}^n B_w(-1)^{(w-1)/2}t(\epsilon)^w = -i\lambda\left(W_{I^\perp}\left(1,it(\epsilon)\right)-W_I\left(1,it(\epsilon)\right)\right).
\end{equation}
This gives us three equivalent simplified expression for the post-distillation noise parameter,
\begin{align}\label{Equation:TsUpdateSimple}
    \epsilon'(\epsilon) =& 1+\frac{i\lambda}{\sqrt{3}}\frac{W_{I^\perp}(1,it(\epsilon))-W_I(1,it(\epsilon))}{W_I(1,it(\epsilon))}\\
    =&  1-\frac{\lambda}{\sqrt{3}}\frac{\sum_{\text{$w$ odd}}B_w(-1)^{(w-1)/2}t(\epsilon)^w}{\sum_{\text{$w$ even}}A_w(-1)^{w/2}t(\epsilon)^w}\\
    =& 1-\frac{\lambda}{\sqrt{3}}\frac{\Im W_{I^\perp}(1,it(\epsilon))}{W_I(1,it(\epsilon))}.
\end{align}
Thus, up to the choice of $\lambda\in\{\pm1\}$ (which essentially amounts to choosing $\ket{T}$ or $\ket{T^\perp}$ as your target magic state~\cite{KalraPrakash2025}), for $M_3$ codes the distillation map $\epsilon\mapsto\epsilon'(\epsilon)$ and the success probability $P_s(\epsilon)$ are determined entirely by the weight distributions $A_w$ and $B_w$, or equivalently, by the simple (unsigned) weight enumerators. See Refs.~\cite{Rall2017,KalraPrakash2025} for a more thorough derivation of these expressions for $\epsilon'(\epsilon)$ and $P_s(\epsilon)$.

\subsubsection{Distilling \texorpdfstring{$T$}{T} states with quantum QR codes}\label{Section:QQRTDistillation}

Throughout this section, let $C$ be a $[p,(p-1)/2]_4$ Hermitian self-orthogonal expurgated quadratic residue code, and let $I=\iota^{-1}(C)$ be the corresponding $[\![p,1]\!]_2$ quantum quadratic residue code.

We want to apply the results of the previous section to analyzing the performance of quantum quadratic residue codes for distilling $T$ states. Any stabilizer code on $n$ qubits corresponding to an isotropic subspace of dimension less than $n$ that is $\mathbb{F}_4$-linear under the map $\iota$ is compatible with a transversal $M_3$ gate~\cite{DasuBurton2025}. Since the expurgated quadratic residue code $C$ is linear, this applies to the corresponding quantum quadratic residue code $I$. That is, there exists a choice of signs $r:I\to\mathbb{F}_2$ such that the code $(I,r)$ has a transversal $M_3$ gate. This choice of signs allows us to use the simplified expressions Eq.~\eqref{Equation:TsSuccessProbSimple} and Eq.~\eqref{Equation:TsUpdateSimple} for the success probability $P_s(\epsilon)$ of the protocol and the post-distillation noise parameter $\epsilon'(\epsilon)$.

For a quantum quadratic residue code $I$, the correspondence of Section~\ref{Section:CodeCorrespondence} identifies $I\subset\mathbb{F}_2^n\times\mathbb{F}_2^n$ with an expurgated quadratic residue code $C\subset\mathbb{F}_4^n$, and with Lemma~\ref{Lemma:MTHEqualsHermitian}, the symplectic complement $I^\perp$ is identified with the Hermitian dual $C^{\perp_H}$. Under this identification, $W_I$ and $W_{I^\perp}$ are exactly the Hamming weight enumerators $W_C$ and $W_{C^{\perp_H}}$, the latter of which can be obtained from $W_C$ by the MacWilliams identity Eq.~\eqref{Equation:MacWilliamsIdentity}.

All that remains is to compute the weight enumerators of the expurgated quadratic residue codes. A direct enumeration of the weights of vectors in $C$ is infeasible for the lengths of interest since the cardinality of a code, $|C|=2^{p-1}$, is exponential in the code length. In order to compute the Hamming weight enumerators of expurgated quadratic residue codes over $\mathbb{F}_4$, we use Prange's theorem and classical methods~\cite{NebeSloane2006}.

The expurgated quadratic residue code $C$ is obtained by shortening the extended quadratic residue code $\tilde{C}$ of length $n=p+1$ on any coordinate. Therefore, $W_C$ can be obtained from $W_{\tilde{C}}$ via Prange's theorem, Eq.~\eqref{Equation:PrangeWeights}. Since $C^{\perp_H}$ is the augmented quadratic residue code of length $p$ (see the proof of Theorem~\ref{Theorem:QQRDistance}), $C^{\perp_H}$ is obtained by puncturing $\tilde{C}$ on any coordinate, the weight enumerator $W_{C^{\perp_H}}$ can also be obtained from $W_{\tilde{C}}$ by Prange's theorem. Using the partial derivative form of Prange's theorem, Eqs.~\eqref{Equation:PrangeDerivative1} and~\eqref{Equation:PrangeDerivative2}, we get an even simpler expression for $\epsilon'(\epsilon)$, namely,
\begin{equation}\label{Equation:TsUpdateDerivatives}
    \epsilon'(\epsilon)=1+\frac{i\lambda}{\sqrt{3}}\frac{\frac{\partial}{\partial y}W_{\tilde{C}}(x,y)}{\frac{\partial}{\partial x}W_{\tilde{C}}(x,y)}\Bigg|_{(x,y)=(1,it(\epsilon))}.
\end{equation}

Then we only need to determine $W_{\tilde{C}}$. The code $\tilde{C}$ is Hermitian self-dual, and through invariant theory, we know that the Hamming weight enumerators of $\mathbb{F}_4$-linear Hermitian self-dual codes live in the ring $\mathbb{C}[i_2,h_6]$ where
\begin{equation}
    \begin{cases}
        i_2(x,y)=x^2+3y^2\\
        h_6(x,y)=x^6+45x^2y^4+18y^6
    \end{cases}
\end{equation}
are the Hamming weight enumerators of the $[2,1,2]_4$ repetition code and the $[6,3,4]_4$ hexacode respectively~\cite[\S7.3.1]{NebeSloane2006}. With $n=p+1$, since $W_{\tilde{C}}(x,y)$ is homogeneous of degree $n$, we can write
\begin{equation}\label{Equation:ExtendedWEDecomposition}
    W_{\tilde{C}}(x,y)=\sum_{j=0}^{\lfloor n/6\rfloor} a_ji_2(x,y)^{\frac{n}{2}-3j}h_6(x,y)^{j}.
\end{equation}
There are $\lfloor n/6\rfloor+1$ coefficients $\{a_j\}$ in Eq.~\eqref{Equation:ExtendedWEDecomposition}. There is exactly one vector in $\tilde{C}$ of weight $0$, namely, the zero vector, and so $\tilde{A}_0=1$. This gives one linear constraint on the coefficients $\{a_j\}$, leaving $\lfloor n/6\rfloor$ independent parameters. Computing $\lfloor n/6\rfloor$ additional weight multiplicities, for example $\tilde{A}_2,\tilde{A}_4,\dots,\tilde{A}_{2\lfloor n/6\rfloor}$, determines $W_{\tilde{C}}$ uniquely by solving a linear system for the coefficients $\{a_j\}$.

To obtain these low-weight multiplicities $\tilde{A}_w$ without enumerating all codewords of $\tilde{C}$, we use the algorithm of Ref.~\cite{GaboritWassermann2005}, which computes the number of codewords of specified small weights from a complementary pair of generator matrices.\footnote{We have a Julia implementation of this algorithm that runs on Nvidia GPUs available at \href{https://github.com/mzurel/WeightDistributionsGPU.jl}{github.com/mzurel/WeightDistributionsGPU.jl}.} Using this strategy, we computed the weight enumerators of all quantum quadratic residue codes on qubits up to length $71$. The weight enumerators are listed in Appendix~\ref{Section:WeightEnumeratorTables}.

\begin{table}
    \centering
    \begin{tabular}{|c|c|}
        \hline
        Code & Threshold \\
        \hline\hline
        $[5,2,4]_4$ & $\approx0.34535$ \\
        \hline
        $[23,11,8]_4$ & $\approx0.32237$ \\
        \hline
        $[29,14,12]_4$ & $\approx0.24190$ \\
        \hline
        $[47,23,12]_4$ & $\approx0.05050$ \\
        \hline
        $[53,26,16]_4$ & $\approx0.27343$ \\
        \hline
        $[71,35,12]_4$ & $\approx0.00664$ \\
        \hline
    \end{tabular}
    \caption{Noise thresholds of magic state distillation protocols for qubit $T$ states based on quantum quadratic residue codes.}
    \label{Table:TStateThresholds}
\end{table}

Using these weight enumerators, and the relation between weight enumerators and the performance of $M_3$ codes for distilling $T$ states described in Section~\ref{Section:TDistillationAndWeights}, we find that the qubit quantum quadratic residue codes of length $5$, $23$, $29$, $47$, $53$, and $71$ distill $T$ states with nontrivial thresholds. These thresholds are listed in Table~\ref{Table:TStateThresholds}. The distillation functions $\epsilon'(\epsilon)$ for these codes are shown in Figure~\ref{Figure:TStateDistillation}, and their success probabilities $P_s(\epsilon)$ are shown in Figure~\ref{Figure:TStateSuccessProb}. An additional plot showing the functions $\epsilon'(\epsilon)$ for all qubit quantum QR codes up to length $71$, not just those with non-trivial thresholds, is shown in Appendix~\ref{Section:AllPlots}.

\begin{figure}
    \centering
    \includegraphics[width=0.85\linewidth]{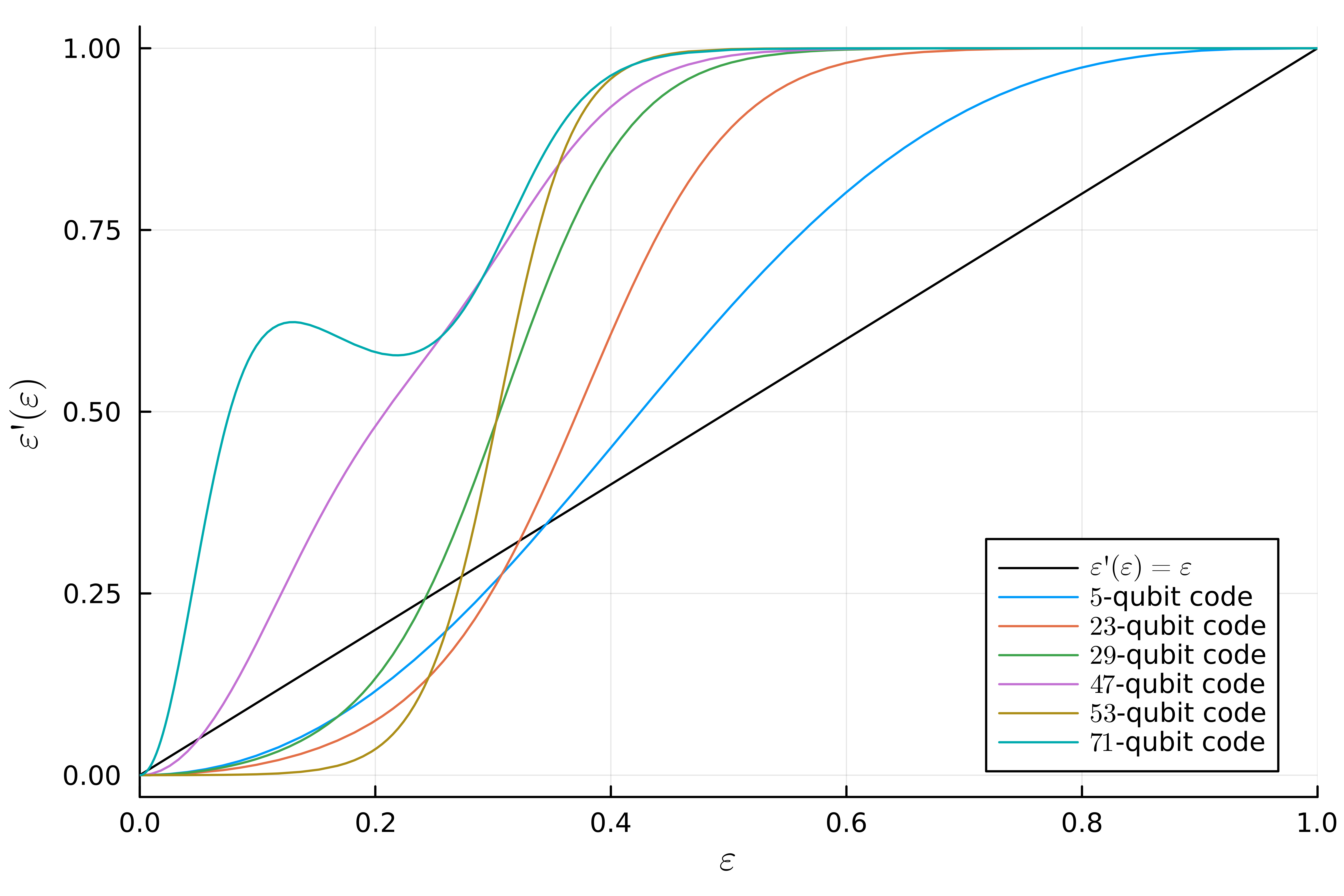}
    \caption{Post-distillation noise $\epsilon'$ as a function of pre-distillation noise $\epsilon$ for quantum quadratic residue codes that distill qubit $T$ states. The $5$-, $23$-, $29$-, $47$-, $53$-, and $71$-qubit codes have non-trivial distillation thresholds.}
    \label{Figure:TStateDistillation}
\end{figure}

\begin{figure}
    \centering
    \includegraphics[width=0.85\linewidth]{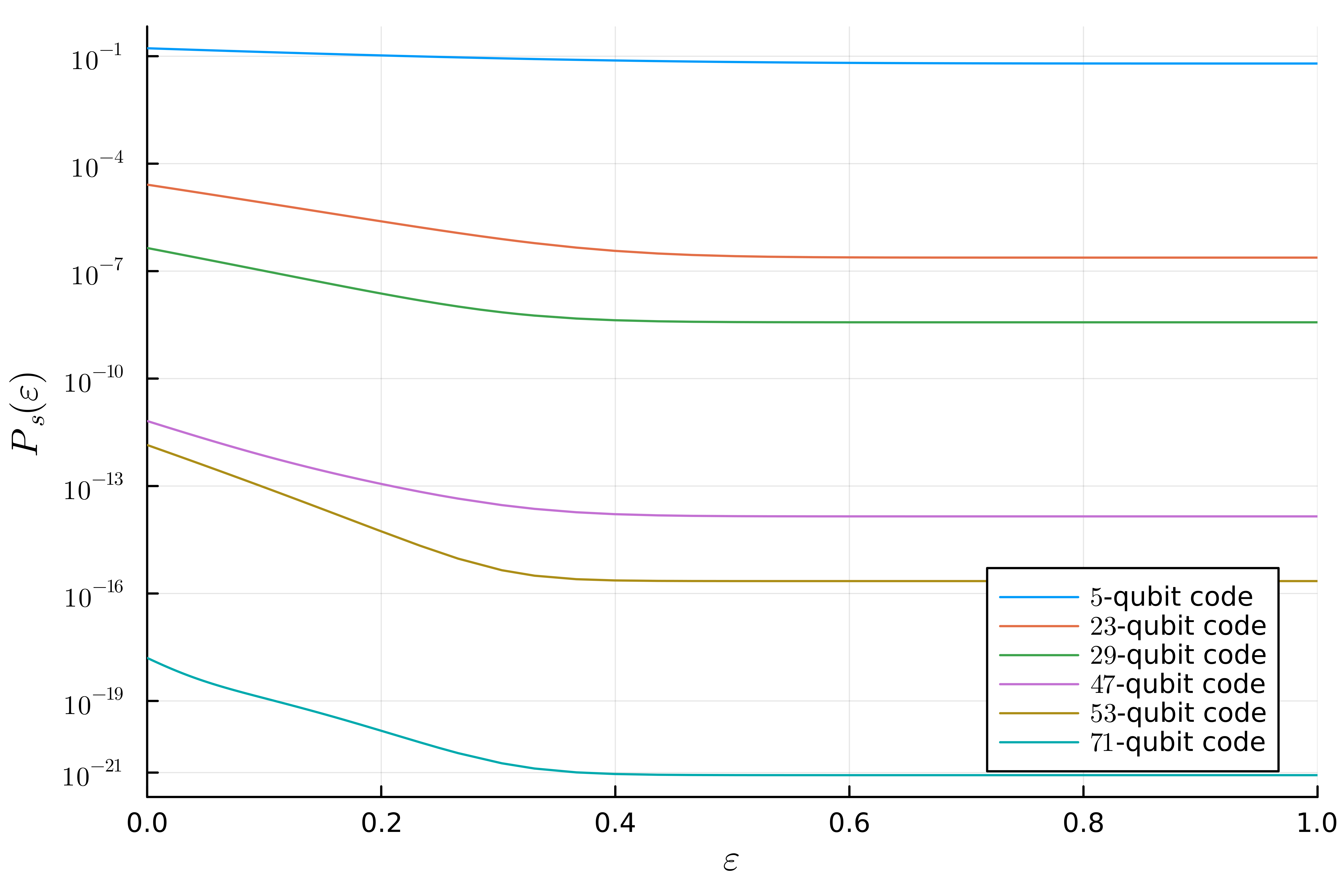}
    \caption{Success probability of magic state distillation protocols for qubit $T$ states based on quantum quadratic residue codes.}
    \label{Figure:TStateSuccessProb}
\end{figure}

The noise threshold is not the only figure of merit for a magic state distillation protocol. In Appendix~\ref{Section:ResourceAnalysis}, we compare the resource consumption of a single round of these protocols based on quantum QR codes against the baseline of protocols formed by concatenating the $5$-qubit protocol, both in terms of the number of $T$ states consumed by a successful fully post-selected run of the protocol and the expected number of raw input $T$ states consumed taking into account the success probability of the protocols.

\subsubsection{Infinitely many quantum QR codes distill \texorpdfstring{$T$}{T} states}

Table~\ref{Table:TStateThresholds} shows that the qubit quantum quadratic residue codes of length $5$, $23$, $29$, $47$, $53$, and $71$ distill $T$ states with non-trivial thresholds. Here we show that there are infinitely many code lengths with non-trivial threshold for $T$ state distillation. This is the result of the following theorem, and corollary.

\begin{Theorem}\label{Theorem:InfiniteTDistillation}
    Let $p$ be a prime with $p\equiv5\text{ or }23\mod24$. Let $C\subset\mathbb{F}_4^p$ be the expurgated quadratic residue code of length $p$ and let $I=\iota^{-1}(C)$ be the corresponding quantum quadratic residue code. Then there is a sign $\lambda$ in Eq.~\eqref{Equation:TsUpdateSimple} such that the map $\epsilon'(\epsilon)$ for the corresponding $T$ state distillation protocol satisfies
    \begin{equation}
        \epsilon'(\epsilon)=O(\epsilon^2),\quad(\epsilon\to0).
    \end{equation}
    In other words, the quantum quadratic residue code of length $p$ distills $T$ states with a non-trivial threshold. Furthermore, these protocols have a positive success probability for sufficiently small $\epsilon$.
\end{Theorem}

\begin{proof}[Proof of Theorem~\ref{Theorem:InfiniteTDistillation}]
    The congruence $p\equiv5\text{ or }23\mod24$ implies $p\equiv5\text{ or }7\mod8$, and so a quantum quadratic residue code of this length exists by Corollary~\ref{Corollary:F4HermitianSelfOrthogonality}. It also implies that $p\equiv2\mod3$.

    Let $\tilde{C}\subset\mathbb{F}_4^{p+1}$ be the extended quadratic residue code. The expurgated quadratic residue code $C$ is obtained from $\tilde{C}$ by shortening on any coordinate, and the corresponding quantum quadratic residue code is obtained via $\iota^{-1}(C)$. We use the expression Eq.~\eqref{Equation:TsUpdateDerivatives} for the post-distillation noise parameter $\epsilon'(\epsilon)$ of the distillation protocol based on the code $C$. This expression involves evaluating partial derivatives of $W_{\tilde{C}}(x,y)$ at $(x,y)=(1,it(\epsilon))$ where $t(\epsilon)=(1-\epsilon)/\sqrt{3}$.
    
    As stated above, $\tilde{C}$ is Hermitian self-dual with Hamming weight enumerator in the ring $\mathbb{C}[i_2,h_6]$. Let $n=p+1=6m$ (so $m=(p+1)/6\in\mathbb{Z}$). Since $\deg(i_2)=2$ and $\deg(h_6)=6$, the weight enumerator of $\tilde{C}$ must have the form
    \begin{equation}\label{Equation:CTildeInvariantExpansion}
        W_{\tilde{C}}(x,y)=\sum_{j=0}^m a_ji_2(x,y)^{3j}h_6(x,y)^{m-j}= a_0h_6(x,y)^m + i_2(x,y)^3R(x,y)
    \end{equation}
    for some coefficients $a_j\in\mathbb{C}$ and some homogeneous polynomial $R\in\mathbb{C}[i_2,h_6]$.

    Let $u(\epsilon) := i_2(1,it(\epsilon))=1-3t^2 = 2\epsilon-\epsilon^2$ and $v(\epsilon):= h_6(1,it(\epsilon))$. Differentiating Eq.~\eqref{Equation:CTildeInvariantExpansion} for $W_{\tilde{C}}(x,y)$, we have
    \begin{equation}
        \begin{cases}
            \frac{\partial W_{\tilde{C}}(x,y)}{\partial x}\bigg|_{(x,y)=(1,it(\epsilon))}=a_0mv(\epsilon)^{m-1}\frac{\partial h_6(x,y)}{\partial x}\bigg|_{(x,y)=(1,it(\epsilon))}+O(u(\epsilon)^2),\vspace{2mm}\\
            \frac{\partial W_{\tilde{C}}(x,y)}{\partial y}\bigg|_{(x,y)=(1,it(\epsilon))}=a_0mv(\epsilon)^{m-1}\frac{\partial h_6(x,y)}{\partial y}\bigg|_{(x,y)=(1,it(\epsilon))}+O(u(\epsilon)^2).
        \end{cases}
    \end{equation}
    Here the error terms are $O(u^2)$ because every term coming from differentiating $i_2^3R$ carries at least a factor $i_2^2$. Therefore, for $\epsilon\ll1$,
    \begin{equation}\label{Equation:CTildeInvariantDerivativeExpansion}
        \frac{\frac{\partial}{\partial y}W_{\tilde{C}}(x,y)}{\frac{\partial}{\partial x}W_{\tilde{C}}(x,y)}\Bigg|_{(x,y)=(1,it(\epsilon))}=\frac{\frac{\partial}{\partial y}h_6(x,y)}{\frac{\partial}{\partial x}h_6(x,y)}\Bigg|_{(x,y)=(1,it(\epsilon))} + O(u(\epsilon)^2).
    \end{equation}
    Computing these partial derivatives, we find
    \begin{equation}
        \Phi(t):=\frac{\frac{\partial}{\partial y}h_6(x,y)}{\frac{\partial}{\partial x}h_6(x,y)}\Bigg|_{(x,y)=(1,it(\epsilon))}=i\frac{-30t^3+18t^5}{1+15t^4}.
    \end{equation}
    At $\epsilon=0$, we have $t=1/\sqrt{3}$ and so $\Phi(1/\sqrt{3})=-i\sqrt{3}$. Furthermore,
    \begin{equation}
        \frac{d\Phi}{dt}=i\frac{90t^2(t^2+1)^2(3t^2-1)}{(1+15t^4)^2},
    \end{equation}
    so at $t=1/\sqrt{3}$, $d\Phi/dt=0$. Since $t-1/\sqrt{3}=-\epsilon/\sqrt{3}$, this implies $\Phi(t)=\Phi(1/\sqrt{3})+O(\epsilon^2)=-i\sqrt{3}+O(\epsilon^2)$, and therefore
    \begin{equation}
        \frac{\frac{\partial}{\partial y}W_{\tilde{C}}(x,y)}{\frac{\partial}{\partial x}W_{\tilde{C}}(x,y)}\Bigg|_{(x,y)=(1,it(\epsilon))} = -i\sqrt{3} + O(\epsilon^2).
    \end{equation}
    Finally, we plug this into Eq.~\eqref{Equation:TsUpdateDerivatives}. With the choice of $\lambda$ for which $\epsilon'(0)=0$, the constant term cancels and we obtain $\epsilon'(\epsilon)=O(\epsilon^2)$. Thus, $\epsilon'(\epsilon)/\epsilon\to 0$ as $\epsilon\to0$, so there exists $\epsilon_{\mathrm{th}}(p)>0$ such that $\epsilon'(\epsilon)<\epsilon$ for all $\epsilon\in(0,\epsilon_{\mathrm{th}}(p))$, i.e. a non-trivial threshold.

    The last thing to show is that these protocols maintain a nonzero success probability for sufficiently small $\epsilon$. As shown above, the success probability is proportional to $W_C(1,it(\epsilon))$. By Prange's theorem, we have
    \begin{equation}
        W_C(x,y)=\frac{1}{n}\frac{\partial}{\partial x}W_{\tilde{C}}(x,y)
    \end{equation}
    and so it suffices to show that $\partial W_{\tilde{C}}(x,y)/\partial x$ evaluated at $(x,y)=(1,it(\epsilon))$ is nonzero in the limit $\epsilon\to0$, or equivalently, that
    \begin{equation}
        D:=\frac{\partial W_{\tilde{C}}(x,y)}{\partial x}\bigg|_{(x,y)=\left(1,i/\sqrt{3}\right)}\neq0.
    \end{equation}
    Note that this also justifies the expansion Eq.~\eqref{Equation:CTildeInvariantDerivativeExpansion} by showing that the denominator is nonzero in the small $\epsilon$ limit.

    Since $\tilde{C}$ is Hermitian self-dual, every code word of $\tilde{C}$ has even Hamming weight~\cite{MacWilliamsSloane1977}. Thus, evaluating at $(x,y)=(1,i/\sqrt{3})$ introduces powers of $-1/3$ for terms with nonzero powers of the $y$ coordinate, that is, $D\in\mathbb{Z}[1/3]$. Then since $p\neq3$, reduction modulo $p$ is well-defined in $\mathbb{Z}[1/3]$, and so it suffices to show that $D\not\equiv0\pmod p$.

    Let $\sigma$ be the cyclic shift of the $p$ coordinates of $\tilde{C}$ that are shared with $C$ which fixes the extension coordinate. This is an automorphism of $\tilde{C}$ since quadratic residue codes are cyclic codes. Then, since $p$ is prime, every orbit of codewords of $\tilde{C}$ under the action of the group generated by $\sigma$ has size $1$ or $p$, and every word in an orbit has the same weight. Therefore, orbits of code words that are not fixed by $\sigma$ contribute to the coefficients of $W_{\tilde{C}}(x,y)$ in multiples of $p$, and thus, $W_{\tilde{C}}(x,y)$ is congruent modulo $p$ to the weight enumerator $W_{\tilde{C}^\sigma}(x,y)$ of the subcode $\tilde{C}^\sigma\subset\tilde{C}$ formed from code words of $\tilde{C}$ that are fixed by $\sigma$. Differentiating preserves congruences modulo $p$, so we also have
    \begin{equation}
        \frac{\partial W_{\tilde{C}}(x,y)}{\partial x}\equiv\frac{\partial W_{\tilde{C}^\sigma}(x,y)}{\partial x}\pmod p,
    \end{equation}
    and so it suffices to show that the partial derivative of $W_{\tilde{C}^\sigma}(x,y)$ is nonzero at $(x,y)=(1,i/\sqrt{3})$.

    Code words of $\tilde{C}$ that are fixed by $\sigma$ have the form $(a,\dots,a,b)$ with $a,b\in\mathbb{F}_4$. Again, since the code words of $\tilde{C}$ all have even Hamming weight, there are no words of weight $1$ or $p$, and so the code words of $\tilde{C}$ that are fixed by $\sigma$ all have weight $0$ or $n=p+1$. Thus,
    \begin{equation}
        \frac{\partial}{\partial x}W_{\tilde{C}^\sigma}=nx^{n-1}.
    \end{equation}
    Evaluating this at $x=1$ gives
    \begin{equation}
        \frac{\partial}{\partial x}W_{\tilde{C}}(x,y)\bigg|_{(x,y)=(1,i/\sqrt{3})}\equiv n\equiv 1\pmod p.
    \end{equation}
    Therefore, the success probability of these protocols is nonzero in the small $\epsilon$ limit, and by continuity, it remains nonzero for all sufficiently small $\epsilon>0$.
\end{proof}

\begin{Corollary}\label{Corollary:InfiniteTDistillation}
    There are infinitely many quantum quadratic residue codes that distill $T$ states with non-trivial threshold.
\end{Corollary}

\begin{proof}[Proof of Corollary~\ref{Corollary:InfiniteTDistillation}]
    There exists a quantum quadratic residue code for every prime length $p\equiv5\text{ or }7\mod8$. By Theorem~\ref{Theorem:InfiniteTDistillation}, they are guaranteed to distill $T$ states whenever $p\equiv5\text{ or }23\mod24$. By Dirichlet's theorem, there are infinitely many prime lengths $p$ satisfying this condition.
\end{proof}

\subsection{Qutrit Strange state distillation}

As evidenced by the $11$-qutrit Golay code example given in Section~\ref{Section:MSDSpecialCases}, some quantum quadratic residue codes are useful for distilling qutrit Strange states~\cite{VeitchEmerson2014,JainPrakash2020,Prakash2020}. The Strange state is the single-qutrit state $\ket{S}:=(\ket{1}-\ket{2})/\sqrt{2}$. Prakash and Singhal showed that the performance of a magic state distillation protocol for distilling Strange states is determined by the weight distribution of the corresponding code, and relies on a relation between the weight enumerator of the code and the Wigner function of the state~\cite{PrakashSinghal2026}. For completeness, we include a derivation of this fact here.

\subsubsection{Wigner functions and weight enumerators}

The phase space point operators of the Wigner function~\cite{Gross2006,Gross2008} on $n$ $d$-dimensional qudits are
\begin{equation}
    A_u=\frac{1}{d^n}\sum\limits_{v\in\mathbb{F}_d^{2n}}\omega^{[u,v]}T_v,\qquad\forall u\in\mathbb{F}_d^{2n}.
\end{equation}
These operators are Hermitian, have unit trace, and together they form an orthogonal basis for the real vector space of Hermitian operators on the $n$-qudit Hilbert space satisfying the orthogonality relation $\Tr(A_uA_v)=d^n\delta_{u,v}$. The Wigner function of a state $\rho$ is given by the coefficients of $\rho$ when expanded in this basis,
\begin{equation}\label{Equation:WignerExpansion}
    \rho=\sum\limits_{u\in\mathbb{F}_d^{2n}}W_\rho(u)A_u,
\end{equation}
which, by orthogonality, can be extracted by $W_\rho(u)=\Tr(\rho A_u)/d^n$. Since the phase space point operators are Hermitian, the Wigner function is real, and taking a trace of Eq.~\eqref{Equation:WignerExpansion}, we get the normalization condition, $1=\sum_{u\in\mathbb{F}_d^{2n}}W_\rho(u)$. A summary of how the Wigner function transforms under the stabilizer operations is given in Appendix~\ref{Section:PhaseSpaceMethods}.

The Wigner function of the Strange state has a particularly simple form. As a density matrix, the pure Strange state is
\begin{equation}
    \ketbra{S}{S}=\frac{1}{3}\mathbbm{1}-\frac{1}{2}\left(A_0-\frac{1}{3}\mathbbm{1}\right)=-\frac{1}{3}A_0+\frac{1}{6}\sum\limits_{u\in\mathbb{F}_3^2\setminus\{0\}}A_u,
\end{equation}
and therefore has Wigner function $W_{\ketbra{S}{S}}(0)=-1/3$ and $W_{\ketbra{S}{S}}(u)=1/6$ for all $u\ne0\in\mathbb{F}_3^2$. When subject to depolarizing noise, the noisy Strange state,
\begin{equation}\label{Equation:NoisyStrangeState}
    \rho_S(\epsilon):=(1-\epsilon)\ketbra{S}{S}+\frac{\epsilon}{3}\mathbbm{1},
\end{equation}
has Wigner function
\begin{equation}\label{Equation:NoisyStrangeStateWigner}
    W_{\rho_S(\epsilon)}(u)=\begin{cases}
            \frac{4\epsilon-3}{9}\quad\text{if $u=0$,}\\
            \frac{3-\epsilon}{18}\quad\text{if $u\ne0$.}
		\end{cases}
\end{equation}

We will assume that the resource states for magic state distillation are depolarized Strange states of the form Eq.~\eqref{Equation:NoisyStrangeState}. No generality is lost in making this assumption since we can always apply a group twirl with respect to the subgroup of the Clifford group that fixes $A_0$ to get a state that lies on the line running through $\ketbra{S}{S}$ and the maximally mixed state. Concretely, let
\begin{equation}
    G_0=\{g\in\mathcal{C}\ell\;|\;gA_0g^\dagger=A_0\}\cong \mathrm{Sp}(\mathbb{F}_3^2).
\end{equation}
The group twirl
\begin{equation}\label{Equation:StrangeTwirl}
    \mathcal{T}_S(\rho)=\frac{1}{|G_0|}\sum\limits_{g\in G_0}g\rho g^\dagger
\end{equation}
is a stabilizer operation that does not change the value of the Wigner function at the origin of phase space, $W_{\mathcal{T}_S(\rho)}(0)=W_\rho(0)$, but makes the Wigner function $W_{\mathcal{T}_S(\rho)}(u)$ uniform on all other phase space points. With the normalization of the Wigner function, $\sum_uW_\rho(u)=1$, the Wigner function of any state after the twirl is determined uniquely by the value of the Wigner function at the origin, or equivalently, by the noise parameter $\epsilon$ in Eq.~\eqref{Equation:NoisyStrangeState}.

Let $x(\epsilon)=(4\epsilon-3)/9$ and $y(\epsilon)=(3-\epsilon)/18$ be the values of the Wigner function of the state $\rho_S(\epsilon)$ in Eq.~\eqref{Equation:NoisyStrangeStateWigner}. Then the initial state used for magic state distillation, an $n$-fold tensor product of noisy Strange states, can be written as
\begin{equation}
    \rho_S(\epsilon)^{\otimes n}=\left(x(\epsilon)A_0+y(\epsilon)\sum\limits_{u\in\mathbb{F}_3^{2}\setminus\{0\}}A_u\right)^{\otimes n}=\sum\limits_{u\in\mathbb{F}_3^{2n}}x(\epsilon)^{n-swt(u)}y(\epsilon)^{swt(u)}A_u,
\end{equation}
and so the Wigner function of this state is
\begin{equation}\label{Equation:nStrangeStateWignerFunction}
    W_{\rho_S(\epsilon)^{\otimes n}}(u)=x(\epsilon)^{n-swt(u)}y(\epsilon)^{swt(u)}.
\end{equation}

Fix an $n$-to-$1$ qutrit distillation protocol defined by an $[\![n,1]\!]_3$ stabilizer code $I\subset\mathbb{F}_3^{2n}$, post-selection on measurement outcomes $r:I\to\mathbb{F}_3$, and decoding Clifford $D$. The recovery map of the protocol is
\begin{equation}\label{Equation:MSDRecoveryMap}
	\mathcal{R}(A):=\Tr_{2,3,\dots,n}\left(D\Pi_I^rA\Pi_I^rD^\dagger\right).
\end{equation}
When applied to a resource state $\rho$, the success probability of the protocol is given by $\Tr(\mathcal{R}(\rho^{\otimes n}))$, and the post-distillation state looks like
\begin{equation}
	\rho'=\frac{\mathcal{R}(\rho^{\otimes n})}{\Tr(\mathcal{R}(\rho^{\otimes n}))}=\frac{\Tr_{2,3,\dots,n}(D\Pi_I^r\rho^{\otimes n}\Pi_I^rD^\dagger)}{\Tr(D\Pi_I^r\rho^{\otimes n}\Pi_I^rD^\dagger)}.
\end{equation}

Post-selecting on measurement outcomes $r(u)=0$ for all $u\in I$, with Lemma~\ref{Lemma:PhaseSpacePauliUpdate} in Appendix~\ref{Section:PhaseSpaceMethods}, we have
\begin{equation}
    \Pi_I^0A_v\Pi_I^0=\delta_{v\in I^\perp}\frac{1}{|I|}\sum\limits_{w\in v+I}A_w,
\end{equation}
and so using the expansion $\rho^{\otimes n}=\sum_{v\in\mathbb{F}_3^{2n}}W_{\rho^{\otimes n}}(v)A_v$, the subnormalized post-measurement state $\Pi_I^0\rho^{\otimes n}\Pi_I^0$ has Wigner function
\begin{equation}
    W_{\Pi_I^0\rho^{\otimes n}\Pi_I^0}(u)=\begin{cases}
        \frac{1}{|I|}\sum\limits_{v\in u+I}W_{\rho^{\otimes n}}(v)\quad & \text{if $u\in I^{\perp}$,}\\
        0\quad & \text{otherwise.}
    \end{cases}
\end{equation}
The decoding $D$ identifies the cosets $I^\perp/I\cong\mathbb{F}_3^2$ with the phase space of the first physical qutrit, and so the Wigner function of the normalized post-distillation state is
\begin{equation}\label{Equation:StabilizerReductionWignerFunction}
    W_{\rho'}(u)=\frac{\sum\limits_{v\in\bar{u}+I}W_{\rho^{\otimes n}}(v)}{\sum\limits_{u'\in\mathbb{F}_3^2}\sum\limits_{v\in\bar{u}'+I}W_{\rho^{\otimes n}}(v)}=\frac{\sum\limits_{v\in\bar{u}+I}W_{\rho^{\otimes n}}(v)}{\sum\limits_{v\in I^\perp}W_{\rho^{\otimes n}}(v)},
\end{equation}
where $\bar{u}\in I^\perp$ is a representative of the coset of $u\in\mathbb{F}_3^2$, and the denominator is determined by imposing normalization of the Wigner function, $\sum_{u\in\mathbb{F}_3^2}W_{\rho'}(u)=1$.

If we twirl by $\mathcal{T}_S$ after the protocol, the result is a depolarized Strange state $\rho_S(\epsilon')$ of the form Eq.~\eqref{Equation:NoisyStrangeState} with noise parameter $\epsilon'$. The value of $\epsilon'$ is determined by the value of the Wigner function at phase space point $0\in\mathbb{F}_3^2$, via the relation $W_{\rho_S(\epsilon')}(0)=:x(\epsilon')=(4\epsilon'-3)/9$. Solving this for $\epsilon'$, we get the expression
\begin{equation}
    \epsilon'=\frac{3}{4}\left(1+3W_{\rho_S(\epsilon')}(0)\right)=\frac{3}{4}\left(1+3\frac{\sum\nolimits_{v\in I}W_{\rho^{\otimes n}}(v)}{\sum\nolimits_{v\in I^\perp}W_{\rho^{\otimes n}}(v)}\right).
\end{equation}
If the resource state is $\rho=\rho_S(\epsilon)$, then the numerator in this expression is exactly the symplectic weight enumerator polynomial $W_I(x,y)$ of the code $I$ evaluated at $(x,y)=(x(\epsilon),y(\epsilon))$, and the denominator is the symplectic weight enumerator polynomial $W_{I^\perp}(x,y)$ of $I^\perp$ evaluated at the same point. That is, we get the expression for the post-distillation noise $\epsilon'$ as a function of the pre-distillation noise $\epsilon$,
\begin{equation}
    \epsilon'(\epsilon)=\frac{3}{4}\left(1+3\frac{W_I(x(\epsilon),y(\epsilon))}{W_{I^\perp}(x(\epsilon),y(\epsilon))}\right).
\end{equation}
Using the MacWilliams identity to write the symplectic weight enumerator of $I^\perp$ in terms of that of $I$, we get the final expression
\begin{equation}
    \epsilon'(\epsilon)=\frac{3}{4}\left(1+3^n\frac{W_I(\frac{4\epsilon-3}{9},\frac{3-\epsilon}{18})}{W_I(1,\frac{\epsilon-1}{2})}\right),
\end{equation}
where we use the identities $x(\epsilon)+8y(\epsilon)=1$ (from normalization of the Wigner function), and $x(\epsilon)-y(\epsilon)=(\epsilon-1)/2$.

The success probability of the protocol applied to resource states $\rho_S(\epsilon)$ is
\begin{equation}
    P_s(\epsilon)=\Tr(\mathcal{R}(\rho^{\otimes n}))=\sum\limits_{v\in I^\perp}W_{\rho_S(\epsilon)^{\otimes n}}(v)=W_{I^\perp}(x(\epsilon),y(\epsilon)).
\end{equation}
See Ref.~\cite{PrakashSinghal2026} for a more complete derivation of these expressions for $\epsilon'(\epsilon)$ and $P_s(\epsilon)$.

\subsubsection{Distilling Strange states with quantum QR codes}

The result of the previous section is that the performance of a magic state distillation protocol for distilling Strange states is determined entirely by the symplectic weight enumerator of the stabilizer code $I$. Via the map $\iota$ of Section~\ref{Section:CodeCorrespondence}, this weight enumerator is exactly the Hamming weight enumerator of the corresponding $\mathbb{F}_9$ code $\iota(I)$. Here we will use these results to analyze the performance of quantum quadratic residue codes for distilling Strange states. Throughout this section, let $C$ be a $[p,(p-1)/2]_9$ Hermitian self-orthogonal expurgated quadratic residue code, and let $I=\iota^{-1}(C)$ be the corresponding $[\![p,1]\!]_3$ quantum quadratic residue code.

Our strategy for computing these weight enumerators will be the same as in Section~\ref{Section:QQRTDistillation}. We use Prange's theorem to relate the weight enumerator of $C$ to that of $\tilde{C}$, the length $n=p+1$ extended quadratic residue code via Eq.~\eqref{Equation:PrangeWeights}. From invariant theory, the Hamming weight enumerators of Hermitian self-dual codes over $\mathbb{F}_9$ live in the ring
\begin{equation}
    \mathbb{C}[r_2,r_4]\oplus r_6\mathbb{C}[r_2,r_4]
\end{equation}
with syzygy
\begin{equation}
    r_6^2=\frac{3}{4}r_2^4r_4-\frac{3}{2}r_2^2r_4^2-\frac{1}{4}r_4^3-r_2^3r_6+3r_2r_4r_6
\end{equation}
where
\begin{equation}
    \begin{cases}
        r_2=x^2+8y^2\\
        r_4=x^4+32xy^3+48y^4\\
        r_6=x^6+16x^3y^3+72x^2y^4+288xy^5+352y^6
    \end{cases}
\end{equation}
are the weight enumerators of certain length $2$, $4$, and $6$ codes~\cite[\S5.8]{NebeSloane2006}. Therefore, the weight enumerator $W_{\tilde{C}}(x,y)$ can be written as
\begin{equation}
    W_{\tilde{C}}(x,y)=\sum\limits_{j=0}^{\lfloor n/4\rfloor}a_jr_2(x,y)^{\frac{n}{2}-2j}r_4(x,y)^j + r_6(x,y)\sum\limits_{j=0}^{\lfloor(n-6)/4\rfloor}b_jr_2(x,y)^{\frac{n-6}{2}-2j}r_4(x,y)^j.
\end{equation}
The coefficients $a_j$ and $b_j$ are obtained by imposing a number of linear constraints obtained by computing the numbers of vectors of certain low weights. These numbers are obtained using the algorithm of Ref.~\cite{GaboritWassermann2005}.\footnote{A Julia implementation of this algorithm for codes over $\mathbb{F}_9$ that runs on Nvidia GPUs is also available at \href{https://github.com/mzurel/WeightDistributionsGPU.jl}{github.com/mzurel/WeightDistributionsGPU.jl}}

We computed the weight enumerators for all qutrit quantum quadratic residue codes up to length $47$. The weight enumerators are listed in Appendix~\ref{Section:WeightEnumeratorTables}. Using the weight enumerators found, together with the expression for the distillation function $\epsilon'(\epsilon)$ in terms of the weight enumerators of the corresponding code, we find that the length $11$, $17$, $23$, $41$, and $47$ qutrit quantum quadratic residue codes distill Strange states with non-trivial thresholds. The distillation functions $\epsilon'(\epsilon)$ for these codes are plotted in Figure~\ref{Figure:StrangeStateDistillation}, and the thresholds of these codes are listed in Table~\ref{Table:StrangeStateThresholds}. A plot showing the success probability of these protocols is given in Figure~\ref{Figure:StrangeStateSuccessProb}. Additional figures showing the functions $\epsilon'(\epsilon)$ for all codes up to length $47$ (not just those with non-trivial thresholds) are given in Appendix~\ref{Section:AllPlots}.

\begin{figure}
    \centering
    \includegraphics[width=0.85\linewidth]{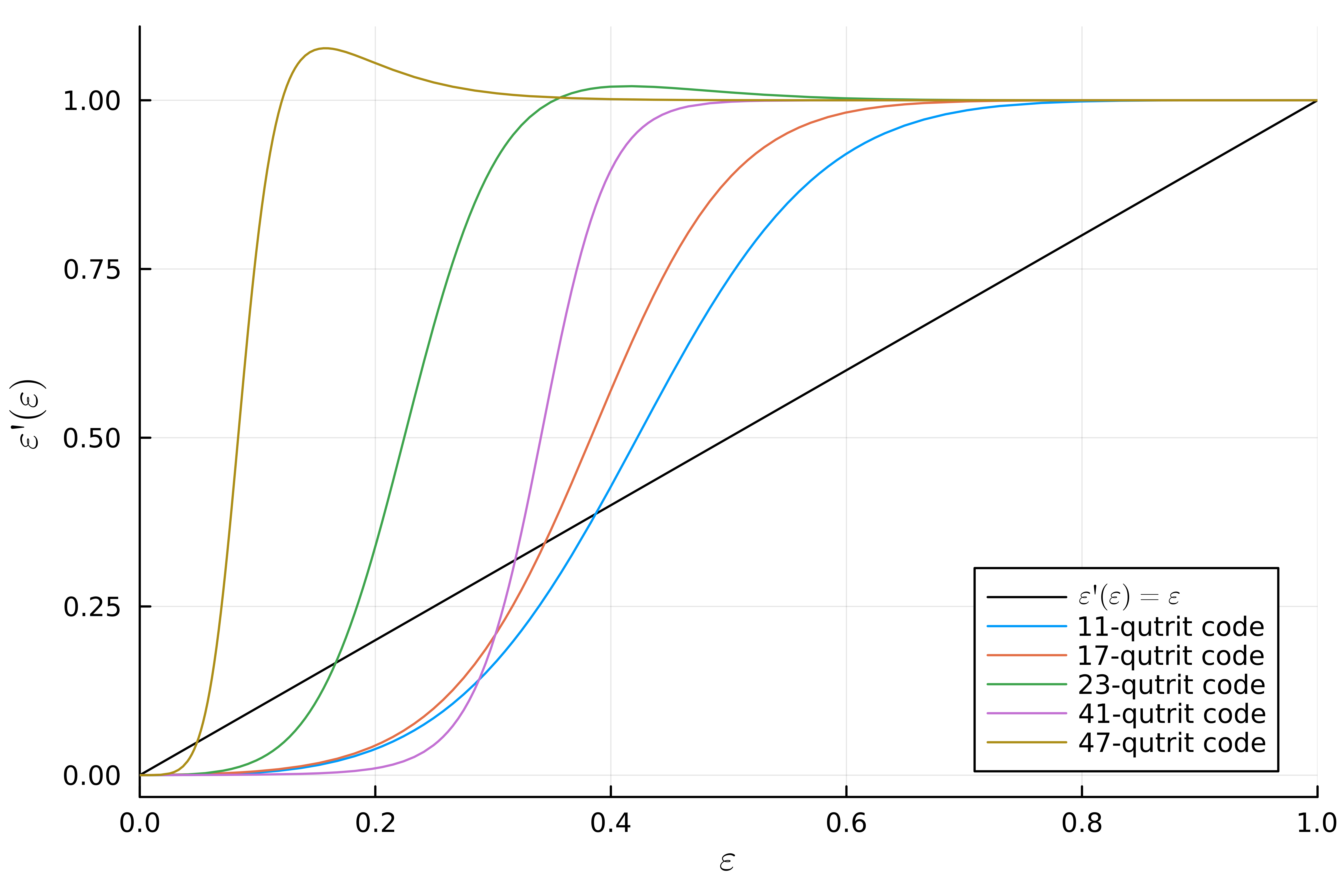}
    \caption{Post-distillation noise as a function of pre-distillation noise for codes that distill qutrit Strange states. The $11$-qutrit Golay code, as well as the $17$-, $23$-, $41$-, and $47$-qutrit codes have non-trivial distillation thresholds.}
    \label{Figure:StrangeStateDistillation}
\end{figure}

\begin{figure}
    \centering
    \includegraphics[width=0.85\linewidth]{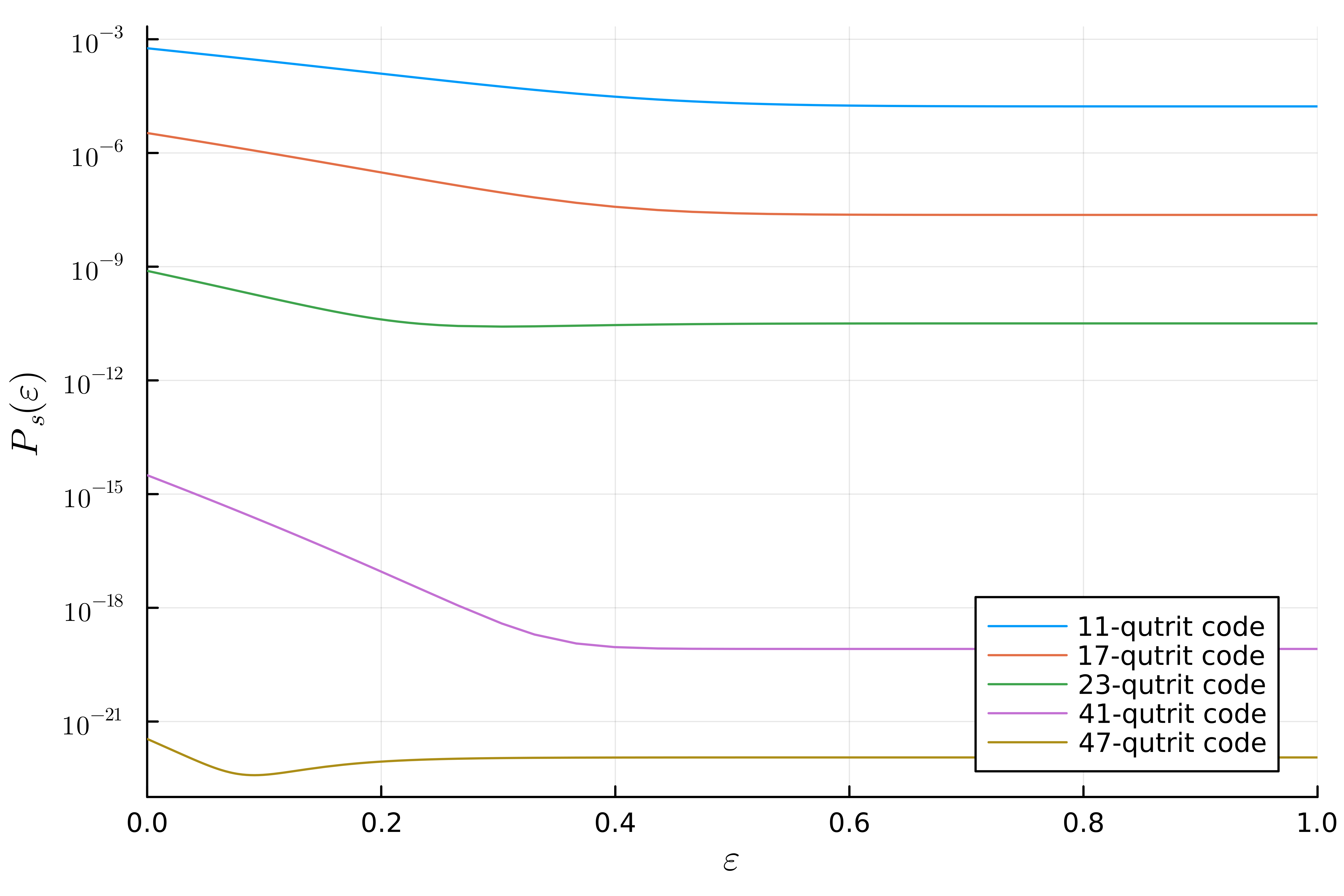}
    \caption{Success probability as a function of pre-distillation noise for codes that distill qutrit Strange states.}
    \label{Figure:StrangeStateSuccessProb}
\end{figure}

\begin{table}
    \centering
    \begin{tabular}{|c|c|}
        \hline
        Code & Threshold \\
        \hline\hline
        $[11,5,6]_9$ & $\approx0.38715$ \\
        \hline
        $[17,8,8]_9$ & $\approx0.34394$ \\
        \hline
        $[23,11,9]_9$ & $\approx0.16636$ \\
        \hline
        $[41,20,14]_9$ & $\approx0.31877$ \\
        \hline
        $[47,23,15]_9$ & $\approx0.04864$ \\
        \hline
    \end{tabular}
    \caption{Noise thresholds of magic state distillation protocols for qutrit Strange states based on quadratic residue codes over $\mathbb{F}_9$.}
    \label{Table:StrangeStateThresholds}
\end{table}

\section{Discussion}\label{Section:Discussion}

In this paper, we showed that quantum quadratic residue codes are useful for distilling magic states. This includes showing that some codes which were previously known to distill qubit $H$ and $T$ states and qutrit Strange states are equivalent to certain quantum QR codes. We also found new examples of quantum QR codes that distill qubit $T$ states and qutrit Strange states. Finally, we showed that there are infinitely many quantum QR codes on qubits that distill $T$ states with a non-trivial threshold.

A strength of the quantum quadratic residue code family is that many particularly special codes are unified under a single umbrella. This includes the codes with the highest currently known thresholds for $T$ and Strange state distillation, namely the $5$-qubit perfect code~\cite{BravyiKitaev2005} and the $11$-qutrit Golay code~\cite{Prakash2020}, as well as the $7$-qubit Steane code and $23$-qubit Golay code, which achieve the highest possible threshold for $H$ state distillation~\cite{Reichardt2005}. Therefore, quantum quadratic residue codes provide a common algebraic framework for defining magic state distillation protocols for a variety of magic states on qubits and on qutrits.

In addition to reproducing previous magic state distillation protocols, we find a number of new distillation protocols, including some fairly large codes that distill qubit $T$ states and qutrit Strange states with fairly high thresholds. This is contrary to the common expectation that smaller codes tend to have better distillation thresholds. In fact, from Tables~\ref{Table:TStateThresholds} and~\ref{Table:StrangeStateThresholds}, we can see that for quantum QR codes the distillation thresholds are not monotonic either with the length of the code or with the code distance. Rather, the thresholds depend on the full weight distributions of the codes.

There is an infinite sequence of quantum QR codes of increasing size and increasing distance for both qubits and qutrits, as well as more generally for prime-dimensional qudits. Theorem~\ref{Theorem:InfiniteTDistillation} shows that there are infinitely many quantum QR codes that distill $T$ states with a non-trivial threshold. Although this falls short of a resolution of the conjecture stated in the introduction, it does suggest that the capacity of these codes to distill $T$ states is not merely incidental for some small examples. Some deeper algebraic structure of these codes makes them genuinely well-suited to magic state distillation.

Taken together, the above-mentioned results make quantum quadratic residue codes natural candidates for resolving the two conjectures stated in the introduction on which states are universal for quantum computation. That said, the performance of these codes for distilling magic states seems to depend on the entire weight distribution of the codes, and a complete analytic characterization of these weight distributions is likely out of reach with present methods. The weight distributions of the corresponding classical quadratic residue codes are not known in general either, and characterizing them has been an open problem since the codes were first introduced several decades ago.

In light of this limitation, one possible avenue for future work could be to extend the numeric results obtained here to larger code sizes. Extending the numerics to larger code lengths would also extend the resource usage analysis in Appendix~\ref{Section:ResourceAnalysis}, where we compare the performance of single-round protocols based on quantum QR codes against the baseline of multi-round protocols formed by concatenating the $5$-qubit protocol with itself. For the accessible protocol sizes, the concatenated protocols seem to consume fewer raw input states for a given target output fidelity. However, there are reasons to expect that larger single-round protocols could have benefits over multi-round protocols when the protocols' success probabilities are boosted, for example, by error correcting rather than post-selecting on a single syndrome. These are detailed in Appendix~\ref{Section:ResourceAnalysis}.

Furthermore, Theorem~\ref{Theorem:InfiniteTDistillation} shows that some asymptotic results are possible even without a complete characterization of the weight distributions of these codes. Another avenue could involve extending Theorem~\ref{Theorem:InfiniteTDistillation}, for example, by showing that there exists a subsequence of qubit quantum QR codes with thresholds that remain bounded away from $0$, or even which increase monotonically. An equivalent of Theorem~\ref{Theorem:InfiniteTDistillation} could also be shown for qutrit Strange state distillation. Finally, quantum quadratic residue codes are defined on qudits of any prime Hilbert space dimension, not just qubits and qutrits. They could lead to distillation protocols for ququint magic states and magic states on higher-dimensional qudits. Overall, the present results suggest that quantum quadratic residue codes are a promising testing ground for both finite computational results and asymptotic questions in magic state distillation.

\section*{Acknowledgements}

We thank Amolak Ratan Kalra, Shiroman Prakash, and Mark Howard for insightful discussions. M.Z.~is funded by the Natural Sciences and Engineering Research Council of Canada~(NSERC). S.J.~is funded by the NSERC-European Commission project FoQaCiA. N.D.~acknowledges support from the Canada Research Chair program, NSERC Discovery Grant RGPIN-2022-03103, the NSERC-European Commission project FoQaCiA, and the Faculty of Science of Simon Fraser University. This work was also supported by NSERC Alliance International Catalyst Quantum Grant (ALLRP 592647-24). Finally, this work was aided by discussions that took place at the Mathematical Foundations of Quantum Advantage workshop held at Simon Fraser University in May of 2025 and funded by the Pacific Institute for the Mathematical Sciences~(PIMS).

\bibliographystyle{quantum}
\bibliography{biblio}

\appendix

\section{The structure of magic state distillation protocols for qudits}\label{Section:MSDStructure}

In the main text, we use the Campbell-Browne result~\cite{CampbellBrowne2009}, which says that in analyzing the performance of $n$-to-$1$ qudit magic state distillation protocols (as measured by the conditional output fidelity or achievable noise threshold) it suffices to consider stabilizer reductions, i.e., protocols of the form
\begin{equation}
    \rho\longmapsto\mathcal{R}(\rho^{\otimes n}):=\Tr_{2,\dots,n}(D\Pi_I^r\rho^{\otimes n}\Pi_I^rD^\dagger),
\end{equation}
where $I\subset\mathbb{F}_d^{2n}$ is an isotropic subspace, $r:I\to\mathbb{F}_d$ is a consistent value assignment (see Eq.~\eqref{Equation:NoncontextualityCondition}), $\Pi_I^r$ is the stabilizer projector (Eq.~\eqref{Equation:PauliProjector}), and $D$ is a decoding Clifford unitary. In Ref.~\cite{CampbellBrowne2009}, this was proven only for qubits. In this appendix, we show that this result also holds for magic state distillation on odd-prime-dimensional qudits.

\begin{Theorem*}\label{Theorem:StabilizerReductionPrime}
    Let $d$ be an odd prime. Consider any $n$-to-$1$ qudit stabilizer protocol $\Phi$ on $d$-dimensional qudits (a completely positive, trace non-increasing map built from Clifford unitaries, Pauli measurements with postselection, and partial trace). Fix an input state $\rho$ and a target pure state $\ket{M}$. Then either (i)~there is a stabilizer state whose fidelity with respect to $\ket{M}$ is at least that of $\Phi(\rho^{\otimes n})$, or (ii)~there exists an $n$-to-$1$ qudit stabilizer protocol $\mathcal{R}$ whose conditional output fidelity with respect to $\ket{M}$ is at least that of $\Phi$ on input state $\rho^{\otimes n}$ and which has the form of a stabilizer reduction, for some $(n-1)$-dimensional isotropic subspace $I\subset\mathbb{F}_d^{2n}$, consistent assignment $r:I\to\mathbb{F}_d$, and decoding Clifford $D$.
\end{Theorem*}

\begin{proof}[Proof sketch.]
    Any stabilizer protocol $\Phi$ has a decomposition into a number of branches, with each branch determined by the outcomes of the Pauli measurements and any classical randomness used. For the fixed input state $\rho^{\otimes n}$, write the unnormalized output state of the protocol $\Phi$ as
    \begin{equation}
        \Phi(\rho^{\otimes n})=\sum\limits_{\alpha}\Pi_\alpha(\rho^{\otimes n}),
    \end{equation}
    where $\alpha$ ranges over indices of accepted branches, and each channel $\Pi_\alpha$ is simply a sequence of Clifford gates $g_i\in\mathcal{C}\ell$ and Pauli projectors $\Pi_{I_j}^{r_j}$, followed by a partial trace of qudits $2$ through $n$. Let $p_\alpha=\Tr(\Pi_\alpha(\rho^{\otimes n}))$, and let $\rho_\alpha=\Pi_\alpha(\rho^{\otimes n})/p_\alpha$ be the normalized output state of each branch. The normalized output of $\Phi$ is
    \begin{equation}
        \rho'=\frac{\Phi(\rho^{\otimes n})}{\Tr(\Phi(\rho^{\otimes n}))}=\sum_\alpha\frac{p_\alpha}{\sum_\beta p_\beta}\rho_\alpha.
    \end{equation}

    The fidelity of the post-distillation state with respect to $\ket{M}$ is
    \begin{equation}
        \Tr(\ketbra{M}{M}\rho')=\sum_\alpha\frac{p_\alpha}{\sum_\beta p_\beta}\Tr(\ketbra{M}{M}\rho_\alpha)\le\max_\alpha\Tr(\ketbra{M}{M}\rho_\alpha),
    \end{equation}
    where the first equality follows from linearity of the trace, and the last inequality follows from the fact that the coefficients are nonnegative. Therefore, there is a branch $\alpha^*$ such that the fidelity of $\rho_{\alpha^*}$ with respect to $\ket{M}$ is at least the fidelity of $\rho'$ with respect to $\ket{M}$.

    The channel $\Pi_{\alpha^*}$ has the form
    \begin{equation}
        \Pi_{\alpha^*}(\rho^{\otimes n})=\Tr_{2,\dots,n}(\cdots \Pi_{I_2}^{r_2}g_2\Pi_{I_1}^{r_1}g_1\rho^{\otimes n}g_1^\dagger\Pi_{I_1}^{r_1}g_2^\dagger\Pi_{I_2}^{r_2} \cdots).
    \end{equation}
    with $g_1,g_2,\dots\in\mathcal{C}\ell$, $I_1,I_2,\dots\subset\mathbb{F}_d^{2n}$ isotropic subspaces, and $r_1,r_2,\dots$ consistent value assignments for the subspaces. This can be reduced to the form of a stabilizer reduction using standard techniques from Pauli-based computation (see, for example, Refs.~\cite{BravyiSmolin2016,PeresGalvao2021,Peres2023}). In particular, to start the Clifford gates can be propagated past the projectors, conjugating the Pauli projectors into other Pauli projectors, resulting in a map of the form
    \begin{equation}
        \Pi_{\alpha^*}(\rho^{\otimes n})=\Tr_{2,\dots,n}(g\cdots \Pi_{I_2'}^{r_2'}\Pi_{I_1'}^{r_1'}\rho^{\otimes n}\Pi_{I_1'}^{r_1'}\Pi_{I_2'}^{r_2'}\cdots g^\dagger).
    \end{equation}
    Then using techniques from Pauli-based computation (see Ref.~\cite{Peres2023}), the sequence of Pauli projectors can be reduced to a single Pauli projector $\Pi_I^r$ to get a map of the form
    \begin{equation}
        \Pi_{\alpha^*}(\rho^{\otimes n})=\Tr_{2,\dots,n}(g\Pi_I^r\rho^{\otimes n}\Pi_I^rg^\dagger).
    \end{equation}
    Note that this step holds for qudits of all prime dimensions, but not necessarily for qudits of composite dimensions, because it assumes the Clifford group acts transitively on the sets of Pauli projectors $\Pi_i^r$ of each dimension.

    The last step involves tracing out the last $n-1$ qudits. This is equivalent to measuring observables $Z_2,\dots,Z_n$ on these qudits and then discarding the outcomes. Let $J=\langle z_2,\dots,z_n\rangle$ and $s:J\to\mathbb{F}_d$ be the labels of the measurement outcomes for measurements $J$. Then
    \begin{equation}
        \Pi_{\alpha^*}(\rho^{\otimes n})=\sum\limits_{s\in J^*}p(s)\Tr_{2,\dots,n}(\Pi_J^sg\Pi_I^r\rho^{\otimes n}\Pi_I^rg^\dagger\Pi_J^s).
    \end{equation}
    with nonnegative coefficients $p(s)$. By the same argument as before, without decreasing the output fidelity (but reducing the success probability), we can post-select on one of the outcomes $s:J\to\mathbb{F}_d$.

    Now, propagating the new projector $\Pi_J^s$ past the Clifford group element $g$, and absorbing it into the old projector, we get
    \begin{equation}
        \Pi_J^sg\Pi_I^r\propto g\Pi_{I'}^{r'},
    \end{equation}
    where $\Pi_{I'}^{r'}$ has rank $1$ or $d$. Here we have two cases. If $\Pi_{I'}^{r'}$ has rank $1$, then the post-distillation state is a stabilizer state and we are in case~1 of the theorem statement. Otherwise, $\Pi_{I'}^{r'}$ has rank $d$, we are still in the second case of the theorem statement, and we can assume $\dim(I')=n-1$. Also, by construction, $g$ is a decoding Clifford for $\Pi_{I'}^{r'}$, it maps the code space onto the first physical qudit.

    The result is a protocol of the form
    \begin{equation}
        \rho\longmapsto\mathcal{R}(\rho^{\otimes n})=\Tr_{2,\dots,n}\left(g\Pi_{I'}^{r'}\rho^{\otimes n}\Pi_{I'}^{r'}g^\dagger\right),
    \end{equation}
    where $I'\subset\mathbb{F}_d^{2n}$ is a $n-1$-dimensional isotropic subspace, $r':I'\to\mathbb{F}_d$ is a consistent value assignment for $I'$, and $g$ is a Clifford that maps the code space of $I'$ onto the first physical qudit. This is exactly a stabilizer reduction.  By construction, this stabilizer reduction has conditional output fidelity at least that of the original protocol $\Phi$ on $\rho^{\otimes n}$.
\end{proof}

\section{Stabilizer operations in phase space}\label{Section:PhaseSpaceMethods}

In this section, we assume we have a system of $n$ qudits, each with odd-prime Hilbert space dimension $d$. It is helpful to understand how stabilizer operations act in the phase space picture. This section is reproduced here for convenience from parts of Section~2.4 of Ref.~\cite{Zurel2024}.

First, we consider Clifford unitaries. For odd-prime $d$, the Clifford group splits as
\begin{equation}
    \mathcal{C}\ell\cong\mathrm{Sp}(\mathbb{F}_d^{2n})\ltimes\mathbb{F}_d^{2n},
\end{equation}
and Clifford operations act on the phase space as affine symplectic transformations,
\begin{equation}\label{Equation:PhaseSpaceCliffordUpdate}
    gA_ug^\dagger=A_{S_gu+a_g}\qquad\forall g\in\mathcal{C}\ell\quad\forall u\in\mathbb{F}_d^{2n}.
\end{equation}
where $S_g\in\mathrm{Sp}(\mathbb{F}_d^{2n})$ and $a_g\in\mathbb{F}_d^{2n}$~\cite[Theorem~7]{Gross2006}. The Pauli group is a subgroup of the Clifford group corresponding to translations of the phase space, $T_aA_uT_a^\dagger=A_{u+a}$. As a corollary of Eq.~\eqref{Equation:PhaseSpaceCliffordUpdate}, the Wigner function is Clifford covariant, i.e., for any quantum state $\rho$ and any Clifford group element $g\in\mathcal{C}\ell$,
\begin{equation}
    W_{g\rho g^\dagger}(S_gu+a_g)=W_\rho(u).
\end{equation}

The effect of Pauli measurements on the phase space is described by the following lemma.

\begin{Lemma}\label{Lemma:PhaseSpacePauliUpdate}
    For any isotropic subspace $I\subset\mathbb{F}_d^{2n}$, any value assignment $r:I\to\mathbb{F}_d$, and any phase space point $u\in\mathbb{F}_d^{2n}$,
    \begin{itemize}
        \item[(i)] if $r(a)=-[u,a]$ for all $a\in I$, then $\Tr(\Pi_I^rA_u)=1$ and
        \begin{equation}
            \Pi_I^rA_u\Pi_I^r=\frac{1}{|I|}\sum\limits_{v\in u+I}A_v,
        \end{equation}
        \item[(ii)] if $r(a)\ne-[u,a]$ for some $a\in I$, then $\Tr(\Pi_I^rA_u)=0$ and $\Pi_I^rA_u\Pi_I^r=0$.
    \end{itemize}
\end{Lemma}

\begin{proof}[Proof of Lemma~\ref{Lemma:PhaseSpacePauliUpdate}]
	For any isotropic subspace $I\subset \mathbb{F}_d^{2n}$, any value assignment $r:I\to\mathbb{F}_d$, and any phase space point $u\in\mathbb{F}_d^{2n}$,
	\begin{align}
		\Pi_I^rA_u\Pi_I^r =& \frac{1}{|I|^2d^n}\sum\limits_{a,c\in I}\sum\limits_{b\in\mathbb{F}_d^{2n}}\omega^{-r(a)-r(c)+[u,b]}T_aT_bT_c\\
		=& \frac{1}{|I|^2d^n}\sum\limits_{a,c\in I}\sum\limits_{b\in\mathbb{F}_d^{2n}}\omega^{-r(a+c)+[u,b]+[b,c]}T_{a+c}T_b\\
		=& \frac{1}{|I|^2d^n}\sum\limits_{a\in I}\sum\limits_{b\in\mathbb{F}_d^{2n}}\omega^{-r(a)+[u,b]}\left[\sum\limits_{c\in I}\omega^{[b,c]}\right]T_aT_b\\
		=& \frac{1}{|I|d^n}\sum\limits_{a\in I}\sum\limits_{b\in I^\perp}\omega^{-r(a)+[u,b]}T_aT_b.
	\end{align}
	Taking a trace of this equation, we get
	\begin{align}
		\Tr(\Pi_I^rA_u) =& \frac{1}{|I|d^n}\sum\limits_{a\in I}\sum\limits_{b\in I^\perp}\omega^{-r(a)+[u,b]}d^n\delta_{b,-a}\\
		=& \frac{1}{|I|}\sum\limits_{a\in I}\omega^{-r(a)-[u,a]}\\
		=& \begin{cases}
			1\quad\text{if $r(a)=-[u,a]\;\forall a\in I$},\\
			0\quad\text{otherwise}.
		\end{cases}
	\end{align}
	Here the last line follows from orthogonality of characters. Now, assuming $r(a)=-[u,a]\;\forall a\in I$, we have
	\begin{align}
		\Pi_I^rA_u\Pi_I^r =& \frac{1}{|I|d^n}\sum\limits_{a\in I}\sum\limits_{b\in I^\perp}\omega^{[u,a+b]}T_{a+b}.
	\end{align}
	Since $I$ is isotropic, $I\subset I^\perp$, and we can write $I^\perp$ as $I^\perp=\bigsqcup_{a\in I^\perp/ I}I+a$. Then $I+I^\perp=|I|\cdot I^\perp$, and
	\begin{equation}
		\Pi_I^rA_u\Pi_I^r=\frac{1}{d^n}\sum\limits_{a\in I^\perp}\omega^{[u,a]}T_a.
	\end{equation}
	It is easy to check that for any $v\in u+I$, we have $[v,w]=[u,w]\;\forall w\in I^\perp$. Then
	\begin{align}
		\frac{1}{|I|}\sum\limits_{v\in u+I}A_v =& \frac{1}{|I|d^n}\sum\limits_{v\in u+I}\sum\limits_{b\in\mathbb{F}_d^{2n}}\omega^{[v,b]}T_b\\
		=& \frac{1}{|I|d^n}\sum\limits_{b\in\mathbb{F}_d^{2n}}\left[\sum\limits_{v\in u+I}\omega^{[v,b]}\right]T_b\\
		=& \frac{1}{d^n}\sum\limits_{b\in I^\perp}\omega^{[u,b]}T_b=\Pi_I^rA_u\Pi_I^r.
	\end{align}
	This proves the first statement of the lemma. It remains to show that in the case that there exists some $a\in I$ such that $r(a)\ne-[u,a]$, we have $\Pi_I^rA_u\Pi_I^r=0$ as an operator. To see this, consider the inner product $\Tr(T_a\Pi_I^rA_u\Pi_I^r)$ for any Pauli operator $T_a,\;a\in\mathbb{F}_d^{2n}$. Here we have three cases.
	
	(I)~If $a\in I$, then $T_a\Pi_I^r=\omega^{r(a)}\Pi_I^r$, and so $\Tr(T_a\Pi_I^rA_u\Pi_I^r)=\omega^{r(a)}\Tr(\Pi_I^rA_u)=0$.
	
	(II)~If $a\in I^\perp\setminus I$, then by Lemma~5 of \cite{ZurelHeimendahl2024}, $\Pi_I^r=\sum\limits_{r'\in\Gamma}\Pi_{\langle a,I\rangle}^{r'}$, where $\Gamma$ is the set of all value assignments on $\langle a,I\rangle$ which agree with $r$ on $I$. Then
	\begin{equation}
		\Tr(\Pi_I^rA_u)=\sum\limits_{r'\in\Gamma}\Tr(\Pi_{\langle a,I\rangle}^{r'}A_u).
	\end{equation}
	Here the left hand side is zero by assumption, and on the right hand side we have a sum of nonnegative terms. Therefore, each term in the sum on the right hand side is zero. The spectral decomposition of $T_a$ is
	\begin{equation}
		T_a=\sum\limits_{k\in\mathbb{Z}_d}\omega^k\Pi_a^k.
	\end{equation}
	Then,
	\begin{align}
		\Tr(T_a\Pi_I^rA_u\Pi_I^r) =& \sum\limits_{k\in\mathbb{Z}_d}\omega^k\Tr(\Pi_a^k\Pi_I^rA_u\Pi_I^r)\\
		=& \sum\limits_{k\in\mathbb{Z}_d}\omega^k\Tr(\Pi_{\langle a,I\rangle}^{r_k}A_u)=0.
	\end{align}
	
	(III)~If $a\notin I^\perp$, then
	\begin{align}
		\Pi_I^rT_a\Pi_I^r=&\frac{1}{|I|^2}\sum\limits_{b,c\in I}\omega^{-r(b)-r(c)}T_bT_aT_c\\
		=&\frac{1}{|I|^2}T_a\sum\limits_{b,c\in I}\omega^{-r(b+c)+[b,a]}T_{b+c}\\
		=&\frac{1}{|I|^2}T_a\Pi_I^r\sum\limits_{b\in I}\omega^{[b,a]}.
	\end{align}
	By character orthogonality, the sum in the final expression vanishes. Therefore,
	\begin{equation}
		\Tr(T_a\Pi_I^rA_u\Pi_I^r)=\Tr(\Pi_I^rT_a\Pi_I^rA_u)=0.
	\end{equation}
	Thus, $\Tr(T_a\Pi_I^rA_u\Pi_I^r)=0$ for all Pauli operators $T_a,a\in\mathbb{F}_d^{2n}$, and so $\Pi_I^rA_u\Pi_I^r=0$ as an operator.
\end{proof}

\section{Weight enumerator tables}\label{Section:WeightEnumeratorTables}

\renewcommand{\arraystretch}{1.1}

\begin{table}[H]
    \centering
    \begin{tabular}{| p{0.11\textwidth} | p{0.89\textwidth} |}
        \hline
        Code & Weight enumerator\\
        \hline\hline
        $[5,2,4]_4$ & \hspace{-1pt}$x^{5} + 15 x y^{4}$ \\
        \hline
        $[7,3,4]_4$ & \hspace{-1pt}$x^{7} + 21 x^{3} y^{4} + 42 x y^{6}$ \\
        \hline
        $[13,6,6]_4$ & \hspace{-1pt}$x^{13} + 156 x^{7} y^{6} + 1053 x^{5} y^{8} + 2028 x^{3} y^{10} + 858 x y^{12}$ \\
        \hline
        $[23,11,8]_4$ & \hspace{-6pt}$\begin{array}{l}x^{23} + 1518 x^{15} y^{8} + 110124 x^{11} y^{12} + 425040 x^{9} y^{14} + 1298649 x^{7} y^{16} + 1530144 x^{5} y^{18}\\ + 743820 x^{3} y^{20} + 85008 x y^{22}\end{array}$ \\
        \hline
        $[29,14,12]_4$ & \hspace{-6pt}$\begin{array}{l}x^{29} + 71253 x^{17} y^{12} + 613872 x^{15} y^{14} + 5618025 x^{13} y^{16} + 24701040 x^{11} y^{18} + 65315250 x^{9} y^{20}\\ + 91040976 x^{7} y^{22} + 62560134 x^{5} y^{24} + 17276112 x^{3} y^{26} + 1238793 x y^{28}\end{array}$ \\
        \hline
        $[31,15,8]_4$ & \hspace{-6pt}$\begin{array}{l}x^{31} + 1395 x^{23} y^{8} + 58590 x^{19} y^{12} + 468720 x^{17} y^{14} + 6072807 x^{15} y^{16} + 37914240 x^{13} y^{18}\\ + 135749124 x^{11} y^{20} + 296335200 x^{9} y^{22} + 345151365 x^{7} y^{24} + 200987136 x^{5} y^{26}\\ + 48076350 x^{3} y^{28} + 2926896 x y^{30}\end{array}$ \\
        \hline
        $[37,18,12]_4$ & \hspace{-6pt}$\begin{array}{l}x^{37} + 12987 x^{25} y^{12} + 163836 x^{23} y^{14} + 4190472 x^{21} y^{16} + 49694700 x^{19} y^{18}\\ + 402823773 x^{17} y^{20} + 2140777080 x^{15} y^{22} + 7315768908 x^{13} y^{24} + 15818285880 x^{11} y^{26}\\ + 20702769285 x^{9} y^{28} + 15425347212 x^{7} y^{30} + 5876290827 x^{5} y^{32} + 942953436 x^{3} y^{34}\\ + 40398339 x y^{36}\end{array}$ \\
        \hline
        $[47,23,12]_4$ & \hspace{-6pt}$\begin{array}{l}x^{47} + 38916 x^{35} y^{12} + 1070190 x^{31} y^{16} + 9988440 x^{29} y^{18} + 231731808 x^{27} y^{20}\\ + 3350122776 x^{25} y^{22} + 32422518540 x^{23} y^{24} + 226165393080 x^{21} y^{26} + 1134417096120 x^{19} y^{28}\\ + 4010049018360 x^{17} y^{30} + 9894334090905 x^{15} y^{32} + 16673213458440 x^{13} y^{34}\\ + 18572482193952 x^{11} y^{36} + 13080263144520 x^{9} y^{38} + 5432928694380 x^{7} y^{40}\\ + 1192386006504 x^{5} y^{42} + 113540737860 x^{3} y^{44} + 2948872872 x y^{46}\end{array}$ \\
        \hline
        $[53,26,16]_4$ & \hspace{-6pt}$\begin{array}{l}x^{53} + 196365 x^{37} y^{16} + 1008696 x^{35} y^{18} + 94594188 x^{33} y^{20} + 1497103296 x^{31} y^{22}\\ + 25024672920 x^{29} y^{24} + 272515859616 x^{27} y^{26} + 2301852043920 x^{25} y^{28}\\ + 14237883206976 x^{23} y^{30} + 65437002577947 x^{21} y^{32} + 220371245649360 x^{19} y^{34}\\ + 538407879739080 x^{17} y^{36} + 937381604347968 x^{15} y^{38} + 1135693969571784 x^{13} y^{40}\\ + 925961040090528 x^{11} y^{42} + 484514894731920 x^{9} y^{44} + 151674555714240 x^{7} y^{46}\\ + 25413417241719 x^{5} y^{48} + 1867132096440 x^{3} y^{50} + 38016923532 x y^{52}\end{array}$ \\
        \hline
        $[61,30,18]_4$ & \hspace{-6pt}$\begin{array}{l}x^{61} + 201300 x^{43} y^{18} + 9791964 x^{41} y^{20} + 297188340 x^{39} y^{22} + 7311487755 x^{37} y^{24}\\ + 134130915792 x^{35} y^{26} + 1902476580870 x^{33} y^{28} + 20779092892368 x^{31} y^{30}\\ + 175323674956545 x^{29} y^{32} + 1141965426170520 x^{27} y^{34} + 5726086694463408 x^{25} y^{36}\\ + 21992184927250200 x^{23} y^{38} + 64200120673894854 x^{21} y^{40} + 140927247768490320 x^{19} y^{42}\\ + 229267347842000340 x^{17} y^{44} + 271133644941905040 x^{15} y^{46} + 227146495371493950 x^{13} y^{48}\\ + 130168863213334452 x^{11} y^{50} + 48592446257133540 x^{9} y^{52} + 11002063150069140 x^{7} y^{54}\\ + 1350253440272367 x^{5} y^{56} + 73516463568960 x^{3} y^{58} + 1121442784950 x y^{60}\end{array}$ \\
        \hline
        $[71,35,12]_4$ & \hspace{-6pt}$\begin{array}{l}x^{71} + 7455 x^{59} y^{12} + 499485 x^{55} y^{16} + 44730 x^{53} y^{18} + 42565068 x^{51} y^{20} + 51290400 x^{49} y^{22}\\ + 2044474110 x^{47} y^{24} + 24392610900 x^{45} y^{26} + 478715797065 x^{43} y^{28} + 8385401713080 x^{41} y^{30}\\ + 123970358688225 x^{39} y^{32} + 1477842682144710 x^{37} y^{34} + 14063872023855282 x^{35} y^{36}\\ + 107143099467071640 x^{33} y^{38} + 652728283699886100 x^{31} y^{40}\\ + 3172457555098365720 x^{29} y^{42} + 12254266094922768945 x^{27} y^{44}\\ + 37402445047293925872 x^{25} y^{46} + 89525811715214684355 x^{23} y^{48}\\ + 166409387015909428710 x^{21} y^{50} + 237190163461701506280 x^{19} y^{52}\\ + 255090582785520762000 x^{17} y^{54} + 202747446918460600590 x^{15} y^{56}\\ + 115908476245380852180 x^{13} y^{58} + 45970393268864625351 x^{11} y^{60}\\ + 12033457215125864280 x^{9} y^{62} + 1933977124137515550 x^{7} y^{64}\\ + 170403391711501914 x^{5} y^{66} + 6731583805555770 x^{3} y^{68} + 75335382697656 x y^{70}\end{array}$ \\
        \hline
    \end{tabular}
    \caption{Weight enumerator polynomials of expurgated quadratic residue codes over $GF(4)$.}
    \label{Table:GF4Weights}
\end{table}

\begin{table}[H]
    \centering
    \begin{tabular}{| p{0.11\textwidth} | p{0.89\textwidth} |}
        \hline
        Code & Weight enumerator\\
        \hline\hline
        $[5,2,4]_9$ & \hspace{-1pt}$x^{5} + 40 x y^{4} + 40 y^{5}$ \\
        \hline
        $[11,5,6]_9$ & \hspace{-1pt}$x^{11} + 528 x^{5} y^{6} + 7920 x^{3} y^{8} + 11000 x^{2} y^{9} + 23760 x y^{10} + 15840 y^{11}$ \\
        \hline
        $[17,8,8]_9$ & \hspace{-6pt}$\begin{array}{l}x^{17} + 2720 x^{9} y^{8} + 77248 x^{7} y^{10} + 220864 x^{6} y^{11} + 1214752 x^{5} y^{12} + 3179680 x^{4} y^{13}\\ + 7936960 x^{3} y^{14} + 12206272 x^{2} y^{15} + 12404152 x y^{16} + 5804072 y^{17}\end{array}$ \\
        \hline
        $[23,11,9]_9$ & \hspace{-6pt}$\begin{array}{l}x^{23} + 10120 x^{14} y^{9} + 366528 x^{11} y^{12} + 2003760 x^{10} y^{13} + 13115520 x^{9} y^{14} + 62298720 x^{8} y^{15}\\ + 241544160 x^{7} y^{16} + 803325600 x^{6} y^{17} + 2154216064 x^{5} y^{18} + 4512831840 x^{4} y^{19}\\ + 7226651520 x^{3} y^{20} + 8268193824 x^{2} y^{21} + 6005888064 x y^{22} + 2090613888 y^{23}\end{array}$ \\
        \hline
        $[29,14,12]_9$ & \hspace{-6pt}$\begin{array}{l}x^{29} + 43848 x^{17} y^{12} + 82824 x^{16} y^{13} + 1714944 x^{15} y^{14} + 13739040 x^{14} y^{15} + 90658176 x^{13} y^{16}\\ + 571797408 x^{12} y^{17} + 3028883424 x^{11} y^{18} + 13993514304 x^{10} y^{19} + 56123954040 x^{9} y^{20}\\ + 192285362808 x^{8} y^{21} + 559241011200 x^{7} y^{22} + 1362023628000 x^{6} y^{23}\\ + 2723779515240 x^{5} y^{24} + 4357956668520 x^{4} y^{25} + 5363900914848 x^{3} y^{26}\\ + 4767732820928 x^{2} y^{27} + 2724475922400 x y^{28} + 751572223008 y^{29}\end{array}$ \\
        \hline
        $[41,20,14]_9$ & \hspace{-6pt}$\begin{array}{l}x^{41} + 13120 x^{27} y^{14} + 298480 x^{25} y^{16} + 984000 x^{24} y^{17} + 51010560 x^{23} y^{18}\\ + 241860640 x^{22} y^{19} + 3178539760 x^{21} y^{20} + 21320311600 x^{20} y^{21} + 169521814400 x^{19} y^{22}\\ + 1076372094720 x^{18} y^{23} + 6584910839040 x^{17} y^{24} + 35468786005600 x^{16} y^{25}\\ + 175512737265408 x^{15} y^{26} + 778023868378880 x^{14} y^{27} + 3116202323611840 x^{13} y^{28}\\ + 11168036013571840 x^{12} y^{29} + 35749106898555904 x^{11} y^{30} + 101465555251461632 x^{10} y^{31}\\ + 253683676686494280 x^{9} y^{32} + 553469380198545480 x^{8} y^{33} + 1041846451913296000 x^{7} y^{34}\\ + 1666936669433112672 x^{6} y^{35} + 2222593774949021136 x^{5} y^{36}\\ + 2402798170045927120 x^{4} y^{37} + 2023411266902614720 x^{3} y^{38}\\ + 1245175538958225280 x^{2} y^{39} + 498070323942465632 x y^{40} + 97184444550609056 y^{41}\end{array}$ \\
        \hline
        $[47,23,15]_9$ & \hspace{-6pt}$\begin{array}{l}x^{47} + 1141536 x^{32} y^{15} + 50417840 x^{29} y^{18} + 55935264 x^{27} y^{20} + 2167776864 x^{26} y^{21}\\ + 12302333472 x^{25} y^{22} + 115066828800 x^{24} y^{23} + 957374289216 x^{23} y^{24}\\ + 7020700962528 x^{22} y^{25} + 47575787881248 x^{21} y^{26} + 296009366479904 x^{20} y^{27}\\ + 1691060003624640 x^{19} y^{28} + 8862795495027936 x^{18} y^{29} + 42539887522934208 x^{17} y^{30}\\ + 186634650694261440 x^{16} y^{31} + 746558999374719168 x^{15} y^{32}\\ + 2714712409349744760 x^{14} y^{33} + 8942498372112933936 x^{13} y^{34}\\ + 26572195657617214176 x^{12} y^{35} + 70859323050286695456 x^{11} y^{36}\\ + 168529756696023632448 x^{10} y^{37} + 354799566202158364320 x^{9} y^{38}\\ + 655015292493578207136 x^{8} y^{39} + 1048023930654922368576 x^{7} y^{40}\\ + 1431446993858724464160 x^{6} y^{41} + 1635940023559588443840 x^{5} y^{42}\\ + 1521804500115282894720 x^{4} y^{43} + 1106766736976396745216 x^{3} y^{44}\\ + 590275754871576999040 x^{2} y^{45} + 205313252389258020096 x y^{46}\\ + 34946943259659753984 y^{47}\end{array}$\\
        \hline
    \end{tabular}
    \caption{Weight enumerator polynomials of expurgated quadratic residue codes over $GF(9)$.}
    \label{Table:GF9Weights}
\end{table}

\pagebreak

\section{Post-distillation noise plots}\label{Section:AllPlots}

\begin{figure}[H]
    \centering
    \includegraphics[width=0.85\linewidth]{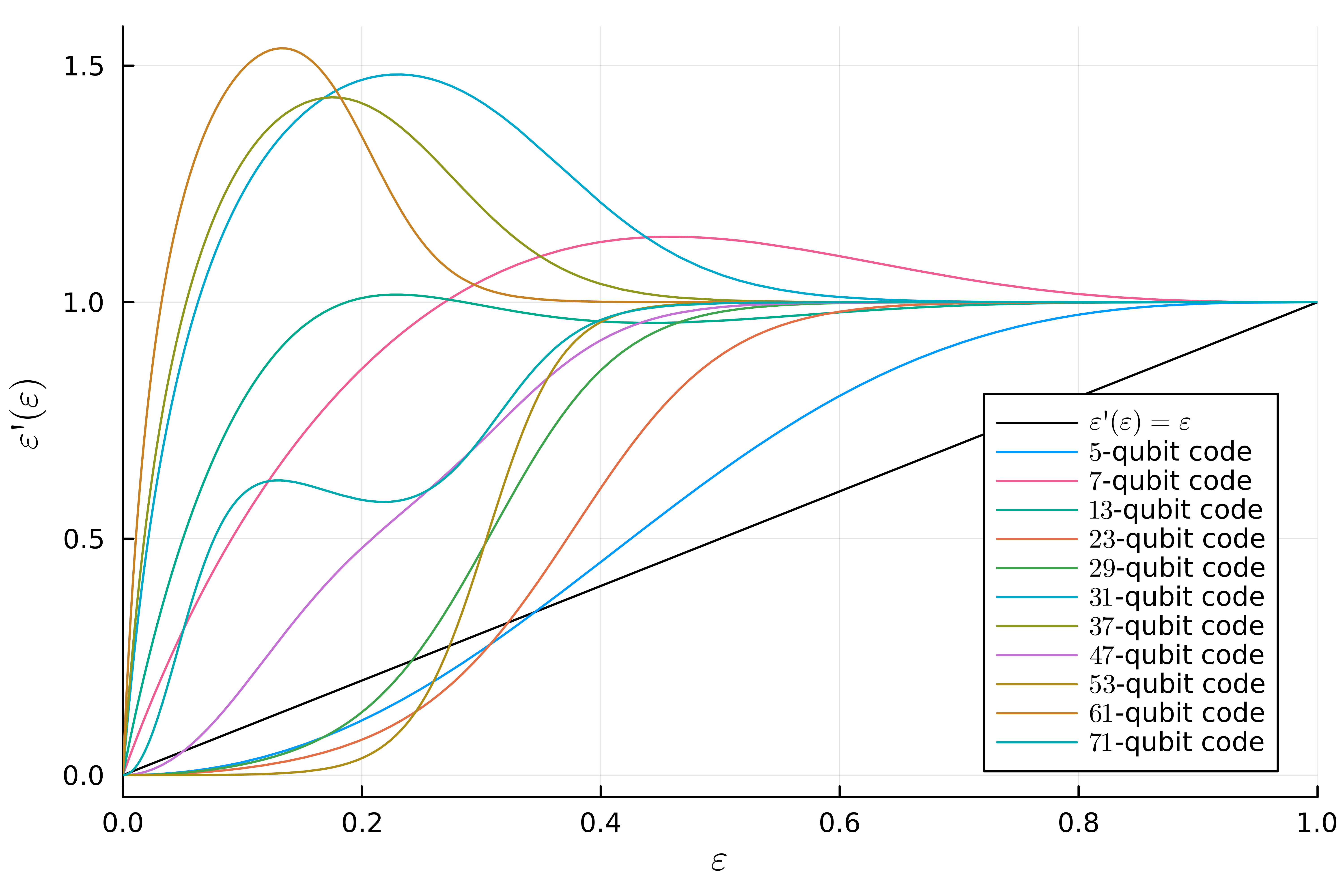}
    \caption{Post-distillation noise as a function of predistillation noise for magic state distillation protocols based on expurgated quadratic residue codes for qubit $T$ states. The length $5$, $23$, $29$, $47$, $53$, and $71$ codes have nonzero threshold. The other code lengths tested have a fixed point at the pure magic state but cannot distill them from noisy magic states.}
    \label{Figure:TStateDistillationAll}
\end{figure}

\begin{figure}[H]
    \centering
    \includegraphics[width=0.85\linewidth]{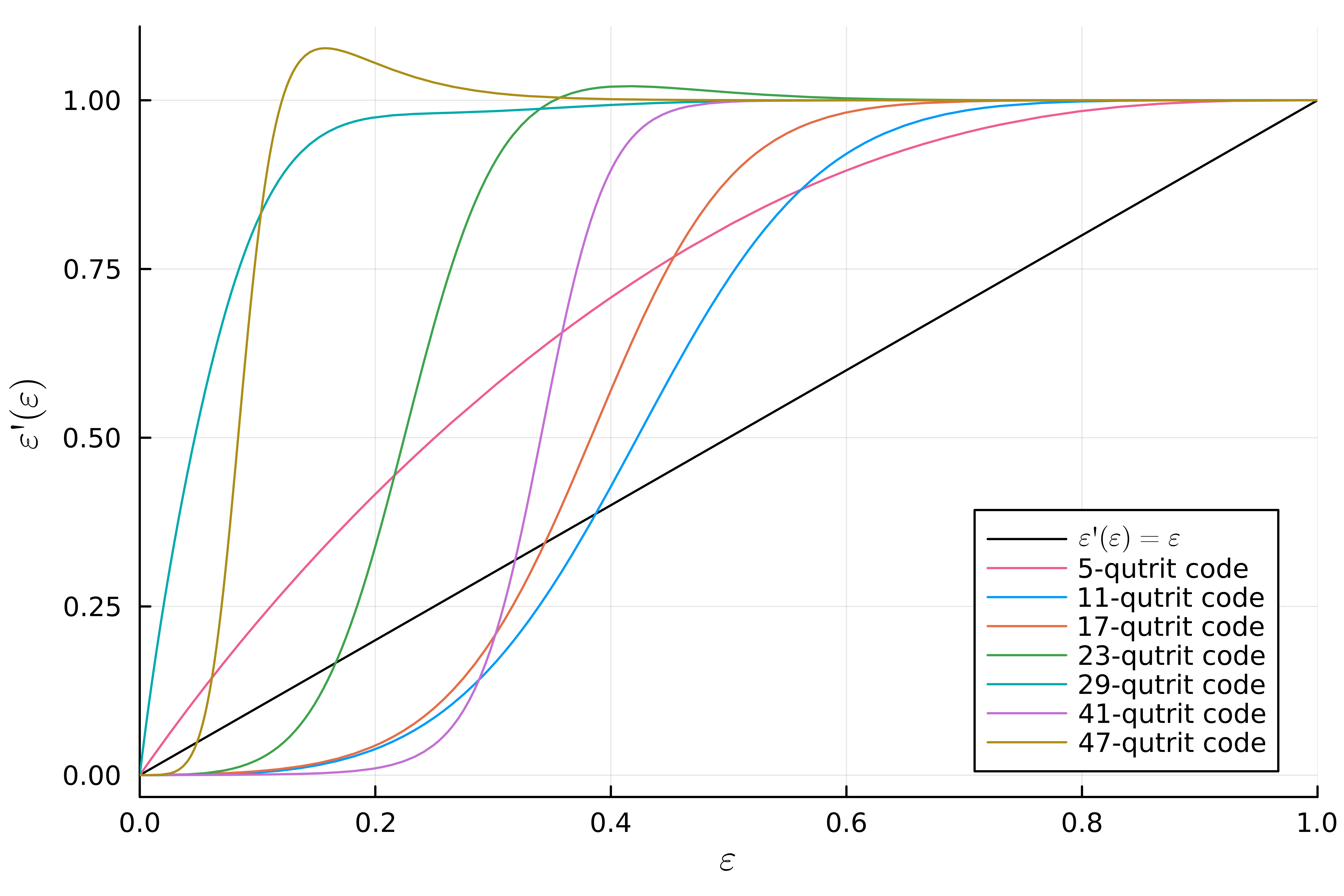}
    \caption{Post-distillation noise as a function of predistillation noise for magic state distillation protocols based on expurgated quadratic residue codes for qutrit Strange states. The length $11$, $17$, $23$, $41$, and $47$ codes have a nonzero threshold. The other code lengths tested have a fixed point at the pure magic state but cannot distill them from noisy magic states.}
    \label{Figure:StrangeStateDistillationAll}
\end{figure}

\section{Resource analysis for distillation with concatenated codes}\label{Section:ResourceAnalysis}

The main figure of merit in this paper is the noise threshold of magic state distillation protocols. However, as long as the input noise parameter $\epsilon_{\mathrm{in}}$ is below the threshold of a given protocol, the protocol can be applied repeatedly, allowing the noise parameter to be made arbitrarily small. If this condition is satisfied, then the practically relevant figure of merit becomes the resource consumption to achieve a given output noise parameter $\epsilon_{\mathrm{out}}$. There are many ways of measuring this, for example, the number of distillation rounds required, the number of input states required to reach a given noise parameter, and taking the success probability of the protocol into account, the expected number of raw input states required to achieve a given noise parameter.

Two properties of the quantum quadratic residue code family make them potentially favourable for this sort of resource usage comparison. First, a protocol formed by concatenating a fixed $n$-to-$1$ qubit protocol with itself is possible only at sizes $n^r$ for $r\in\mathbb{N}$. For the case of the concatenated $5$-qubit protocol, these protocols are available at sizes $5$, $25$, $125$, $625$, \dots. If a target output error lies between the errors achieved by the $r$-round and $r+1$-round concatenated protocols, the $r+1$-round protocol must be used, potentially requiring a large number of additional input states, as well as additional latency from running an additional round. On the other hand, qubit quantum quadratic residue codes that distill $T$ states exist for every prime length $p$ such that $p\equiv 5$ or $23\pmod{24}$ (see Theorem~\ref{Theorem:InfiniteTDistillation}). The result is a much denser distribution of distillation protocol sizes than those formed from concatenating small protocols, potentially saving resource usage. For example, between the length $25$ and $125$ concatenated $5$-qubit protocols, there are quantum QR code-based protocols of length $29$, $47$, $53$, $71$, and $101$.

Second, an $r$-round concatenated protocol requires $r$ sequential rounds of encoding, measurement, post-selection, decoding, and classical processing between rounds, whereas a single-round protocol can reach the desired output noise parameter with a single round of measurement. This paper assumes ideal stabilizer operations throughout, so this effect is not seen in the plots here, but on realistic near-term hardware, fewer rounds of distillation can itself be advantageous, especially on hardware with limited coherence times.

In this appendix, we compare the performance of the quantum quadratic residue codes for distilling qubit $T$ states against the baseline of protocols formed from concatenating the $5$-qubit protocol~\cite{BravyiKitaev2005}. Concatenation is modelled by iterating the distillation map $\epsilon'(\epsilon)$. That is, the initial noise parameter is $\epsilon_0=\epsilon_{\mathrm{in}}$, and the noise parameter after $r$ rounds is $\epsilon_r=\epsilon'(\epsilon_{r-1})$. A successful, fully post-selected run of the protocol will consume $n^r$ qudits.

Figure~\ref{Figure:LiveResourceUsage} plots the output noise parameter $\epsilon_{\mathrm{out}}$ against the number of input $T$ states used, assuming we post-select on the desired measurement outcomes. The protocols based on quantum QR codes fill some of the gaps between the concatenated $5$-qubit protocols. For example, at $\epsilon_{\mathrm{in}}=0.1$, a single round of the $23$-qubit protocol reaches an output noise parameter intermediate between one and two rounds of the $5$-qubit protocol, and a single round of the $53$-qubit protocol reaches output noise parameter intermediate between two and three rounds of the $5$-qubit protocol. That said, the concatenated $5$-qubit protocol achieves the fastest error suppression of the protocols tested, especially with $\epsilon_{\mathrm{in}}=0.01$.

\begin{figure}
    \centering
    \includegraphics[width=0.85\linewidth]{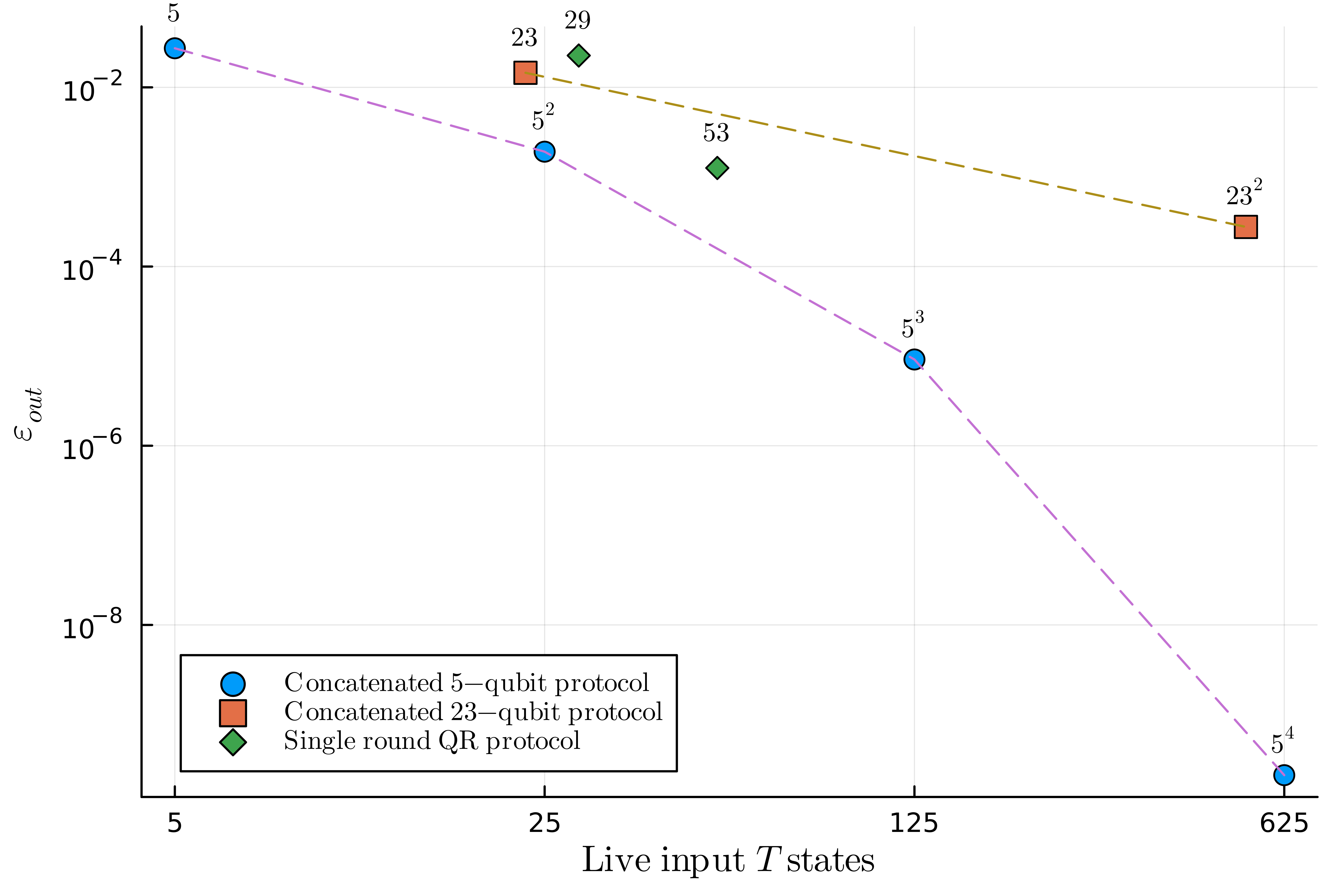}
    \includegraphics[width=0.85\linewidth]{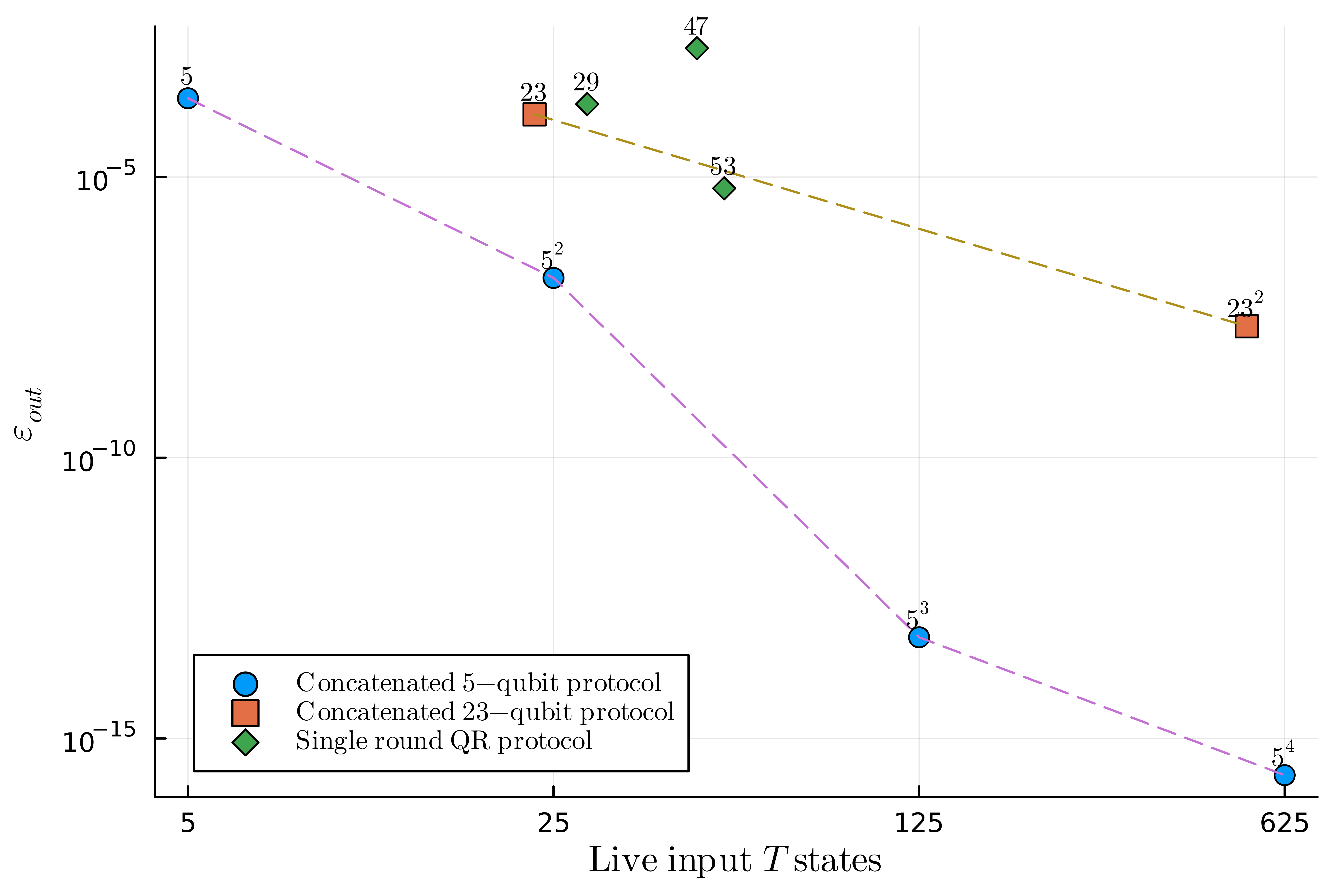}
    \caption{Output error $\epsilon_{\mathrm{out}}$ as a function of number of $T$ states used in concatenated magic state distillation protocols based on quantum quadratic residue codes, assuming post-selection on the desired measurement outcomes. The first plot assumes the raw input $T$ states have noise parameter $\epsilon_{\mathrm{in}}=0.1$, and the second plot assumes $\epsilon_{\mathrm{in}}=0.01$.}
    \label{Figure:LiveResourceUsage}
\end{figure}

Figure~\ref{Figure:RawResourceUsage} shows a similar plot, but without assuming post-selection on the desired measurement outcomes. Instead, it takes into account the success probability of the protocols, and plots the output noise parameter against the expected number of raw input $T$ states. For $r$ rounds of a $n$-to-$1$ qubit protocol, the expected number of $T$ states is
\begin{equation}
    N_r = \prod_{j=0}^{r-1}\frac{n}{P_s(\epsilon_j)},
\end{equation}
since producing one accepted output at round $j+1$ requires on average $n/P_s(\epsilon_j)$ states from round $j$. Both plots include a subplot that assumes the raw input $T$ states have initial noise parameter $\epsilon_{\mathrm{in}}=0.1$, and another subplot for $\epsilon_{\mathrm{in}}=0.01$.

Figure~\ref{Figure:RawResourceUsage} shows that the efficiency of larger protocols suffers when the success probability is taken into account since the success probability seems to fall off rapidly with the size of the protocol (see Figure~\ref{Figure:TStateSuccessProb} in the main text). The expected raw resource state consumption of these larger protocols is dominated by rejected runs, and the concatenated $5$-qubit protocols seem to consume far fewer raw resource states per accepted output state for a given target output noise parameter.

\begin{figure}
    \centering
    \includegraphics[width=0.85\linewidth]{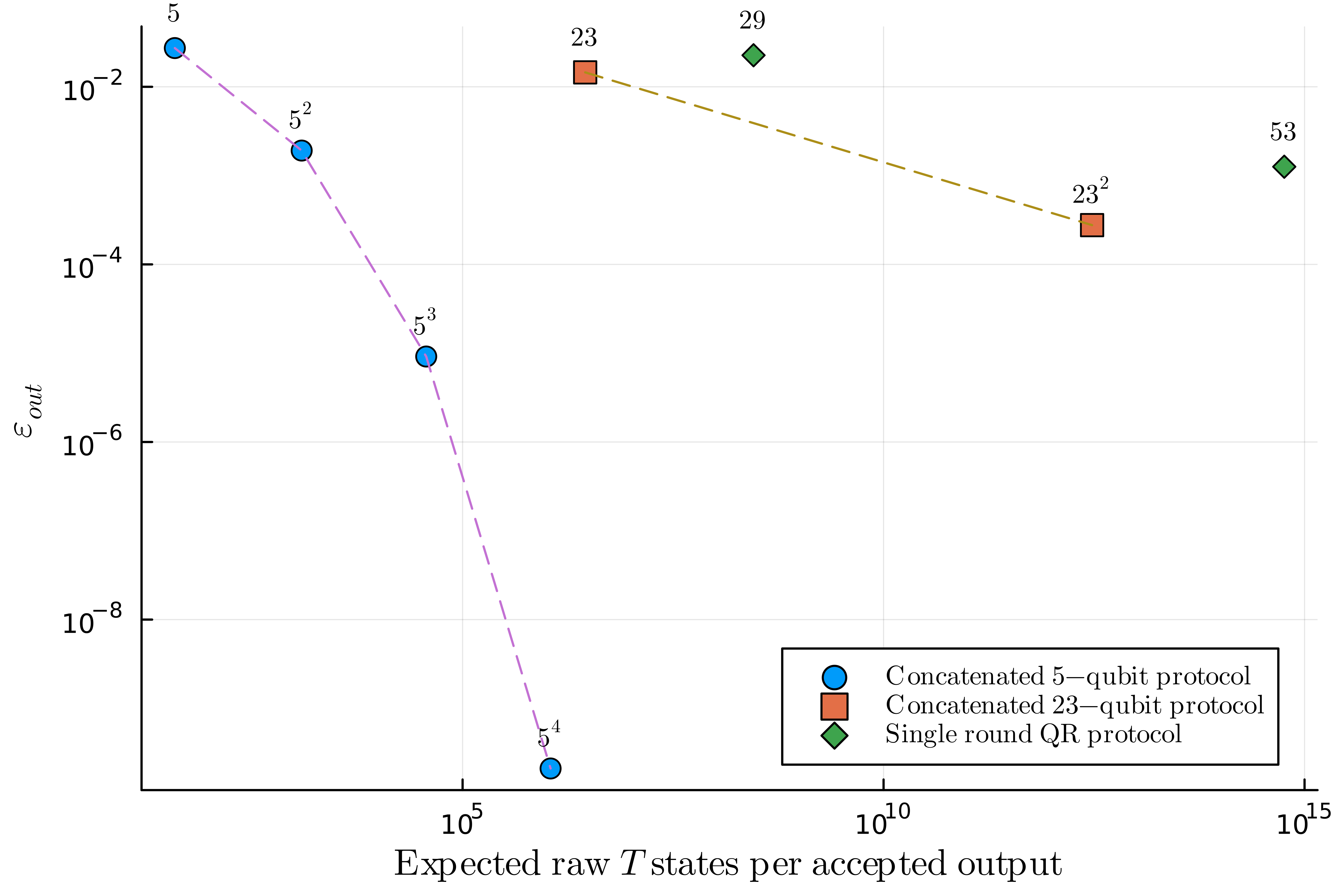}
    \includegraphics[width=0.85\linewidth]{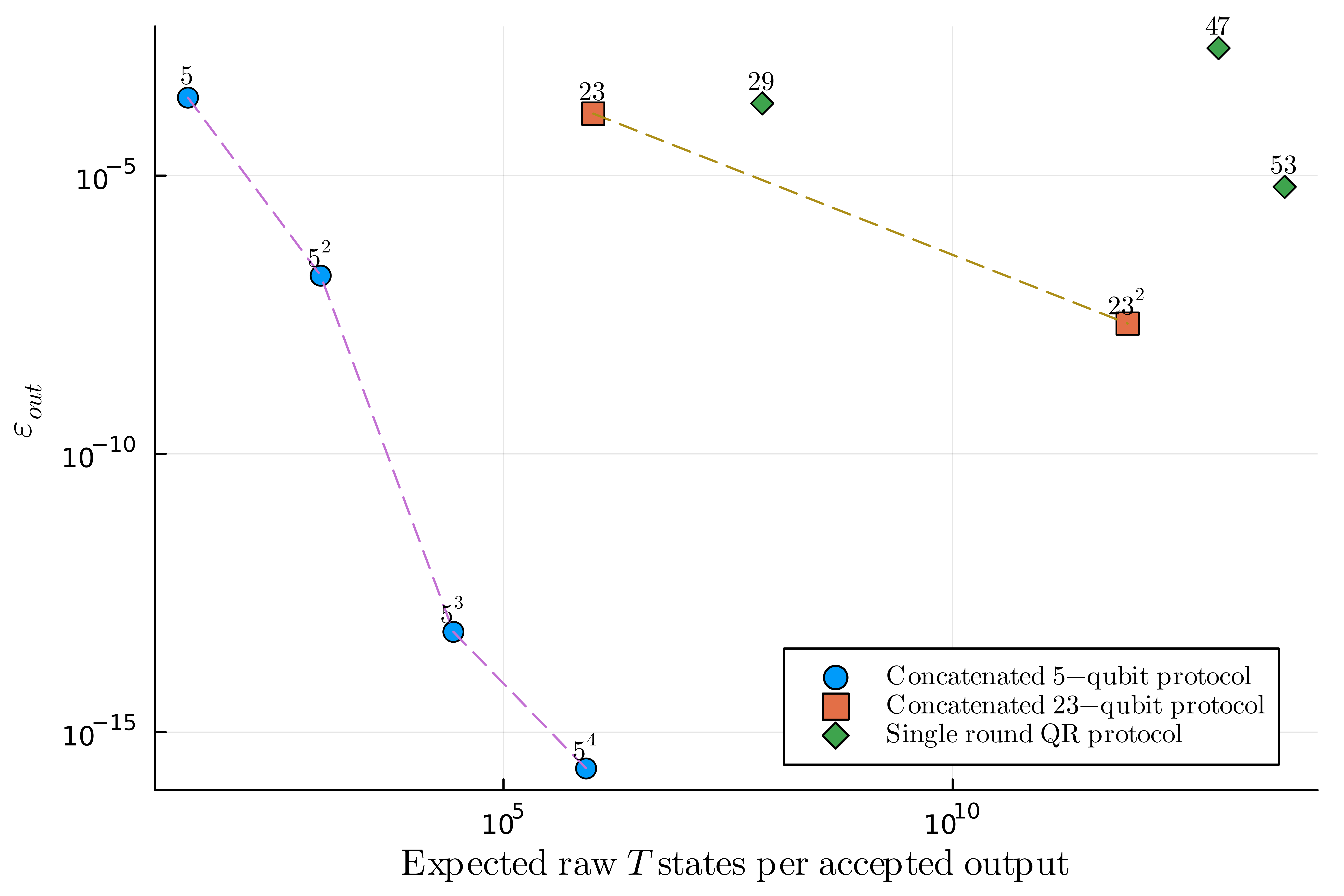}
    \caption{Output error parameter $\epsilon_{\mathrm{out}}$ as a function of the expected number of raw input $T$ states consumed by concatenated magic state distillation protocols based on quantum quadratic residue codes. The first plot assumes the raw input $T$ states have noise parameter $\epsilon_{\mathrm{in}}=0.1$, and the second plot assumes $\epsilon_{\mathrm{in}}=0.01$.}
    \label{Figure:RawResourceUsage}
\end{figure}

There are two limitations of this resource usage analysis to keep in mind. First, our numerical results are limited to codes of length $p\le71$ as a result of the cost of computing the weight enumerators of the relevant codes. Therefore, the above argument on the relative density of protocol sizes for single-round QR code protocols versus concatenated protocols could become more significant for larger qubit numbers.

Second, since the noise threshold is our main figure of merit for the bulk of the paper, we make a number of assumptions that simplify the analysis without impacting the threshold. These assumptions could, however, have a significant effect on the current resource usage estimates. For example, in the protocols defined in Section~\ref{Section:MSD}, we always post-select on a single syndrome measurement. This has a huge effect on the success probability. It is possible that allowing post-selection on other syndromes can boost the success probability of larger protocols without significant impact on other figures of merit like the noise threshold. We expect this to hold in the case of quantum QR codes, since Corollary~\ref{Corollary:SquareRootBound} shows that the distance of the codes grows with the square root of the length. A distance $\delta$ code can correct errors up to weight $\lfloor(\delta-1)/2\rfloor$, so a protocol in which low-weight nontrivial syndromes are corrected instead of discarded could significantly boost the success probability. The analysis of error-corrected protocols is more complicated, and we leave it to future work.

The framework presented in this appendix, and in this paper broadly, can continue to be applied as the weight enumerators of larger codes become practically computable. The above mentioned properties of quantum quadratic residue codes make them promising objects of study for improved magic state distillation protocols.

\end{document}